\documentclass{article}

\usepackage{misc}
\usepackage{amsmath,amsfonts}%amssymb
\usepackage{amsthm}
\usepackage{algorithmic}
\usepackage{graphicx}
\usepackage{textcomp}
\usepackage{xcolor}
\usepackage{xspace}
\usepackage{csquotes}
\usepackage{enumerate}
\usepackage{hyperref}
\usepackage{thm-restate}
\usepackage[all]{nowidow}
\usepackage[T1]{fontenc}

\newtheorem{theorem}{Theorem}
\newtheorem{example}[theorem]{Example}
\newtheorem{definition}[theorem]{Definition} 
\newtheorem{lemma}[theorem]{Lemma}
\newtheorem{corollary}[theorem]{Corollary}
\newtheorem{claim}[theorem]{Claim}

\usepackage{listings}
\begin{document}

\newcommand{\Gfull}{\ensuremath{\Gmc_{@,U}}\xspace}
\newcommand{\Gat}{\ensuremath{\Gmc_@}\xspace}
\newcommand{\Hfull}{\ensuremath{\Hmc_{@,U}}\xspace}
\renewcommand{\Hat}{\ensuremath{\Hmc_@}\xspace}
\newcommand{\Hu}{\ensuremath{\Hmc_U}\xspace}
\renewcommand{\ML}{\ensuremath{\mathsf{ML}}\xspace}
\newcommand{\GL}{\ensuremath{\mathsf{GHL}}\xspace}
\newcommand{\HL}{\ensuremath{\mathsf{HL}}\xspace}
\newcommand{\M}{\Mmc}
\newcommand{\N}{\Nmc}
\renewcommand{\phi}{\varphi}
\newcommand{\case}{\ensuremath{\mathsf{case}}\xspace}
\newcommand{\Case}{\ensuremath{\mathsf{Case}}\xspace}

\newcommand{\Var}{\ensuremath{\mn{Var}}\xspace}
\newcommand{\Rel}{\ensuremath{\mn{Rel}}\xspace}
\newcommand{\Nom}{\ensuremath{\mn{Nom}}\xspace}
\newcommand{\sig}{\ensuremath{\mn{sig}}\xspace}
\newcommand{\mydiam}[1]{\langle{#1}\rangle}
\newcommand{\mybox}[1]{[{#1}]}
\newcommand{\sep}{\ensuremath{\mathsf{sep}}\xspace}
\newcommand{\Sep}{\ensuremath{\mathsf{Sep}}\xspace}
\newcommand{\TP}{\ensuremath{\mathsf{Type}}\xspace}
\newcommand{\tp}{\ensuremath{\mathsf{tp}}\xspace}
\newcommand{\SepProf}{\text{$\mathsf{Sep}$-$\TP$}} %{\text{$\Gamma$-$\TP$}}
\newcommand{\SepCaseType}{\text{$\mathsf{Sep}_\cmf$-$\TP$}}
\newcommand{\OneSepProf}{\SepProf^\star}
\newcommand{\SepCaseStarType}{\SepCaseType^\star}
\newcommand{\mos}{\ensuremath\mathsf{m}\xspace}
\newcommand{\Mos}{\ensuremath\mathsf{Mos}\xspace}
\newcommand{\D}{\mathcal{D}}

% \conference{CI-BD-SOQE'26: Workshop on Craig Interpolation, Beth Definability, and Second-Order Quantifier Elimination}

%%

\title{Computation and Size of Interpolants for Hybrid Modal Logics}

\author{Jean Christoph Jung \\ TU Dortmund University \and Jędrzej Kołodziejski \\ University of Warsaw \and Frank Wolter \\ University of Liverpool}

\maketitle

\begin{abstract}
Recent research has established complexity results for the problem of deciding the existence of interpolants in logics lacking the Craig interpolation property (CIP). The proof techniques developed so far are non-constructive, and no meaningful bounds on the size of interpolants are known. Hybrid modal logics (or modal logics with nominals) are a particularly interesting class of logics without CIP: in their case, CIP cannot be restored without sacrificing decidability and, in applications, interpolants in these logics can serve as definite descriptions and separators between positive and negative data examples in description logic knowledge bases. In this contribution we show, using a new hypermosaic elimination technique, that in many standard hybrid modal logics Craig interpolants can be computed in fourfold exponential time, if they exist. On the other hand, we show that the existence of uniform interpolants is undecidable, which is in stark contrast to modal or intuitionistic logic where uniform interpolants always exist.  
\end{abstract}

\section{Introduction}

Consider formulae $\varphi$ and $\psi$ such that $\varphi\rightarrow \psi$ is valid. Then a \emph{Craig interpolant} for $\varphi$ and $\psi$ is any formula $\chi$ using only the shared non-logical symbols of $\varphi$ and $\psi$ such that both $\varphi\rightarrow \chi$ and $\chi \rightarrow \psi$ are valid. The formula $\chi$ often provides a useful explanation of why $\varphi\rightarrow \psi$ is valid. 
Craig interpolants have found many applications, including the inference of loop invariants and program abstractions in program verification~\cite{DBLP:reference/mc/McMillan18}, query rewriting in databases~\cite{DBLP:series/synthesis/2016Benedikt,DBLP:conf/kr/TomanW21}, and the construction of separators for positive and negative data examples in concept learning~\cite{DBLP:journals/ai/JungLPW22}. In all of these settings, it is crucial that interpolants can be computed efficiently, given $\varphi$ and $\psi$. 

A logic is said to enjoy the Craig interpolation property (CIP) if a Craig interpolant exists whenever $\varphi\rightarrow\psi$ is valid. In such cases, interpolants can often be extracted directly from proofs of $\varphi\rightarrow\psi$ in an appropriate proof system~\cite{DBLP:journals/apal/TabatabaiJ25a,DBLP:journals/apal/FittingK15,TenEtAl13}. By contrast, if a logic does not have CIP, then the existence of a proof does not guarantee the existence of an interpolant, and it remains unclear how interpolants could be extracted even when they do exist. Although recent work has made significant progress on the problem of deciding whether interpolants exist in logics without CIP, very little has been achieved on the construction of interpolants or on estimating their size. This problem  arises because the algorithms developed so far are non-constructive: rather than providing explicit interpolants, they rely on variants of Robinson’s joint consistency, which in turn rely on compactness.

The aim of this paper is to establish the first elementary upper bounds for constructing Craig interpolants in modal and description logics that lack CIP. We focus on the family of decidable hybrid modal logics, or equivalently, description logics with nominals. The computation of Craig interpolants in these logics is particularly significant for two reasons.

First, in hybrid modal logics, CIP cannot be restored merely by adding further logical connectives without sacrificing decidability~\cite{DBLP:journals/jsyml/Cate05}. Consequently, the natural strategy of repairing CIP by moving to a more expressive logic is effectively ruled out, at least for applications where effective reasoning is essential.

Second, when hybrid modal logics are viewed as description logics with nominals, Craig interpolants have numerous potential applications in knowledge representation and reasoning~\cite{DBLP:journals/corr/abs-2512-08833}. They can explain modularity violations between ontologies or knowledge bases~\cite{KonevLWW09,DBLP:conf/kr/KonevLPW10}, act as separators between positive and negative data examples in description logic knowledge bases~\cite{DBLP:journals/ai/JungLPW22,DBLP:journals/tocl/ArtaleJMOW23}, provide definitions of concept names~\cite{TenEtAl13}, enable alignments between ontologies~\cite{DBLP:conf/ekaw/GeletaPT16}, and supply definite descriptions (or referring expressions) for individuals in description logic knowledge bases~\cite{ArtEtAl21}.

We now present our main contributions. We study hybrid modal logic $\Hmc$, which extends modal logic with \emph{nominals} interpreted as singleton sets, together with its extensions $\Hat$ and $\Hfull$, incorporating the \emph{@-operator}, $@_{a}\varphi$, stating that $\varphi$ is true in the node denoted by nominal $a$, and the \emph{universal modality}, $\mybox{U}\varphi$, stating that $\varphi$ is true everywhere, respectively. In addition, we examine the further extensions of these logics with \emph{graded modalities} (also known as qualified number restrictions in description logics)~\cite{DBLP:journals/corr/abs-1910-00039,DBLP:journals/sLogica/Rijke00}. Let $\HL$ and $\GL$ denote the respective collections of these logics.

For logics in $\HL \cup \GL$, satisfiability is \PSpace-complete in the absence of the universal modality, and \ExpTime-complete when the universal modality is present~\cite{DBLP:conf/csl/ArecesBM99,DBLP:journals/jair/Tobies00,BMLSW-LTCS-05-02}. No logic in $\HL \cup \GL$ enjoys the CIP, as witnessed by the validity of
$(a \wedge \mydiam{R}a) \rightarrow (b \rightarrow \mydiam{R} b)$, with $a,b$ nominals,  for which no interpolant exists.

Interpolant existence has previously been investigated~\cite{DBLP:journals/tocl/ArtaleJMOW23}, where it was shown to be co\NExpTime-complete for $\Hmc$ and 2$\ExpTime$-complete for $\Hfull$. Building on this, one can easily show that for all logics in $\HL\cup \GL$ interpolant existence is co\NExpTime-complete in the absence of the universal modality and 2\ExpTime-complete otherwise. Our interest here is not in deciding the existence of interpolants but in understanding their complexity. We state our main result which gives an elementary upper bound on the size of Craig interpolants.
\begin{theorem}\label{thm:mainthm}
    Let $\Lmc\in \HL\cup \GL$. Given $\varphi,\psi$ in $\Lmc$ such that a Craig interpolant for $\varphi \rightarrow\psi$ exists, one can construct such an interpolant in fourfold exponential time (and of fourfold exponential size) in the size of $\varphi$ and $\psi$.
\end{theorem}
Before turning to the proof, we note that a standard strategy for bounding the size of interpolants for $\varphi \rightarrow \psi$, namely, bounding the size of the logically strongest consequence of $\varphi$ over the symbols shared by $\varphi$ and $\psi$,
is not applicable here, as there is no elementary upper bound on the size of this consequence. 
\begin{example}\label{exm:non-elem}
Let $n\geq 2$ and $\sigma=\{R,S,p\}$ with $R$ and $S$ modal accessibility relations and $p$ a propositional variable. Let   
\[\varphi_n = \mybox{R} a\wedge \mybox{S} a\wedge \bigwedge_{1\leq i\leq n-1}\mybox{R}^{i}\mybox{S}\bot\wedge\mybox{R}^n\bot \wedge \mybox{S}^{2}\bot\]
with $a$ a nominal. Then, up to logical equivalence, there are only finitely many formulae in $\Hmc$ over $\sigma$ that are consistent with $\varphi_n$. Their disjunction, $\chi_{n}$, is the logically strongest consequence 
of $\varphi_{n}$ in $\Hmc$ using symbols in $\sigma$ only. Hence, if there is an interpolant between $\varphi_{n}$ and some formula $\psi$ with shared symbols in $\sigma$, then $\chi_{n}$ is such an interpolant. However, using the fact that $\varphi_n\to (\mydiam{R}\chi\to \mydiam{S}\chi)$ is valid for all formulae $\chi$ one can 
easily show that there is no elementary upper bound on the size of formulae in $\Hmc$ equivalent to $\chi_{n}$. 
\qed
\end{example}

Instead, our proof of Theorem~\ref{thm:mainthm} builds on the following two ideas. Given formulae $\varphi$ and $\psi$ sharing a set of non-logical symbols $\sigma$:
\begin{description}
    \item[(bisim)] for modal and guarded logics, the non-existence of an interpolant for $\varphi \rightarrow \psi$ is often equivalent to the existence of $\sigma$-bisimilar models satisfying $\varphi$ and $\neg\psi$~\cite{DBLP:conf/amast/Marx98,DBLP:journals/corr/abs-2508-12805};
    \item[(elim)] for modal and guarded logics with CIP, interpolants can sometimes be constructed using the standard type elimination procedure, a technique commonly employed to establish complexity upper bounds for satisfiability~\cite{DBLP:journals/tocl/BenediktCB16,bezhanishvili2025six}.
    The procedure eliminates pairs of types that cannot be jointly satisfied, and inductively computes interpolants between them.
\end{description}
% \begin{itemize}
%     \item[(bisim)] for modal and guarded logics, the non-existence of an interpolant for $\varphi \rightarrow \psi$ can be shown by constructing $\sigma$-bisimilar models satisfying $\varphi$ and $\neg\psi$~\cite{DBLP:conf/amast/Marx98,DBLP:journals/corr/abs-2508-12805}; and
%     \item[(elim)] for modal and guarded logics with CIP, interpolants can be constructed using the type elimination procedure, a technique commonly employed to establish complexity upper bounds for satisfiability~\cite{DBLP:journals/tocl/BenediktCB16,bezhanishvili2025six}. Intuitively, for each eliminated pair of types, an interpolating formula is inductively constructed when running the elimination procedure. {\color{red} I guess it's alright, but one could even be more honest and more closely connected to (bisim) and say that one can adapt type elimination to decide existence of bisimilar models, and from this procedure construct interpolants? Or is it too complicated?}
% \end{itemize}
We illustrate (elim) in a warm-up section by showing that one can always construct at most exponential size interpolants for $\varphi\rightarrow\psi$ in $\Hmc$ \emph{if} all nominals in $\varphi,\psi$ are shared. 
To apply the elimination procedure to logics without the CIP, it must be lifted to \emph{mosaics}, pairs of sets of types rather than pairs of individual types. In this form, the method (elim) has already been successfully employed to obtain tight complexity bounds for the existence of interpolants in a variety of logics lacking the CIP~\cite{DBLP:conf/lics/JungW21,DBLP:journals/tocl/ArtaleJMOW23,DBLP:journals/lmcs/KuruczWZ25}. The underlying idea is that sets of types in a mosaic represent conceivable types of points from a bisimilarity class, which may contain more than two elements.\footnote{We note an unfortunate clash of terminology: in~\cite{DBLP:journals/tocl/BenediktCB16}, mosaics are a generalization of types tailored to guarded logics, and are always realized within a single model. In~\cite{DBLP:conf/lics/JungW21,DBLP:journals/tocl/ArtaleJMOW23} and the present paper, mosaics represent sets of types realized by bisimilar nodes.}

The transition from pairs of types to mosaics, however, comes at a cost. Making mosaic elimination constructive is substantially more difficult, and has so far been achieved only in very restricted settings~\cite{DBLP:conf/dl/interpolantshard}. We argue that for hybrid modal logics, mosaic elimination is actually too local to compute interpolants at each elimination step. To capture global interactions between mosaics, we extend the elimination method further by eliminating unrealizable \emph{sets} of mosaics, called hypermosaics. 

It remains open whether the fourfold exponential upper bound is tight. However, we show a triply exponential (doubly exponential for DAG representation) lower bound on the size of Craig interpolants for languages in $\HL\cup \GL$ with universal modality. Thus, interpolants are at least one exponential larger than in modal logic. 

In the final section of this article, we investigate uniform interpolants in hybrid modal logics. In description logic, uniform interpolants have been extensively investigated and applied to forget symbols in ontologies or knowledge bases~\cite{FAME,LETHE}. Here we further demonstrate the complexity of constructing interpolants in hybrid modal logics by 
proving that the existence of uniform interpolants for a formula and set of non-logical symbols is undecidable. We show this for a broader class of logics, which in particular subsumes $\HL \cup \GL$.
\begin{theorem} Let $\Lmc$ be a logic containing $\Hmc$ and contained in ${\sf FO}$.
The existence of a uniform $\Lmc(\sigma)$-interpolant for given $\varphi\in\Lmc$ and signature $\sigma$ is undecidable.
\end{theorem}
This result is in stark contrast to modal and intuitionistic logic, where uniform interpolants always exist~\cite{DBLP:journals/jsyml/Pitts92,visser1996uniform} and can, at least for modal logic, even be constructed in exponential time~\cite{TenEtAl06}.

% {\bf Related Work}. 
\subparagraph{Related Work.}We discuss related work not yet mentioned above. An overview of what is currently known about the size of (uniform) interpolants in modal logic is given in~\cite{DBLP:journals/corr/abs-2511-04577}. The CIP has been investigated extensively for hybrid modal logics and corresponding description logics with nominals. The main positive results are that by adding the binding operator $\downarrow$ to $\Hmc(@)$ one obtains a logic with CIP (also if restricted to popular classes of frames)~\cite{DBLP:journals/jsyml/BlackburnM03}, and that CIP holds for languages in {\sf HL} if arbitrary nominals can be used in interpolants (called CIP for propositional variables)~\cite{DBLP:journals/logcom/BezhanishviliC06,DBLP:journals/jsyml/ArecesBM01}. However, satisfiability, and hence also interpolant existence, are undecidable for $\Hmc(\downarrow,@)$. This makes its use in knowledge representation applications problematic. Likewise, allowing arbitrary nominals in interpolants undermines the very purpose of interpolants in description logic. The decision problem of whether an interpolant exists for a given implication in hybrid modal logics lacking CIP is studied in~\cite{DBLP:journals/tocl/ArtaleJMOW23}. The principal question left open there—how to effectively construct interpolants when they do exist, and what bounds can be established on their size—is the focus of the present contribution.

\section{Preliminaries}
We introduce hybrid logics $\Hmc,\Hat,\Hfull$ following~\cite{areces200714}. They are based on countably infinite sets $\Rel$, $\Var$, and $\Nom$ of binary \emph{relation symbols}, \emph{propositional variables}, and \emph{nominals}, respectively. We use the symbol~$U$ to denote the \emph{universal modality}. Formulae in the logic $\Hfull$ are defined by the rule
\[\varphi,\psi ::= \top\mid p \mid a\mid \neg \varphi\mid \varphi\wedge\psi\mid \mydiam{R}\varphi\mid @_a\varphi\]
where $p\in\Var$, $R\in\Rel\cup\{U\}$, and $a\in \Nom$. We use the standard abbreviations $\varphi\vee \psi$, $\varphi\to\psi$, $\mybox{R}\varphi$, etc. We denote with $\Hat$ the fragment of \Hfull in which the constructor $\mydiam{U}\varphi$ is dropped and with $\Hmc$ the fragment
of $\Hat$ in which also $@_a\varphi$ is dropped. In $\Hfull$ we regard $@_{a}\varphi$ as an abbreviation of $\mydiam{U}(a \wedge \varphi)$. We refer with $\HL$ to the set of all logics defined so far. 

The semantics is given by \emph{models} $\M=(W^\M,\cdot^\M)$ where $W^\M$ is the \emph{domain} of \M, and $\cdot^\M$ is an \emph{interpretation function} that maps every $R\in \Rel\cup\{U\}$ to a binary relation $R^\M\subseteq W^\M\times W^\M$ such that $U^\M= W^\M\times W^\M$,
every variable $p\in\Var$ to a subset $p^\M\subseteq W^\M$, and every $a\in\Nom$ to an element $a^\M\in W$. A \emph{pointed model} is a pair $\M,w$ with $w\in W^\M$. When the model $\M$ is clear from the context we will skip the superscripts and write $W$, $R$ etc.
The satisfaction relation between pointed models and formulae from \Hfull is then defined as usual:
% \begin{align*}
% %
% \M,w & \models \top && \text{always} &
% %
% \M,w & \models p && \text{if }w \in p^\M \\
% %
% \M,w & \models a && \text{if }w=a^\M &
% %
% \M,w & \models \neg \varphi && \text{if }\M,w\not\models \varphi \\
% %
% \M,w & \models \varphi\wedge\psi && \text{if }\M,w\models\varphi\text{ and }\M,w\models\psi &
% %
% \M,w & \models \mydiam{R}\varphi && \text{if }\M,v\models\varphi\text{ for some $(w,v)\in R^\M$} \\
% %
% \M,w & \models @_a\varphi && \text{if }\M,a^\M\models \varphi
% %
% % \end{align*}
\begin{align*}
\M,w & \models \top && \text{always}; &
\M,w & \models \varphi\wedge\psi && \text{if }\M,w\models\varphi\text{ and }\M,w\models\psi;\\
\M,w & \models p && \text{if }w \in p^\M; &
\M,w & \models \mydiam{R}\varphi && \text{if }\M,v\models\varphi\text{ for some $(w,v)\in R^\M$}; \\
\M,w & \models a && \text{if }w=a^\M; &
\M,w & \models @_a\varphi && \text{if }\M,a^\M\models \varphi;\\
\M,w & \models \neg \varphi && \text{if }\M,w\not\models \varphi.
\end{align*}

% The semantics is given by (hybrid) \emph{models} $\M=(W,(R^\M)_{R\in\Rel}, V)$ where $W$ is a set of \emph{worlds}, for all $R\in \Rel$, $R^\M$ is a binary relation over $W$ with $U^\M=W\times W$, and $V$ is a map that assigns a set $V(p)\subseteq W$ to every proposition $p\in\Var$ and an element $V(a)\in W$ to every $a\in\Nom$. The satisfaction relation is then defined as usual, for each $w\in W$:
% %
% \begin{align*}
% %
% \M,w & \models \top \\
% %
% \M,w & \models p && \text{if }w \in V(p) \\
% %
% \M,w & \models a && \text{if }w=V(a) \\
% %
% \M,w & \models \neg \varphi && \text{if }\M,w\not\models \varphi \\
% %
% \M,w & \models \varphi\wedge\psi && \text{if }\M,w\models\varphi\text{ and }\M,w\models\psi \\
% %
% \M,w & \models \mydiam{R}\varphi && \text{if }\M,v\models\varphi\text{ for some $(w,v)\in R^\M$} \\
% %
% \M,w & \models @_a\varphi && \text{if }\M,V(a)\models \varphi
% %
% \end{align*}
%
% If $\M,w\models \varphi$, then $\M,w$ is called \emph{model of $\varphi$}\textcolor{blue}{model already used}.

We write $\varphi\models \psi$ if $\Mmc,w\models \psi$ whenever $\Mmc,w\models \varphi$, for all pointed models $\Mmc,w$. A formula $\varphi$ is \emph{valid}, written $\models \varphi$, if $\top\models\varphi$; it is \emph{satisfiable} if there is a pointed model $\M,w$ with $\M,w\models \varphi$.

The \emph{size} of a formula $\varphi$, denoted by $|\varphi|$, is the number of nodes in the syntax tree underlying $\varphi$. In addition, we also consider the \emph{DAG-size} of a formula, defined as the number of distinct subformulae.

A \emph{signature} is a finite subset $\sigma$ of $\Rel\cup\Var\cup\Nom$. The signature $\sig(\varphi)$ of a formula $\varphi$ is the set of all relation symbols, propositional variables, and nominals used in $\varphi$. Let $\Lmc\in \HL$ and $\sigma$ be a signature. We denote by $\Lmc(\sigma)$ the set of all formulae $\varphi\in\Lmc$ with $\sig(\varphi)\subseteq \sigma$. A \emph{Craig interpolant} for $\varphi_{1},\varphi_{2}\in\Lmc$ is any formula $\theta\in\Lmc$ in the shared signature $\sigma=\sig(\varphi_{1})\cap\sig(\varphi_{2})$ of $\varphi_{1},\varphi_{2}$
such that $\models \varphi_{1}\to\theta$ and $\models \theta\to \varphi_{2}$. We say that $\Lmc$ has \emph{Craig interpolation (CIP)} if
for every $\varphi_{1},\varphi_{2}$ with $\models \varphi_{1} \rightarrow \varphi_{2}$ there exists a Craig interpolant in $\Lmc$ for $\varphi_{1},\varphi_{2}$.
\begin{theorem}[Theorem~4.8 in~\cite{DBLP:journals/jsyml/ArecesBM01}]
    No $\Lmc\in \HL$ has CIP.
\end{theorem}
Moreover, it is known that, under a natural definition of \emph{logic}, 
no decidable logic extending $\Hmc$ and contained in ${\sf FO}$ has CIP~\cite{DBLP:journals/jsyml/Cate05}. 
Another possible way to restore CIP is 
to only restrict the use of propositional variables. Say that $\Lmc\in \HL$ has \emph{CIP for propositional variables} if for
every $\varphi_{1},\varphi_{2}\in \Lmc$ with $\models \varphi_{1}\rightarrow \varphi_{2}$ there is a formula $\theta\in\Lmc(\sigma)$ using only the shared propositional variables of $\varphi_{1},\varphi_{2}$ (but possibly nominals occurring only in $\varphi_{1}$ or $\varphi_{2}$) such that $\models \varphi_{1}\to\theta$ and $\models \theta\to \varphi_{2}$.
\begin{theorem}% [\cite{DBLP:journals/logcom/BezhanishviliC06,DBLP:journals/jsyml/ArecesBM01}]
[Chapter~6 in \cite{tenCatediss}]
\label{thm:propletter}
All $\Lmc\in \HL$ have CIP for propositional variables.
%\nb{J: previous references, still in the comments, were not the right ones}
\end{theorem}
On a technical level, we will deal with separators instead of interpolants. An \emph{$\Lmc(\sigma)$-separator} for $\varphi_{1},\varphi_{2}$ is a formula $\theta\in\Lmc(\sigma)$ such that $\varphi_1\models \theta$ and $\varphi_{2}\wedge\theta\models \bot$. Separators and Craig interpolants are easily polynomially reduced to each other:  clearly, a Craig-interpolant for $\varphi_{1},\varphi_{2}$ in $\Lmc$ is obtained by taking an $\Lmc(\sigma)$-separator for $\varphi_{1},\neg\varphi_{2}$ in the shared signature $\sigma$ of $\varphi_{1},\varphi_{2}$. Conversely, one can obtain an $\Lmc(\sigma)$-separator for $\varphi_{1},\varphi_{2}$ by taking any Craig interpolant for $\varphi_{1}',\neg\varphi_{2}'$ with $\varphi_{1}',\varphi_{2}'$ 
obtained from $\varphi_{1},\varphi_{2}$ by replacing symbols not in $\sigma$ by fresh mutually disjoint symbols.
Hence, we can focus on the following computation problem for every $\Lmc\in\HL$: 
\begin{description}
    \item[Input:] $\varphi_{1},\varphi_{2}\in\Lmc$, signature $\sigma$
    \item[Output:] $\Lmc(\sigma)$-separator for $\varphi_{1},\varphi_{2}$ if one exists, and ``no'' if no such separator exists. 
\end{description}
Separator existence can be model-theoretically characterized using bisimulations, which we recall next. 
\begin{definition}\label{def:bisi}
Let $\M,\N$ be models, and $\sigma$ a signature. An \emph{$\Hmc(\sigma)$-bisimulation} between \M and \N 
is a relation $Z\subseteq \dom(\M)\times \dom(\N)$ such that 
every $(w,w')\in Z$ satisfies: 
\begin{description}%[\IEEEsetlabelwidth{abc}]

    \item[(atom)] %for every $(w,w')\in Z$, 
    $w\in p^\M$ iff $w'\in p^\N$, for all $p\in \sigma \cap \Var$; % and $V(a)=V'(a)$, for all $x\in\sigma\cap\Nom$;
    
    \item[(nom)] %for every $(w,w')\in Z$, 
    $a^\M=w$ iff $a^\N=w'$, for all $a\in\sigma\cap\Nom$;

    \item[(forth)] % for every $(w,w')\in Z$ and 
     for every $(w,v)\in R^\M$, there is $(w',v')\in R^\N$ with $(v,v')\in Z$, for all $R\in \sigma \cap \Rel$;
    
    \item[(back)] % for every $(w,w')\in Z$ and 
    for every $(w',v')\in R^\N$, there is $(w,v)\in R^\M$ with $(v,v')\in Z$, for all $R\in \sigma \cap \Rel$.
     
\end{description}    
An $\Hat(\sigma)$-bisimulation is an $\Hmc(\sigma)$-bisimulation $Z$ that additionally satisfies 
\begin{description}
    \item[(@)] $(a^\M,a^\N)\in Z$, for all $a\in \sigma\cap \Nom$.
\end{description}
Finally, an $\Hu(\sigma)$-bisimulation is an $\Hmc(\sigma)$-bisimulation for which \emph{(back)} and \emph{(forth)} also hold for the universal modality~$U$.

For $\Lmc\in \HL$, we write $\M,w\sim_{\Lmc,\sigma} \N,w'$ if there is an $\Lmc(\sigma)$-bisimulation $Z$ between \M and \N such that $(w,w')\in Z$. 
\end{definition}
Note that an $\Lmc(\sigma)$-bisimulation is surjective and inverse surjective whenever $U\in \sigma$. 
We call formulae $\varphi_{1},\varphi_{2}\in\Lmc$ \emph{jointly $\Lmc(\sigma)$-consistent} if there are pointed models $\M_{i},w_{i}$, $i\in \{1,2\}$, with $\M_{i},w_{i}\models\varphi_{i}$ for all $i\in \{1,2\}$, and $\M_{1},w_{1}\sim_{\Lmc,\sigma} \M_{2},w_{2}$.
We have the following (standard) characterization of interpolant existence, which can be established along the lines of~\cite[Theorem~5.6]{DBLP:journals/tocl/ArtaleJMOW23}.

\begin{lemma}\label{lem:joint-consistent}
Let $\Lmc\in\HL$. Then, for every $\varphi_1,\varphi_2\in \Lmc$ and every signature $\sigma$, the following are equivalent:
  \begin{enumerate}
  
      \item there is an $\Lmc(\sigma)$-separator for $\varphi_1,\varphi_2$;
      
      \item $\varphi_1,\varphi_2$ are not jointly $\Lmc(\sigma)$-consistent.
  \end{enumerate}
\end{lemma}

Hence, an $\Lmc(\sigma)$-separator for $\varphi_1,\varphi_2$ can be viewed as a certificate for the lack of joint $\Lmc(\sigma)$-consistency of $\varphi_1,\varphi_2$. This characterization has been used to decide the existence of interpolants in various logics. Unfortunately, its proof is \emph{not} constructive: indeed, $(2) \Rightarrow (1)$ is proved using compactness and 
does not construct a witnessing separator. 
\section{Warm-Up: First-Order Separators}
\label{sec:first-order}
Hybrid logics can be regarded as fragments of first-order logic (FO) using the \emph{standard translation} $T_{x}(\varphi)$ that maps every formula $\varphi$ in $\Hmc_{@,U}$ to an FO-formula with free variable $x$ such that $\mathcal{M},w\models \varphi$ iff $\mathcal{M}\models T_{x}(\varphi)(w)$.
By CIP of FO and since $\Hmc_{@,U}$ is decidable in \ExpTime, one can decide the existence of FO$(\sigma)$-separators for formulae $T_{x}(\varphi),T_{x}(\psi)$ in exponential time. In this section, we show that FO$(\sigma)$-separators can also be \emph{computed} in exponential time (note that, in general, there is no recursive bound on the length of FO-separators of FO-formulae~\cite{FRIEDMAN197618}). To this end (and because it is of independent interest), we first show the following constructive version of Theorem~\ref{thm:propletter}. 
\begin{restatable}{theorem}{thmnom}
\label{thm:nom}
Let $\Lmc\in \HL$. There is an exponential time algorithm deciding
for any $\varphi_{1},\varphi_{2}\in \Lmc$ and any $\sigma$ containing all nominals in 
$\sig(\varphi_{1})\cup \sig(\varphi_{2})$ 
whether an $\Lmc(\sigma)$-separator exists and, if so, constructs one of exponential DAG-size.
\end{restatable}
%\emph{Remark.} It is known already that the existence of $\Lmc(\sigma)$-separators can be decided in %exponential time because all $\Lmc\in \HL$ are known to have CIP for propositional letters (i.e., when non-%shared nominals are allowed in the %interpolant)~\cite{DBLP:journals/logcom/BezhanishviliC06,DBLP:journals/jsyml/ArecesBM01}). CIP for %propositional letters also follows from the proof below.

We sketch the proof of Theorem~\ref{thm:nom} for $\Lmc=\Hmc$ as it illustrates in an elementary setting the type elimination approach to interpolant construction which we generalize later to prove our main results. 
If $\Psi$ is a set of formulae then a \emph{$\Psi$-type} is a maximal satisfiable subset $t$ of formulae from $\Psi$ and negations thereof. That is, $t=\tp_\M^\Psi(w)$ for some $\M$ and $w\in W$ where
\[
\tp_\M^\Psi(w)=\{\phi \mid  \text{$\phi\in\{\psi,\neg\psi\}$ for $\psi\in\Psi$ and } \M,w\models\phi\}.
\]
If $\tp_{\M}^\Psi(w)=t$, then we say that $\M,w$ \emph{realizes} the type $t$. The set of all $\Psi$-types is denoted $\Psi$-$\TP$. Fix formulae $\varphi_{1},\varphi_{2}$ and let $\sub(\varphi_1,\varphi_2)$ denote the set of subformulae of $\varphi_{1}$ or $\varphi_{2}$. When no $\Psi$ is explicitly specified, by a \emph{type} we mean a $\sub(\varphi_{1},\varphi_{2})$-type, and we write $\tp_\M(w)$ and $\TP$ with no superscript. We sometimes use $t$ to denote the conjunction $\bigwedge_{\chi\in t}\chi$ of all formulae in type $t$.

We come to the proof of Theorem~\ref{thm:nom}.
Fix $\varphi_{1},\varphi_{2}\in \Hmc$ and signature $\sigma$ containing all nominals occurring in $\varphi_{1}$ or $\varphi_{2}$. 
%We show that if $\varphi_{1},\varphi_{2}$ are $\mathcal{H}(\sigma)$-inconsistent, then one can construct %an $\mathcal{H}(\sigma)$-separator in exponential time (and of at most exponential dag size). 
Let $\Nom_{0}=\Nom\cap\sub(\phi_1,\phi_2)$.

A \emph{singleton mosaic} $\mos=(\mos_{1},\mos_{2})$ consists of types $\mos_{1}$ and $\mos_{2}$.
%We generally assume that $T_{i}=\{t_{i}\}$ if $T_{i}\not=\emptyset$, for $i=,1,2$.
Following Lemma~\ref{lem:joint-consistent}, we aim to check whether the types in a singleton mosaic can be satisfied
in $\Hmc(\sigma)$-bisimilar nodes and to construct $\Hmc(\sigma)$-separators if this is not the case.
We will later generalize this idea to sets of non-singleton mosaics.
A basic case in which a singleton mosaic cannot be realized in this way is if its types do not coincide on propositional variables or nominals in $\sigma$. We say that $\mos$ is \emph{$\sigma$-inconsistent} if there is a propositional variable or nominal $r \in\sigma$ with $r\in \mos_{1}\setminus \mos_{2}$ or $r\in \mos_{2}\setminus \mos_{1}$.

We also make a case distinction that we will generalize later.
A \emph{basic nominal case} is a pair $\cmf=(\cmf_1,\cmf_2)$ of sequences of types,
one for each nominal:
\[
\cmf_1,\cmf_2:\Nom_0\to \TP.
\]
%Then $\cmf_{i}(a)=t$ expresses that $a$ satisfies type $t$ on side $i$. 
As with mosaics, for $i\in \{1,2\}$ we refer by $\cmf_i$ to the $i$-th component of a given basic nominal case $\cmf$. We use $\case(\cmf_i)$ to abbreviate the formula:
\[
\case(\cmf_i) = \bigwedge_{a\in\Nom_0} @_a \cmf_{i}(a)
\]
which expresses that all nominals $a$ satisfy the respective type $\cmf_{i}(a)$.
Note that $\case(\cmf_{i})$ is not in \Hmc, but this is not essential here. Also note that, for convenience, we define cases syntactically and allow for $\cmf$ with unsatisfiable $\case(\cmf_i)$, such as when $\cmf_1(a)\models a\wedge b\wedge p$ but $\cmf_1(b)\models a\wedge b\wedge\neg p$.
% For convenience, we define cases syntactically and allow them to be unsatisfiable.

Types $t,t'$ are called \emph{$R$-coherent}, in symbols $t \rightsquigarrow_{R} t'$, if $\chi\in t'$ and $\mydiam{R}\chi \in \sub(\varphi_1,\varphi_2)$ imply $\mydiam{R}\chi\in t$.
Singleton mosaics $\mos$ and $\mos'$ are \emph{$R$-coherent},
in symbols $\mos\rightsquigarrow_{R} \mos'$, if $\mos_{i}\rightsquigarrow_{R} \mos_{i}'$ for all $i\in \{1,2\}$.
Let $\mos$ be a singleton mosaic and $i\in \{1,2\}$. We say that a singleton mosaic $\mos'$ is a \emph{witness} for $\mydiam{R}\chi\in \mos_{i}$ if $\chi\in \mos_{i}'$ and $\mos \rightsquigarrow_{R} \mos'$.

%\item[(W$_{\bar{\sigma}}$)] if $R\not\in\sigma$, then $\chi\in t_{i}'$ and $t_{i}\rightsquigarrow_{R} t_{i}'$.
%\end{description}
Now fix a basic nominal case $\cmf=(\cmf_{1},\cmf_{2})$.

\medskip
\noindent
\emph{Elimination procedure for $\cmf$}. Let $\mathcal{Y}$ be the set of all singleton mosaics. Obtain $\mathcal{Y}_{0}$ from $\mathcal{Y}$ by dropping 
\begin{itemize}
    \item all $\mos$ such that $\mos_{i}\wedge \case(\cmf_{i})$ is not satisfiable for some $i\in \{1,2\}$, and
%    \item all $\mos$ such that $a\in \mos_{i}$ and $\mos_{i}\not=t_{a}^{i}$ for some $a$ and 
%    $i\in \{1,2\}$. 
    \item all $\sigma$-inconsistent $\mos$.
\end{itemize}
Next obtain by induction $\mathcal{Y}_{\ell+1}$ from $\mathcal{Y}_{\ell}$ by dropping all $\mos$ 
such that there is $\mydiam{R}\chi\in \mos_{i}$ with $R\in \sigma$ without any witness in $\mathcal{Y}_{\ell}$.
Let $\mathcal{Y}^{\ast}$ denote the set of singleton mosaics where this procedure stabilizes.
\begin{restatable}{lemma}{lemjointconsone}
\label{lem:jointconsone} 
%Let $\overline{t}_{1}$ and $\overline{t}_{2}$ be nominal sequences
%and $\mos$ a singleton mosaic.
The following conditions are equivalent for all $\mos$:
\begin{itemize}
    \item $\mos_{1} \wedge \case(\cmf_{1}),\mos_{2}\wedge \case(\cmf_{2})$ are jointly $\Hmc(\sigma)$-consistent;
    \item $\mathsf{m}\in \mathcal{Y}^{\ast}$.
\end{itemize}
\end{restatable}

\noindent\emph{Separator construction for $\cmf$}. We now construct by induction for each eliminated $\mos$ a $\mathcal{H}(\sigma)$-formula $\sep_{\cmf}(\mos)$ 
with 
\[
\mos_{1} \wedge \case(\cmf_{1}) \models \sep_{\cmf}(\mathsf{m})\quad\text{and}\quad
\mos_{2} \wedge \case(\cmf_{2}) \models \neg\sep_{\cmf}(\mathsf{m}).
\]
Assume $\mos$ is eliminated in the base case. If $\mos_{i}\wedge \case(\cmf_{i})$ is not satisfiable, then set $\sep_{\cmf}(\mos)=\bot$ if $i=1$ and $\sep_{\cmf}(\mos)=\top$ otherwise.
%If there is $a\in \mos_{i}\not=t_{a}^{i}$, then set $\sep_{\overline{t}_{1},\overline{t}_{2}}(\mos)=\bot$ if %$i=1$ and $\text{sep}_{\overline{t}_{1},\overline{t}_{2}}(\mos)=\top$ otherwise. 
If $\mos$ is $\sigma$-inconsistent then there exists $r\in \sigma$ with $r \in \mos_{i}$ and $\neg r\in \mos_{j}$ for $\{i,j\}=\{1,2\}$. Set $\sep_{\cmf}(\mos)=r$ if $i=1$ and 
$\sep_{\cmf}(\mos)=\neg r$ otherwise.

Now assume $\mos$ is eliminated in round $\ell+1$ because there is 
$\mydiam{R}\chi\in \mos_{i}$ with $R\in \sigma$ and no witness 
left in $\mathcal{Y}_{\ell}$. Assume $i=1$ (the case $i=2$ is dual).
Set $W_{1}=\{ t \in \TP\mid \chi\in t, \mos_{1} \rightsquigarrow_{R} t\}$
and $W_{2}= \{t'\in \TP\mid \mos_{2} \rightsquigarrow_{R} t'\}$. Note that all singleton mosaics in $W_{1}\times W_{2}$ have been eliminated before round $\ell+1$. Then we set
\[
\sep_{\cmf}(\mos) = \mydiam{R}(\bigvee_{t\in W_{1}}\bigwedge_{t'\in W_{2}}\sep_{\cmf}(t,t'))
\]
(Here, if $W_{1}=\emptyset$, then $\sep_{\cmf}(\mos)=\bot$ and otherwise, if $W_{2}=\emptyset$ we have $\sep_{\cmf}(\mos)= \mydiam{R}\top$.)
One can show by induction that $\mos_{1} \wedge \case(\cmf_{1}) \models \sep_{\cmf}(\mathsf{m})$ and $\mos_{2} \wedge \case(\cmf_{2}) \models \neg\sep_{\cmf}(\mathsf{m})$, as required.

Let BCase denote the set of all basic nominal cases. Then it follows from Lemma~\ref{lem:jointconsone} that a mosaic $\mos$ is eliminated for all $\cmf\in \text{BCase}$ if, and only if, $\mos_{1},\mos_{2}$ are not jointly $\Hmc(\sigma)$-consistent. For such $\mos$ we can assemble an `overall' $\Hmc(\sigma)$-separator $\sep(\mos)$ using all $\sep_{\cmf}(\mos)$ with $\cmf\in \text{BCase}$ by taking, for $N =\{ f \mid f:\Nom_0\to \TP\}$,
\[
\sep(\mos) = \bigvee_{\cmf_{1}\in N}\bigwedge_{\cmf_{2}\in N}\sep_{(\cmf_{1},\cmf_{2})}(\mos).
\]
Then $\mos_{1}\models \sep(\mos)$ and $\mos_{2}\models\neg\sep(\mos)$, as required. 

Finally, observe that $\varphi_{1},\varphi_{2}$ are not jointly $\Hmc(\sigma)$-consistent if, and only if, for $S_{i}= \{ t \in \TP \mid \varphi_{i}\in t\}$ no pair of types $t_{1},t_{2}$ with $(t_{1},t_{2})\in S_{1}\times S_{2}$ is jointly $\Hmc(\sigma)$-consistent. Hence we construct an $\Hmc(\sigma)$-separator $\sep(\varphi_{1},\varphi_{2})$ for $\varphi_{1},\varphi_{2}$ 
from $\Hmc(\sigma)$-separators for mosaics $(t_{1},t_{2})\in S_{1}\times S_{2}$ 
by setting
\begin{equation}
\sep(\varphi_1,\varphi_2) = \bigvee_{t_{1}\in S_{1}}\bigwedge_{t_{2}\in S_{2}}\sep(t_{1},t_{2}).\tag{\mn{TP2F}}\label{eq:compile}
\end{equation}
It is routine to show that $\sep(\varphi_{1},\varphi_{2})$ is as required. The construction of $\sep_{\cmf}(\mos)$ is in exponential time, provided that separators are always represented as DAGs: we construct only one formula $\sep_\cmf(\mos)$ for every mosaic $\mos$, and there are only exponentially many. Thus, $\sep_{\cmf}(\mos)$ is of at most exponential DAG-size. These exponential bounds are not affected by the later compilation of the $\sep_\cmf(\mos)$ into $\sep(\varphi_1,\varphi_2)$.
We have proved Theorem~\ref{thm:nom} for $\Hmc$.

\begin{restatable}{theorem}{thmfo}
\label{thm:fo}
Let $\Lmc\in \HL$. There is an exponential time algorithm deciding
for any $\varphi_1,\varphi_2\in \Lmc$ and any $\sigma$ whether an FO$(\sigma)$-separator exists 
and, if so, constructs one of exponential DAG-size.   
\end{restatable}
\begin{proof}
Assume $\varphi_{1},\varphi_2$ and $\sigma$ are given. Let $\mathbf{n}_{i}$ be 
the set of nominals in $\varphi_{i}$ not in $\sigma$, for $i\in \{1,2\}$. We may assume that $\mathbf{n}_{1}$ and $\mathbf{n}_{2}$ are disjoint. Let $\sigma'=\sigma \cup \mathbf{n}_{1} \cup \mathbf{n}_{2}$. Recall that FO has CIP, $\Lmc$ has CIP for propositional variables, and
that satisfiability in $\Lmc$ is in \ExpTime. Then it follows that one can decide the existence of an FO$(\sigma)$-separator for $T_{x}(\varphi_{1}),T_{x}(\varphi_{2})$ in exponential time and such a separator exists iff an $\Lmc(\sigma')$-separator exists for $\varphi_{1},\varphi_{2}$ iff $\varphi_1 \wedge \varphi_{2}$ is not satisfiable. Let
$\chi$ be an $\Lmc(\sigma')$-separator for $\varphi_1,\varphi_2$ given by Theorem~\ref{thm:nom} .
Regard the nominals in $\mathbf{n}_{1}$ and $\mathbf{n}_{2}$ as variables, different 
from $x$. Then 
$\exists \mathbf{n}_{1} \forall \mathbf{n}_{2} T_{x}(\chi)$
is an FO($\sigma$)-separator for $T_{x}(\varphi_{1}),T_{x}(\varphi_{2})$. 
\end{proof}

\section{Separators in Hybrid Logics }\label{sec:hyper}

When generalizing Theorem~\ref{thm:nom} from the previous section to the case of arbitrary signatures, we face two challenges. First, singleton mosaics are not sufficient anymore, as non-shared nominals may enforce that models
witnessing joint consistency contain \emph{multiple} pairwise bisimilar elements~\cite{DBLP:journals/tocl/ArtaleJMOW23}. Algorithmically, this is reflected by the fact that instead of working with mosaics which are pairs of types, one has to work with mosaics that are \emph{pairs of sets of types}. Second, the known algorithm for deciding separator existence cannot be easily extended towards separator construction. To motivate our approach, we recall the decision procedure~\cite{DBLP:journals/tocl/ArtaleJMOW23}. By Lemma~\ref{lem:joint-consistent}, it suffices to decide whether given formulae are jointly consistent.

% {\color{blue}OLD First, while existence of separators is still decidable it is known that it is harder than \ExpTime. Algorithmically, this is reflected by the fact that instead of working with mosaics which are pairs of types, one has to work with mosaics that are \emph{pairs of sets of types}. Second, known algorithms cannot be easily extended towards separator construction. To motivate our approach, we recall the decision procedure for separator existence~\cite{DBLP:journals/tocl/ArtaleJMOW23}. By Lemma~\ref{lem:joint-consistent}, it suffices to decide whether given formulae are jointly consistent.}

Fix $\varphi_1,\varphi_2\in\Hmc$ and signature $\sigma$. A \emph{mosaic} is now a pair $\mos=(\mos_1,\mos_2)$ with sets of ($\sub(\varphi_1,\varphi_2)$)-types $\mos_1,\mos_2$. Intuitively, a mosaic describes \emph{sets} of elements in models $\M_1,\M_2$ which are all pairwise $\Hmc(\sigma)$-bisimilar. Formally,
a mosaic $\mos$ is \emph{realized in a pair $\M_{1},\M_{2}$ by an $\Hmc(\sigma)$-bisimulation $Z$ between $\M_{1}$ and $\M_{2}$}
if for $\mos_{1}=\{t_{1},\ldots,t_{k}\}$ and $\mos_{2}=\{t_{1}',\ldots,t_{l}'\}$ there are elements $w_{1},\ldots,w_{k}$ in $\M_{1}$ and $v_{1},\ldots,v_{l}$ in $\M_{2}$ such that
% \nb{JK: compressed here}  for all $1\leq i\leq k$ and $1\leq j\leq l$ we have: $(w_i,v_j)\in Z$, $\tp_{\M_{1}}(w_{i})= t_{i}$ and $\tp_{\M_{2}}(v_{j})= t_{j}'$.
% \[
% (w_i,v_j)\in Z, \hspace{0,4cm} \tp_{\M_{1}}(w_{i})= t_{i}, \hspace{0.3cm}\text{and}\hspace{0.3cm}\tp_{\M_{2}}(v_{j})= t_{j}'
% \]
\begin{itemize}
    \item $(w_i,v_j)\in Z$; % for all $w,w'\in \{w_{1},\ldots,w_{k},v_{1},\ldots,v_{l}\}$;
    \item $\tp_{\M_{1}}(w_{i})= t_{i}$; % for all $1\leq i \leq k$;
    \item $\tp_{\M_{2}}(v_{j})= t_{j}'$; % for all $1\leq i \leq k$.
\end{itemize}
for all $1\leq i\leq k$ and $1\leq j\leq l$.
Note that realizability of a mosaic in this sense is a generalization of joint $\Hmc(\sigma)$-consistency from pairs of formulae to pairs of sets of formulae. 
%\nb{F:maybe just give the definition of realized mosaics here? After all, that's what is needed; JJ: let's do that: I was hesitating, because I originally planned to have this first part a bit informal, but now it is rather complete anyway}\textcolor{red}{or in other words a collection of types that is joint $\Hmc(\sigma)$-consistent. JK: joint consistency was only introduced for pairs of formulae, so this may be confusing as we mean something stronger than joint consistency of every pair of types from $\mos$. Explain or remove joint consistency from example?} 
For mosaics $\mos,\mos'$ and $R\in\Rel$, we write $\mos\rightsquigarrow_R \mos'$ if for every $i\in \{1,2\}$ and $t\in \mos_i$, there is $t'\in \mos'_i$ with $t\rightsquigarrow_R t'$. A mosaic $\mos'$ is a \emph{witness} for $\mydiam{R}\chi\in t\in\mos_i$ if $\chi\in t'$ for some $t'\in \mos_i'$ with $t\rightsquigarrow_R t'$ and, in case $R\in \sigma$, also $\mos\rightsquigarrow_R\mos'$.

% The mosaic elimination procedure is as follows. Let $\Ymc_{0},\ldots,\Ymc_{n}$ be the list of all maximal sets $\Ymc$ of mosaics that are 
The mosaic elimination procedure removes mosaics which cannot be realized. It is based on the observation that in order for $\mos$ to be realizable, every $\mydiam{R}\chi\in t\in\mos_i$ needs to have a realizable witness. In analogy to the basic nominal cases choosing a type for every nominal and side, we now choose a type \emph{and a mosaic} for every nominal and side. Formally, this is reflected by a choice of a maximal initial set \Ymc of mosaics from which we start the elimination. Let $\Ymc_{0},\ldots,\Ymc_{n}$ be all the maximal sets $\Ymc$ of mosaics which are
\begin{itemize}

    \item \emph{$\sigma$-consistent}: for every $\mos\in\Ymc$, $t,t'\in \mos_1\cup\mos_2$, and every $r\in\sigma\cap(\Var\cup\Nom)$, we have $r\in t$ iff $r\in t'$, and 

    \item \emph{nominal-clean}: for every nominal $a\in \sig(\varphi_1)\cup\sig(\varphi_2)$ and all $i\in \{1,2\}$ there is at most one type $t_{a}^{i}\in \bigcup_{(\mos_{1},\mos_{2})\in \Ymc}\mos_{i}$ containing $a$ and, moreover, there is at most one mosaic $(\mos_{1},\mos_{2})\in \Ymc$ with $t_{a}^{i}\in \mos_{i}$.\footnote{The definition in~\cite{DBLP:journals/tocl/ArtaleJMOW23} requires \emph{exactly one} type/mosaic instead of \emph{at most one} type/mosaic in the corresponding condition. Since we now consider \emph{maximal} sets of mosaics with the given property, these two variants of the requirement define the same sets. However, later in the paper we will consider other sets of mosaics, and then our current definition will be more natural.} 
   
\end{itemize}
In the elimination for $\Ymc_{j}$, we set $\Ymc_{j,0}=\Ymc_{j}$ and obtain $\Ymc_{j,\ell+1}$ from $\Ymc_{j,\ell}$ by eliminating all mosaics $\mos$ such that there is $t\in\mos_i$ and $\mydiam{R}\chi\in t$ such that there is no witness in $\Ymc_{j,\ell}$. Thus, a mosaic is eliminated if it cannot be realized in the intended way in models which realize only mosaics from $\Ymc_{j,\ell}$.
One can then show that the following are equivalent: 
\begin{itemize}

    \item $\varphi_1,\varphi_2$ are jointly $\Hmc(\sigma)$-consistent; 
    
    \item for some $\Ymc_j$ the set of mosaics $\Ymc_{j}^*$ where the procedure stabilizes contains a mosaic $\mos$ such that $\varphi_1\in t_1\in \mos_1$ and $\varphi_2\in t_2\in \mos_2$ for types $t_1,t_2$.

\end{itemize}
As there are at most double exponentially many $\Ymc_{j}$ and the elimination procedure also stabilizes after at most double exponentially many steps, it follows that joint consistency can be decided in 2\ExpTime for $\Hmc$. The extension to $\Hfull$ is straightforward and proves a tight 2\ExpTime upper bound. For $\Hmc$ and $\Hat$ one can show a tight \coNExpTime upper bound by proving, in addition, that  
one can construct exponential size models witnessing joint consistency.

\begin{example}\label{ex:motivation}
Assume a signature $\sigma$ containing (possibly among other symbols) relations $R$ and $S$. Let $\psi_a,\psi_{R},\psi_{S}$ be formulae such that $\psi_{a},\psi_{T}$ are jointly $\Hmc(\sigma)$-consistent for $T=R,S$, but $\psi_{a},\psi_{R},\psi_{S}$ are not jointly $\Hmc(\sigma)$-consistent (that is, $\psi_{a},\psi_{R},\psi_{S}$ are not realizable in three pairwise $\Hmc(\sigma)$-bisimilar pointed models), and $\psi_{a}=\psi_{a}'\wedge a$ for a nominal $a\not\in\sigma$. We assume that $a$ is the only nominal in $\psi_a,\psi_{R},\psi_{S}$. Let 
$$
\varphi_1 = \mydiam{R}\mydiam{R}\psi_a\wedge \mydiam{S}\mydiam{S}\psi_a , \quad
\varphi_2 = \mybox{R}\mybox{R}\psi_R \wedge \mybox{S}\mybox{S}\psi_S.
$$
% \psi_0 & = a\wedge \psi_0' \text{ with }\psi_0' = \mydiam{T}(p_1\vee p_2) \\
    % \psi_i & = \mybox{T}\neg p_i \text{, for $i=1,2$}
%\end{align*}
%for $a\in\Nom$, $p_1,p_2\in\Var$, and $R,S,T\in\Rel$.
Observe that $\varphi_1,\varphi_2$ are not jointly $\Hmc(\sigma)$-consistent. Indeed, in a witness $\M_1,\M_2$ for joint $\Hmc(\sigma)$-consistency, there need to be two elements $v_1,v_2$ in $\M_2$ satisfying $\psi_R,\psi_S$ and $\Hmc(\sigma)$-bisimilar to $a^{\M_1}$ satisfying $\psi_{a}$.
We analyze how our elimination procedure detects this.
For the sake of readability, in this example we only partially specify the types in the mosaics. For instance, we talk about \emph{mosaic} $(\{\varphi_1\},\{\varphi_2\})$ and mean by this any mosaic $(\{t_1\},\{t_2\})$ with $\varphi_1\in t_1$ and $\varphi_2\in t_2$.

% We consider three initial sets of mosaics $\Ymc_0,\Ymc_1,\Ymc_2$. All of them contain $\mos^0=(\{\varphi_1\},\{\varphi_2\})$,
% $\mos^{R}= (\{\mydiam{R}\psi_a\},\{\mybox{R}\psi_R\})$, and $\mos^S = (\{\mydiam{S}\psi_a\},\{\mybox{S}\psi_S\})$, but they differ in the mosaics speaking about nominals: $\Ymc_0$ contains 
% $\mos^{RS}=(\{\psi_{a}\},\{\psi_{R},\psi_{S}\})$ while $\Ymc_1$ and $\Ymc_2$ contain $\mos^a_R=$ and

By the witness condition, all $\Ymc_{j}$
that could stabilize at some $\Ymc_{j}^{\ast}$ containing $\mos^0=(\{\varphi_1\},\{\varphi_2\})$ 
contain $\mos^{R}= (\{\mydiam{R}\psi_a\},\{\mybox{R}\psi_R\})$ and $\mos^S = (\{\mydiam{S}\psi_a\},\{\mybox{S}\psi_S\})$.
Let $\mos^{\ast}=(\{\psi_{a}\},\{\psi_{R},\psi_{S}\})$ and $\mos^{\ast T}= (\{\psi_{a}\},\{\psi_{T}\})$, for $T=R,S$. By nominal-cleanness and the witness condition we obtain three initial sets $\Ymc_{\ast},\Ymc_{\ast R},\Ymc_{\ast S}$ of interest, where $\mos^{\ast},\mos^{R},\mos^{S},\mos^{0}\in \Ymc_{\ast}$ and $\mos^{\ast T},\mos^{R},\mos^{S},\mos^{0}\in \Ymc_{\ast T}$, for $T=R,S$.% \nb{JJ: optimal indices? with $*$? We could just go with $0,1,2$ as in the procedure; it's a bit more to write but not much. But it's also understandable now.  }

For $\Ymc_{\ast}$, first $\mos^{\ast}$ is eliminated, then $\mos^{R}$ and $\mos^{S}$ are eliminated as we do not find witnesses, and finally $\mos^{0}$. For $\Ymc_{\ast R}$, $\mos^{\ast R}$ and $\mos^{R}$ are never eliminated, but $\mos^{S}$ is eliminated because no witness exists, and then $\mos^{0}$ is eliminated. Similarly, for $\Ymc_{\ast S}$, $\mos^{R}$ and $\mos^{0}$ are eliminated.\footnote{Observe that if formulae $\varphi_{1},\varphi_{2}$ and $\sigma$ are such that $\sigma$ and either $\varphi_{1}$ or $\varphi_{2}$ do not contain any nominals, then if $\varphi_{1}$ and $\varphi_{2}$ are not jointly $\Hmc(\sigma)$-consistent, one can obtain an $\Hmc(\sigma)$-separator of exponential size by taking the uniform $\sigma$-interpolant of the respective $\varphi_{i}$ in \emph{modal logic without nominals}~\cite{TenEtAl06}. Example~\ref{exm:non-elem} shows that this approach cannot be generalized. It follows, however, that in the example we discuss here there is a rather straightforward way to construct Craig interpolants which we ignore to motivate our general approach.}
\qed
%
% {\color{red} F: needed?} Note that the example does not rely on the exact number of $\mydiam{R}$ and $\mydiam{S}$ in $\varphi$. Indeed, the path to the nominal $a$ could be much more complicated, and cannot easily be treated in a local fashion. \qed
\end{example}
Similar to the previous section, one might now attempt to associate, for each case $\Ymc_{x}=\Ymc_{\ast},\Ymc_{\ast R},\Ymc_{\ast S}$,
inductively with each eliminated mosaic $\mos$ a `separator' $\sep_{\Ymc_{x}}(\mos)$ of $\Hmc(\sigma)$-formulae that `follows from $\mos$ w.r.t.~$\Ymc_{x}$' and is inconsistent. In contrast to the nominal cases used in the previous section, this cannot work, however, as the cases $\Ymc_{x}$ implicitly refer to $\Hmc(\sigma)$-bisimilarity: in case $\Ymc_{\ast}$ we assume that we have realized $\mos^{\ast}$ in $\Hmc(\sigma)$-bisimilar nodes and in case $\Ymc_{\ast T}$ we assume we have only realized $\mos^{\ast T}$, for $T=R,S$. Such assumptions cannot be reflected using derivability of $\Hmc(\sigma)$-formulae. (By contrast, nominal cases of the previous section simply add the condition that certain types are satisfied.) Note that, on the other hand, the case distinction appears necessary because $\mos^{\ast R},\mos^{R}\in \Ymc_{\ast R}$ and $\mos^{\ast S},\mos^{S}\in \Ymc_{\ast S}$ are not eliminated and their union provides the witnesses needed for $\mos^{0}$. Hence, to avoid the case distinction, we should detect that the \emph{set} of mosaics $\{\mos^{R},\mos^{S}\}$ cannot be realized.
%cannconsider have to detect we shouild 
%in isolation and without the full case  
%
We now introduce hypermosaics as a data structure that allows us to consider different mosaics simultaneously and to move between different nominal-clean cases. We will later, when constructing separators, also make nominal cases dynamic.
\subsection{Hypermosaic Elimination}
%The illustrated necessity of considering different mosaics in conjunction motivates the definition of our %main data structure. 
Fix $\varphi_1,\varphi_2\in \Hmc$ and signature $\sigma$. A \emph{hypermosaic} is a set of mosaics, which shall be realized \emph{together} in bisimilar elements in models $\M_1,\M_2$. 
Formally, a hypermosaic $H$ is \emph{realized in a pair $\M_{1},\M_{2}$ by an $\Hmc(\sigma)$-bisimulation $Z$ between $\M_{1}$ and $\M_{2}$} if all $\mos\in H$ are realized by $Z$. We call $H$ \emph{realizable} if it can be realized in described way. We provide a hypermosaic elimination procedure that identifies all realizable hypermosaics. 

We first need to lift the notion of a witness. Let $H,H'$ be hypermosaics. We say that $H'$ \emph{extends} $H$ if the following monotonicity condition is satisfied: 

\begin{description}
\item[Mon]\label{it:monotonicity} for every $\mos\in H$, there is $\mos'\in H'$ with $\mos_i\subseteq\mos_i'$ for both $i\in \{1,2\}$.
%\item[Mon] if no type in $\mos$ contains a nominal, then $\mos\in H'$; otherwise, there is $\mos'\in H'$ with $\mos_1\subseteq \mos_1'$ and $\mos_2\subseteq \mos_2'$.
\end{description}

Clearly,  if $H'$ extends $H$ and is realizable, $H$ is realizable as well.

Let $\mos\in H$ and $i\in\{1,2\}$. The hypermosaic $H'$ is a \emph{witness for $\mydiam{R}\chi\in t\in \mos_{i}$} if $H'$ extends $H$ and there is $\mos'\in H'$ which is a witness for $\mydiam{R}\chi\in t\in \mos_{i}$.
We call $H'$ a \emph{witness for $\mos\in H$} if it is a witness for all $\mydiam{R}\chi\in t\in \mos_i$, for both $i\in \{1,2\}$. Intuitively, we require the witness $H'$ to extend $H$ since the $\mydiam{R}\chi$ may force us to  realize \emph{larger} or \emph{additional} mosaics.

We are now in a position to give the hypermosaic elimination procedure for \Hmc, that identifies the realizable hypermosaics. It starts with the set $\Zmc_0$ of all hypermosaics that are $\sigma$-consistent and nominal-clean (as defined in the description of the mosaic elimination). 
Then, $\mathcal{Z}_{\ell+1}$ is obtained from $\mathcal{Z}_{\ell}$ by dropping all $H$ that contain a mosaic $\mos$ that does not have a witness in $\Zmc_\ell$. Let $\mathcal{Z}^{*}$ be the set of hypermosaics where this stabilizes.

\begin{restatable}{theorem}{thmhypermosaics}
\label{thm:correctness-elimination}
For a hypermosaic $H$, the following are equivalent: 
\begin{itemize}

    \item  $H$ is realizable by an $\Hmc(\sigma)$-bisimulation;
    
    \item $H\in \Zmc^*$.
    
\end{itemize}
\end{restatable}
\begin{proof}(Sketch)
Assume first that $H$ is realized in $\M_{1}, \M_{2}$ by $\Hmc(\sigma)$-bisimulation $Z$. 
% Without loosing generality $Z$ is the \emph{maximal} $\Hmc(\sigma)$-bisimulation between $\M_{1}$ and $\M_{2}$.
Define $H'$ as the set of maximal (w.r.t.~set-inclusion on both coordinates) mosaics $\mos$ realized by $Z$. By definition, $H'$ is nominal-clean, $\sigma$-consistent, and is a witness for all $\mos\in H'$. So $H'$ is not eliminated. Moreover, $H'$ extends $H$, and hence the latter is not eliminated either.

Conversely, assume that $H$ is not eliminated. Take $H^*\in\Zmc^*$ extending $H$ which is a witness for every $\mos\in H^*$. Since $H^*$ extends $H$, it is enough to realize $H^*$.
We now construct models $\M_{1}$ and $\M_{2}$ realizing $H^*$ by $\Hmc(\sigma)$-bisimulation $Z$. For every $i\in\{1,2\}$, let Nom$_{i}$ be the set of nominals $a\in \Nom \cap (\sig(\varphi_1)\cup\sig(\varphi_2))$ that do not occur in any $t$ with $t\in \mos_{i}$ 
for some $\mos\in H^*$. Take
\[
W^{\M_{i}}= \{(t,\mos)\mid t\in \mos_{i}, \mos\in H^*\} \cup \text{Nom}_{i}
\]
as the domain of $\M_i$.
We next define $R^{\M_{i}}$ for relation symbols $R\in\sig(\varphi_1)\cup\sig(\varphi_2)$. Let $(t,\mos),(t',\mos')\in W^{\M_{i}}$. Then $((t,\mos),(t',\mos'))\in R^{\M_i}$ iff $t\rightsquigarrow_R t'$ and, if $R\in \sigma$, additionally $\mos \rightsquigarrow_R \mos'$.
For a propositional variable $p\in \sig(\varphi_1)\cup\sig(\varphi_2)$, let $p^{\M_{i}}=\{ (t,\mos)\in W_{i} \mid p\in t\}$.
For $a\in \Nom \cap (\sig(\varphi_1)\cup\sig(\varphi_2)$ we make a case distinction.
If there is $(t,\mos)\in W_{i}$ with $a\in t$, then let $a^{\M_{i}}= (t,\mos)$.
Otherwise $a\in \text{Nom}_{i}$ and we set $a^{\M_{i}}=a$.
Finally, define $Z\subseteq W^{\M_1}\times W^{\M_2}$ by setting
\begin{gather*}
(t,\mos)\ Z\ (t',\mos') \iff \mos=\mos'.
\end{gather*}
One can inductively verify that (1) $\psi\in t$ if, and only if, $\Mmc_i,(t,\mos)\models \psi$, for every $i\in\{1,2\}$, $\psi\in\sub(\varphi_1,\varphi_2)$ and every $(t,\mos)\in W^{\Mmc_i}$; and that (2) $Z$ is an $\Hmc(\sigma)$-bisimulation between $\Mmc_1$ and $\Mmc_2$ that realizes each $\mos\in H^*$.
% \begin{itemize}
%     \item $\psi\in t\quad\text{ iff }\quad\Mmc_i,(t,\mos)\models \psi$ for every $i\in\{1,2\}$, $\psi\in\sub(\varphi_1,\varphi_2)$ and every $(t,\mos)\in W^{\Mmc_i}$;
%     \item $Z$ is an $\Hmc(\sigma)$-bisimulation between $\Mmc_1$ and $\Mmc_2$, and each $\mos\in H^*$ is realized in $\Mmc_1,\Mmc_2$ by $Z$.
% \end{itemize}
\end{proof}

The following is an immediate consequence of  
Theorem~\ref{thm:correctness-elimination}. 

\begin{corollary}\label{cor:hypermosaics}
    $\varphi_1$ and $\varphi_2$ are jointly $\Hmc(\sigma)$-consistent iff there exists a hypermosaic $\{(\{t_{1}\},\{t_{2}\})\}\in\Zmc^*$ with $\varphi_1\in t_{1}$ and $\varphi_2\in t_{2}$.
\end{corollary}

The following example illustrates the elimination procedure.

\begin{example}\label{ex:hypermosaics}
For $\varphi_1,\varphi_2,\sigma$ and mosaics $\mos^*,\mos^{\ast R},\mos^{\ast S},\mos^R,\mos^S,\mos^0$ as in Example~\ref{ex:motivation}, the initial set of hypermosaics $\Zmc_0$ will contain (among others) the following hypermosaics $H_0, H_1, H_{\ast R}, H_{\ast S},H_{\ast}$: 
\begin{align*}
    H_0 = \{\mos^0\},\quad
    H_1 = \{\mos^R,\mos^S,\mos^0\},\quad 
    H_{\ast R} = \{\mos^{*R},\mos^R,\mos^S,\mos^0\} \\
    H_{\ast S} = \{\mos^{*S},\mos^R,\mos^S,\mos^0\}, \quad
    H_{\ast} = \{\mos^*,\mos^R,\mos^S,\mos^0\}.
\end{align*}
Note that $H_1$ is a \emph{minimal} witness for $\mos^0\in H_0$ in the sense that every witness for $\mos^0$ extends $H_1$. Similarly, $H_{\ast R}$ and $H_{\ast S}$ are minimal witnesses for $\mos^R$ and $\mos^S$ in $H_1$, respectively, and $H_{\ast}$ is a minimal witness for $\mos^S\in H_{\ast R}$ and $\mos^R\in H_{\ast S}$, respectively. 
Further note that our elimination procedure is monotone in the sense that whenever $H$ is eliminated and $H'$ extends $H$, then $H'$ is eliminated as well. Thus, if a minimal witness is eliminated then all other witnesses are eliminated as well.

Since $\psi_{a},\psi_{R},\psi_{S}$ cannot be realized in $\Hmc(\sigma)$-bisimilar elements, $H_{\ast}$ is not realizable by an $\Hmc(\sigma)$-bisimulation. By Theorem~\ref{thm:correctness-elimination}, there is a round $\ell$ when $H_{\ast}$ is eliminated, that is, $H_{\ast}\notin\Zmc_\ell$. In the next round $\ell+1$, $H_{\ast R}$ and $H_{\ast S}$ are eliminated since $H_{\ast}$ was the minimal witness for $\mos^S\in H_{\ast R}$ and $\mos^R\in H_{\ast S}$, respectively. By similar arguments, $H_1$ and $H_0$ are eliminated in the next two rounds.
\qed
\end{example}

It is not difficult to show that  Theorem~\ref{thm:correctness-elimination} and Corollary~\ref{cor:hypermosaics} also hold for $\Hfull$ and $\Hat$, when the following changes are applied. 

For $\Hat$, we assume without loss of generality that $@_aa\in\sub(\varphi_1,\varphi_2)$ for all nominals $a\in \sigma$, and thus that $@_aa\in t$ for every type $t$ and $a\in\sigma$.
% $(\{t_{a}^{1}\},\{t_{a}^{2}\})\in H$ for all $a\in \sigma$ for nominal-clean hypermosaics $H$. 
In the definition of a witness $H'$ for a mosaic $\mos\in H$, we additionally require that for every $i\in\{1,2\}$, $t\in\mos_i$, and $@_a\chi\in t$ there exists $\mos'\in H'$ such that $a,\chi\in t'$ for some $t'\in \mos'_i$. Finally, we require that all hypermosaics $H$ in $\Zmc_0$ are 
\begin{itemize}
    \item \emph{@-consistent:} for every $\mos,\mos'\in H$ and $t,t'\in \mos_{i}\cup \mos_{i}'$, $@_a\chi\in t$ iff $@_a\chi\in t'$, for all $@_a\chi \in \sub(\varphi_1,\varphi_2)$ and all $i\in \{1,2\}$.
    
\end{itemize}

For $\Hfull$, we assume without loss of generality that $\mydiam{U}a\in \sub(\varphi_1,\varphi_2)$, for all $a\in\sub(\varphi_1,\varphi_2)$, and thus that $\mydiam{U}a$ is contained in every type. Moreover, we allow only mosaics $\mos$ with $\mos_1\neq \emptyset$ and $\mos_2\neq\emptyset$, and require that all hypermosaics $H$ in $\Zmc_0$ be
\begin{itemize}

    \item \emph{$U$-consistent}: for every $\mos,\mos'\in H$ and $t,t'\in \mos_{i}\cup \mos_{i}'$, $\mydiam{U}\chi\in t$ iff $\mydiam{U}\chi\in t'$, for all $\mydiam{U}\chi \in \sub(\varphi_1,\varphi_2)$ and all $i\in \{1,2\}$.
\end{itemize}
Finally, during the elimination, we require witnesses also for formulae $\mydiam{U}\chi$ and treat them as if $U$ was part of $\sigma$.

We conclude by observing that hypermosaic elimination is inherently triply exponential since there are triply exponentially many hypermosaics, and is therefore one exponential more expensive than mosaic elimination. We show next that in contrast to mosaic elimination, hypermosaic elimination can be made constructive. 

\subsection{Separator Construction}

We start with the case of the hybrid logic $\Hmc$. The easy modifications needed to adapt the construction to $\Hat$ and $\Hfull$ will be discussed in the end.
Recall that $\phi_1,\phi_2\in\Hmc$ and signature $\sigma$ were the input to the elimination procedure, and that $\TP$ denotes the set of all relevant types. Let $\Nom_0=\Nom\cap\sub(\phi_1,\phi_2)$ denote the set of all relevant nominals.
% \nb{JK: maybe make the link between separators for $\phi_1,\phi_2$ and types containing them here? or to Lemma~\ref{lem:joint-consistent}}
We first generalize the notion of a separator for a pair of formulae to hypermosaics. As in Section~\ref{sec:first-order}, the separators depend on the `type' of the nominals, but in a more dynamic way. 

A \emph{nominal case} is a pair $\cmf=(\cmf_1,\cmf_2)$ of sequences of $\Hmc$-formulae,
one for each nominal
\[
\cmf_1,\cmf_2:\Nom_0\to \Hmc.
\]
As in Section~\ref{sec:first-order}, we use $\case(\cmf_i)$ to abbreviate the formula
\[
\case(\cmf_i) = \bigwedge_{a\in\Nom_0} @_a \cmf_{i}(a).
\]
A \emph{case distinction} is a finite set $\Case$ of nominal cases which 
exhausts all possibilities. That is, every pair of models $(\M_1,\M_2)$ adheres to some case $\cmf\in\Case$ in the sense that $\M_i\models\case(\cmf_i)$ for both $i$.

A separator for a pair of formulae is a witness that they cannot be realized at $\sigma$-bisimilar points. We now define an analogous notion for hypermosaics: \emph{a hyperseparator} which witnesses that a given hypermosaic $H$ cannot be realized by a bisimulation between two models. Since we will construct hyperseparators in a uniform way, for convenience we first define a hyperseparator without any fixed hypermosaic, and then say what it means that this separator separates a given $H$.

\begin{definition}\label{def:separator}
A \emph{hyperseparator} consists of a case distinction $\Case$
and functions $\sep_1,\sep_2$ of type
\[
\sep_i:\Case\times\TP \to \Hmc(\sigma)
\]
such that for every $t\in\TP$, $\cmf\in\Case$, and all $i\in\{1,2\}$:
\[
t\wedge\case(\cmf_i)\models\sep_i(\cmf,t).
\]
We say that $\Case,\sep_1,\sep_2$ \emph{separates} a hypermosaic $H$ if for every $\cmf\in\Case$ there is $\mos\in H$ such that:
\[
\bigwedge_{i=1,2}\bigwedge_{t\in \mos_i}\sep_i(\cmf,t) \models \bot.
\]
\end{definition}
Intuitively, a hyperseparator for $H$ certifies that $H$ is not realizable. 
% Note that Definition~\ref{def:separator} generalizes the notion of a separator for a pair of formulae to hypermosaics. 
The definition ranges over some exhaustive set of cases, and for each case $\cmf$, side $i$, and type~$t$, provides a $\Hmc(\sigma)$-formula $\sep_i(\cmf,t)$ which follows from $t$ in the case $\cmf$.
The hyperseparator separates $H$ if no matter which case $\cmf\in\Case$ holds, we always find an inconsistency in the consequences of some mosaic $\mos\in H$, showing that $H$ is not realizable.

In what follows, we will supplement the elimination procedure with an elementary construction of hyperseparators which separate all eliminated hypermosaics. Before we do that, let us motivate the construction by showing how the usual separators for types are easily obtained from hyperseparators for hypermosaics.

\begin{lemma}\label{lem:hypersep2sep}
    Assume a hyperseparator $\Case,\sep_1,\sep_2$ for a hypermosaic $\{(\{t_1\},\{t_2\})\}$, and let
    \[
    m=\max\{|\sep_i(\cmf,t)| \mid \cmf\in\Case,\ i=1,2,\ t\in\TP\}.
    \]
    One can construct a separator for $t_1,t_2$ of size polynomial in $2^{|\Case|}\times m$.
    %is the maximal size of a formula used by $\sep_i$.
    % One can construct a separator for $t_1,t_2$ of size polynomial in $2^{|\Case|}\times m$.
\end{lemma}

\begin{proof}
Assume a hyperseparator $\Case,\sep_1,\sep_2$ separating $\{(\{t_1\},\{t_2\})\}$. For $i\in \{1,2\}$ denote $\Cmc_i=\{\cmf_i\mid \cmf\in\Case\}$ and consider the family $\widetilde{\Cmc_i}$ of all \emph{satisfiable} subsets of $\Cmc_i$ in the sense 
that \[
\widetilde{\Cmc_i} = \{S\in\Pmc(\Cmc_i) \mid \chi_S\not\models\bot\}
%\]
%where
%\[
\quad\quad\text{where}\quad\quad \chi_S=\bigwedge_{\cmf_i\in S}\case(\cmf_i) \wedge \bigwedge_{\cmf_i\in \Cmc_{i}\setminus S}\neg\case(\cmf_i). 
\]
Define:
\[
\theta^S_i = \bigwedge_{\substack{\cmf\in\Case\\ \cmf_i\in S}}\sep_i(\cmf,t_{i})
\hspace{1cm} \text{and} \hspace{1cm}
\psi_i=\bigvee_{S\in \widetilde{\Cmc_{i}}} \theta^S_i
\]
for both $i\in \{1,2\}$ and $S\in\widetilde{\Cmc_i}$. We claim that
\[
t_1\models\psi_1\models\neg\psi_2\models\neg t_2
\]
which means that $\psi_1$ (and also $\neg\psi_2$) separates $t_1,t_2$.

Since $\chi_S\models\case(\cmf_i)$ for every $\cmf_i\in S\in\widetilde{\Cmc_i}$, we get $t_i\wedge\chi_S\models \theta^S_i$.
Since $\bigvee_{S\in \widetilde{\Cmc_{i}}}\chi_S$ is tautological by construction of $\widetilde{\Cmc_i}$, this gives us
$t_i\models \psi_i$ for both $i\in \{1,2\}$. It remains to show $\psi_1\wedge\psi_2\models\bot$. Assume the opposite. Then there are $S_1\in\widetilde{\Cmc_1}$ and $S_2\in\widetilde{\Cmc_2}$ such that $\theta_1^{S_1}\wedge\theta_2^{S_2}$ is satisfiable. By definition of $\widetilde{\Cmc_i}$, we have models $\Mmc_{i}$ satisfying $\chi_{S_i}$ for both $i\in \{1,2\}$. The pair $\M_1,\M_2$ of models must adhere to some case $\cmf\in\Case$. Then we have $\cmf_i\in S_i$, for both $i\in \{1,2\}$.
But this implies $\theta^{S_i}_i\models\sep_i(\cmf,t_i)$. This is a contradiction because $\Case,\sep_1,\sep_2$ separates $\{(\{t_1\},\{t_2\})\}$ and hence
$\sep_1(\cmf,t_1)\wedge\sep_2(\cmf,t_2)\models\bot$.
\end{proof}

We now provide the construction of hyperseparators. The construction proceeds by induction on the number of rounds of the elimination procedure, and in round $\ell$ we build a separator for all hypermosaics eliminated up to that round.

\begin{lemma}\label{lem:sep construction}
For every round $\ell$ of the elimination procedure, we can compute a hyperseparator $\Case,\sep_1,\sep_2$ which separates every hypermosaic eliminated up to that round. 
\end{lemma}

% \begin{proof}
We first provide the inductive construction and then prove its correctness. 
In the base case, we let the case distinction determine the type of each nominal, that is:  
\begin{align*}
\Case = \{(\cmf_1,\cmf_2)\ |\ \cmf_1,\cmf_2:\Nom_0\to\TP \}.% \label{eq:base case dist}    
\end{align*}
This case distinction is exhaustive because every nominal in every model has some type.
For every case $\cmf\in\Case$, the functions $\sep_1,\sep_2$ ignore $\cmf$ and map types to their atomic $\sigma$-consequences:
\[
\sep_i(\cmf,t) = \bigwedge(t\cap\{r,\neg r\ |\ r\in\sigma\})
\]
for every $i,\cmf$ and $t$. Clearly, $t\models \sep_i(\cmf,t)$ and thus $\Case,\sep_1,\sep_2$ is a hyperseparator. Moreover, it separates every $\sigma$-inconsistent hypermosaic.

We now deal with the inductive case.
Assume we have a hyperseparator $\Case,\sep_1,\sep_2$ which separates all hypermosaics eliminated up to round $\ell$. We construct a hyperseparator $\Case^+,\sep_1^+,\sep_2^+$ separating all hypermosaics eliminated up to stage $\ell+1$.

We first refine each case in $\Case$.
For $\cmf\in\Case$ consider the set:
\[
\Sep_\cmf=\{\sep_i(\cmf,t)\mid t\in\TP,i=1,2\}
\]
of all $\Hmc(\sigma)$-formulae which are in the image of $\sep_1(\cmf,\ \cdot\ )$ or $\sep_2(\cmf, \ \cdot\ )$.
We split $\cmf$ into subcases specifying which $\Sep_\cmf$-type $d\in\SepCaseType$ is true at each nominal. The new case distinction is
\[
\Case^+=
\left\{
  (\cmf_1 \otimes f_1,\ \ \cmf_2 \otimes f_2) \;\middle|\;
  \begin{aligned}
   \cmf\in\Case,  f_1,f_2:\Nom_0\to\SepCaseType
  \end{aligned}
\right\}
\]
where $\otimes$ denotes pointwise conjunction, that is,  $(f\otimes g)(a)=f(a)\wedge g(a)$, for every $a\in\Nom_0$.
We say that $(\cmf^+_1,\cmf^+_2)\in\Case^+$ \emph{refines} $(\cmf_1,\cmf_2)\in\Case$ if both $\cmf^+_i$ are of the shape $\cmf_i\otimes f_i$ for some $f_i$. For every $\cmf^+\in\Case^+$ the case from $\Case$ refined by $\cmf^+$ can be uniquely read off from $\cmf^+$, and we denote it by $\cmf$. 

We now construct $\sep^+_1$ and $\sep^+_2$. For every $\cmf\in\Case$ define a \emph{$\Sep_\cmf$-star-type} of a given point $v$ to be the full description of $v$ and its neighbors with respect to $\SepCaseType$: the $\Sep_\cmf$-type of $v$ and information which $\Sep_\cmf$-types are true in $v$'s $R$-children for every $R\in\sigma$. The set of all $\Sep_\cmf$-star-types will be denoted $\SepCaseStarType$. Formally, let $\nabla^R\Psi$ for a finite set $\Psi$ of formulae be a shorthand for $\bigwedge_{\psi\in\Psi}\mydiam{R}\psi\wedge\mybox{R}\bigvee_{\psi\in\Psi}\psi$, and set
\[
\SepCaseStarType = \{d\wedge\bigwedge_{R\in\sigma}\nabla^R \mathbb{D}^R\ |\ d\in\SepCaseType,\ \mathbb{D}^R\subseteq\SepCaseType\}.
\]

For every type $t$, case $\cmf^+=(\cmf^+_1,\cmf^+_2)\in\Case^+$, and side $i\in \{1,2\}$ we define:
\[
\sep^+_i(\cmf^+,t) = \bigvee
\left\{
  \delta \in\SepCaseStarType\mid
  \begin{aligned}
  \case(\cmf^+_i)\wedge\delta\wedge t\not\models\bot
  \end{aligned}
\right\}.
\]
That is, $\sep^+_i(\cmf^+,t)$ is the disjunction of all $\Sep_\cmf$-star-types consistent with $t$ in case $\cmf^+$, for $\cmf$ the case from $\Case$ refined by $\cmf^+$. If the set of disjuncts is empty, then as usual 
$\sep^+_i(\cmf_i^+,t) = \bigvee\emptyset=\bot$. This may happen when the type $t$ is inconsistent with case $\cmf^+_i$, for example when already $\case(\cmf^+_i)$ is unsatisfiable, or when
$\mydiam{R}(a\wedge p)\in t$ but $\case(\cmf^+_i)$ enforces that $p$ is false at nominal $a$.

% Note that the definition of $\sep^+_i$ covers a situation when $t$ cannot be satisfied together with $\case(\cmf^+_i)$. This could happen for example when $\mydiam{R}(a\wedge p)\in t$ (with $R$ not necessarily in $\sigma$) but $\cmf^+_i$ enforces that $p$ is false at $a$. Another scenario when $\case(\cmf_i)\wedge t$ is not satisfiable is when $\case(\cmf_i)$ alone is already unsatisfiable. In any such situation $\sep^+_i(\cmf_i^+,t) = \bigvee\emptyset=\bot$.

This finishes the construction of $\Case^+,\sep^+_1,\sep^+_2$. Before we provide the correctness proof, let us illustrate the construction by continuing Examples~\ref{ex:motivation} and~\ref{ex:hypermosaics}.

\begin{example}\label{ex:separator-construction}
We describe how to construct separators along the elimination described in Example~\ref{ex:hypermosaics}. For the sake of readability we only loosely follow the general construction and focus on its key mechanism. Consider the step $\ell$ when the hypermosaic $H_{\ast}$ is eliminated and assume we have a hyperseparator for it. For simplicity assume that this hyperseparator is actually case-free and separates $\mos^\ast$, in the sense that there are $\theta_a,\theta_R,\theta_S\in\Hmc(\sigma)$ with $\psi_x\models\theta_x$ for $x=a,R,S$ and $\theta_a\wedge\theta_R\wedge\theta_S\models\bot$. We construct a hyperseparator $\Case,\sep_1,\sep_2$ which separates $H_{\ast R}$ and $H_{\ast S}$. We distinguish three cases $\Case=\{\cmf_a,\cmf_R,\cmf_S\}$ depending on the behavior of the nominal $a$ on the left: $\cmf_x$ says that $\theta_x$ is false at the nominal $a$ on the left, for $x=a,R,S$. Note that, unlike in our general construction, the case distinction $\Case$ is not necessarily exclusive, but this is not important. What is crucial is that the distinction is exhaustive, thanks to $\theta_a\wedge\theta_R\wedge\theta_S\models\bot$. We set:
\[
\sep_1(\cmf_a,\mydiam{T}\psi_a)=\bot, \hspace{1cm} \sep_1(\cmf_T,\mydiam{T}\psi_a)=\mydiam{T}\neg\theta_T,
\]
and for the other side
\[
\sep_2(\cmf_T,\mybox{T}\psi_T)=\mybox{T}\theta_T
\]
for $T=R,S$. The remaining values of $\sep_i$'s are irrelevant and we set them to $\top$. It is easily verified that in every case formulae follow from the respective types, and hence $\Case,\sep_1,\sep_2$ is a hyperseparator. It separates $H_{\ast R}$ because for every case $\cmf\in\Case$ there is a mosaic $\mos\in H_{\ast R}$ with inconsistent consequences: $\mos^R$ for $\cmf_a$ and $\cmf_R$, and $\mos_S$ for $\cmf_S$. This actually shows that the hyperseparator separates $H_{\ast S}$ and $H_1$ as well. Continuing the construction we build a hyperseparator for $H_0$ in a similar way, and in our example this can be done without further refinement of the case distinction.\qed 
\end{example}

We now prove correctness of the construction, that is, the triple $\Case^+,\sep^+_1,\sep^+_2$ is a hyperseparator, and separates every $H$ eliminated up to round $\ell+1$.
 
% \begin{lemma}\label{lem:construction-correctness}
%     The triple $\Case^+,\sep^+_1,\sep^+_2$ is a hyperseparator, and separates every $H$ eliminated up to round $\ell+1$.
% \end{lemma}
% 
% \begin{proof}
We first show that $\Case^+,\sep^+_1,\sep^+_2$ is a hyperseparator. To see that $\Case^+$ is exhaustive take models $\M_1,\M_2$. Since $\Case$ was exhaustive, there is $\cmf=(\cmf_1,\cmf_2)\in\Case$ such that  for both $i$, $\M_i$ satisfies $\case(\cmf_i)$. Define $f_i(a)$ to be the $\Sep_\cmf$-type of $a$ in $\M_i$. It follows that $\M_i$ satisfies $\case(\cmf_i\otimes f_i)$, and hence the pair of models adheres to $(\cmf_1\otimes f_1,\ \ \cmf_2\otimes f_2)\in\Case^+$. That is, $\Case^+$ is exhaustive. Since $\sep^+_i$ is a disjunction ranging over all consistent possibilities, it follows that:
\[
\case(\cmf^+_i)\wedge t \models \sep^+_i(\cmf^+,t)
\]
for all $t$ and $i$. This proves that $\Case^+,\sep^+_1,\sep^+_2$ is a hyperseparator.

It remains to prove that it separates every hypermosaic $H$ eliminated up to round $\ell+1$. Assume such an $H$.
We need to show that for every $\cmf^+\in\Case^+$ there is $\mos\in H$ such that:
\begin{align}
    \bigwedge_{i=1,2}\bigwedge_{t\in \mos_i}\sep^+_i(\cmf^+,t) \models \bot. \label{eq:nominals proof inconsistency 1}
\end{align}
Note that for every $\cmf\in\Case$, different $\Sep_\cmf$-types are mutually exclusive, and hence so are different $\Sep_\cmf$-star-types.
Consequently, it follows for every $\cmf^+\in\Case^+$ that $\sep^+_i(\cmf^+,t)\wedge\sep^+_j(\cmf^+,s)$ is consistent iff $\sep^+_i(\cmf^+,t)$ and $\sep^+_j(\cmf^+,s)$ have a common disjunct.
Assume towards contradiction that~\eqref{eq:nominals proof inconsistency 1} is false. This means that there is $\cmf^+\in\Case^+$ which refines some $\cmf\in\Case$, and such that for every mosaic $\mos\in H$ there is a $\Sep_\cmf$-star-type $\delta_\mos$ which is a disjunct of $\sep_i^+(\cmf^+,t)$ for every $i$ and $t\in \mos_i$.

Define for each $i\in \{1,2\}$ and $d\in\SepCaseType$:
\[
\mos^d_i=\{t\in\TP\ |\ \case(\cmf^+_i)\wedge t\wedge d\not\models\bot\}.
\]
We now define a new hypermosaic:
\[
H' = \{(\mos^d_1,\mos^d_2)\ |\ d\in\SepCaseType\}.
\]

\begin{restatable}{claim}{claimconstruction}
\label{cl:construction claim}
The following hold:
\begin{enumerate}
  \item $H'$ is nominal clean.
  \item $H'$ is a witness for every $\mos\in H$.
\end{enumerate}
\end{restatable}

% \begin{claim}\label{cl:construction claim}
% We claim that:
% \begin{enumerate}
%   \item $H'$ is nominal clean.
%   \item $H'$ is a witness for every $\mos\in H$.
% \end{enumerate}
% \end{claim}

%We have proven Claim~\ref{cl:construction claim}. 
\begin{proof}
We first prove that $H'$ is nominal-clean. Fix a nominal $a\in\Nom$ and side $i\in \{1,2\}$.
Note that $\cmf_i^+$ specifies the type $t$ and $\Sep_\cmf$-type $d$ of $a$, due to the inductive base of the separator construction and the inductive step, respectively. Hence, on the $i$-th side $a$ can only belong to the type $t$ in the mosaic $(\mos^d_1,\mos^d_2)$.

We now prove the other claim that $H'$ is a witness for every $\mos\in H$. Take $\mydiam{R}\chi\in t\in\mos_i$. Consider the $\Sep_\cmf$-star-type $\delta_\mos$. Since $\delta_\mos$ is a disjunct of $\sep^+_i(\cmf^+,t)$ there is a model
\[
\M,w\models\case(\cmf^+_i)\wedge\delta_\mos\wedge t.
\]
Since $\M,w\models t$, it follows that $w$ has an $R$-child $w'$ satisfying some $t'\wedge d$ with $\chi\in t'\in\TP$ and $d\in\SepCaseType$. In particular $t'\in \mos_i^d$. Denote $\mos'=(\mos^d_1,\mos^d_2)$. Since $wRw'$ as well as $\tp_\M(w)=t$ and $\tp_\M(w')=t'$, we have $t\rightsquigarrow_{R} t'$.

If $R\notin\sigma$ then we are done with showing that $\mos'$ witnesses $\mydiam{R}\chi\in t\in\mos_i$. Otherwise $R\in\sigma$ in which case we need to additionally prove $\mos\rightsquigarrow_{R}\mos'$. Take any $s\in \mos_j$. We need to find some $s'\in \mos^d_j$ such that $s\rightsquigarrow_{R} s'$. Because $\delta_\mos$ is a disjunct of $\sep^+_j(\cmf^+,s)$ there is a model
\[
\N,v\models\case(\cmf^+_j)\wedge\delta_\mos\wedge s.
\]
Since $\M,w\models\delta_\mos$ as well as $wRw'$ and $\M,w'\models d$, we get $\delta_\mos\models\mydiam{R}d$. Hence $v$ has an $R$-child $v'$ in $\N$ such that $\N,v'\models d$. Take the type $s'=\tp_\N(v')$ of this $v'$. Since $\N,v'\models\case(\cmf^+_j)\wedge d\wedge s'$, we get $s'\in\mos_j^d$. Moreover, $s\rightsquigarrow_{R} s'$ follows from $vRv'$. This completes the proof that $\mos'$ witnesses $\mydiam{R}\chi\in t\in\mos_i$.

It remains to prove that $H'$ extends $H$. Take $\mos\in H$. Take the unique $\Sep_\cmf$-type $d$ enforced by $\delta_\mos$ in the sense that $\delta_\mos\models d$. For every $i$ and $t\in \mos_i$ we know that $\case(\cmf^+_i)\wedge\delta_\mos\wedge t$ is consistent, and hence so is $\case(\cmf^+_i)\wedge d\wedge t$. This implies $t\in \mos_i^d$. In other words, $\mos_i\subseteq \mos^d_i$, for both $i$. This proves that $H'$ extends $H$ and hence completes the proof of Claim~\ref{cl:construction claim}.
\end{proof}

Thus, by the induction assumption $H'$ was eliminated up to the previous round $\ell$, and so $\Case,\sep_1,\sep_2$ separates $H'$.
Thus, for the unique $\cmf\in\Case$ refined by $\cmf^+$ there is $d\in\SepCaseType$ such that $(\mos^d_1,\mos^d_2)\in H'$ has inconsistent consequences:
\begin{align}
\bigwedge_{i=1,2}\bigwedge_{t\in \mos^d_i}\sep_i(\cmf,t) \models \bot.\label{eq:nominals proof inconsistency 2}    
\end{align}
Since $\case(\cmf^+_i)\models\case(\cmf_i)$ and, for every  
$t\in \mos^d_i$, $\case(\cmf_i)\wedge t\models\sep_i(\cmf,t)$, it follows that
\[
\case(\cmf^+_i)\wedge t\models\sep_i(\cmf,t),
\]
for every $i\in\{1,2\}$, and $t\in \mos_i^d$. Since $\sep_i(\cmf,t)\in\Sep_\cmf$, either $\sep_i(\cmf,t)$ or $\neg\sep_i(\cmf,t)$ belongs to $d$. It follows from $t\in\mos^d_i$ that $\case(\cmf^+_i)\wedge t\wedge d$ is consistent and therefore the above implies
\[
\sep_i(\cmf,t) \in d.
\]
Since this is true for every $t\in\mos^d_i$ and $i$, by~\eqref{eq:nominals proof inconsistency 2} $d$ is inconsistent, a contradiction. This completes the proof of~\eqref{eq:nominals proof inconsistency 1} and thus finishes the proof of correctness. This % proves % Lemma~\ref{lem:construction-correctness}, which in turn 
completes the proof of Lemma~\ref{lem:sep construction}.
% \nb{JJ: we have the proof of a lemma inside the proof of another lemma}
% Note that in the special case when $\case(\overline{\cmf}^+_i)$ is inconsistent with $t$ the disjunction $\sep^+_i(\overline{\cmf}_1^+,t)$ is empty and hence equivalent to $\bot$.
% \end{proof}

It remains to sum up our constructions.

\begin{restatable}{theorem}{thmsummingup}\label{thm:summing-up}
  If $\varphi_1,\varphi_2$ are not jointly $\Hmc(\sigma)$-consistent, then we can compute in fourfold exponential time an $\Hmc(\sigma)$-separator (of fourfold exponential size) for $\varphi_1,\varphi_2$.
\end{restatable}

\begin{proof}
Suppose $\varphi_1,\varphi_2$ are not jointly $\Hmc(\sigma)$-consistent. By Corollary~\ref{cor:hypermosaics}, all hypermosaics $\{(\{t_1\},\{t_2\})\}$ with $\varphi_1\in t_1,\varphi_2\in t_2$ are eliminated. Lemma~\ref{lem:hypersep2sep} gives us an $\Hmc(\sigma)$-separator for each such pair $t_1,t_2$ of types. These separators can be assembled to an $\Hmc(\sigma)$-separator of $\varphi_1,\varphi_2$ as in  Equation~\eqref{eq:compile}.

The analysis of both the running time of the procedure and the size of the constructed separators is based on two key observations. First, hypermosaics $H'$ witnessing some $\mos\in H$ have to satisfy Condition~\textbf{Mon}. This restriction ensures that the procedure terminates after at most doubly exponentially many rounds, despite the fact that there are triply exponentially many hypermosaics. Second, in each round $\ell$ of the elimination, one can bound the maximal size $m_\ell$ of $\sep_{i}(\cmf,t)$ independently of both the number and the size of the cases in~$\Case$. We give a brief inductive argument. 

Clearly, $m_{0}\leq |\sigma|$ by the definition of the base case. 
For the inductive step, observe that for a case $\cmf\in \Case$, a $\Sep_\cmf$-type is built from exponentially many formulae $\sep_i(\cmf,t)$, one for each type $t$ and, by induction, each of size at most $m_\ell$. Hence, the number of $\Sep_\cmf$-types is bounded by $2^{2^{n}}$ and their size is $O(2^n\times m_\ell)$. Since a 
$\Sep_\cmf$-star-type is determined by one $\Sep_\cmf$-type and a subset of $\Sep_\cmf$-types for every relation symbol, it follows that: 
\begin{enumerate}

    \item the number of $\Sep_\cmf$-star-types is at most $O(2^{2^n}\times 2^{2^{2^n}})$ and
    
   % \item the size of a $\Sep_\cmf$-type is $O(2^n\times m_\ell)$, and

    \item the size of a $\Sep_\cmf$-star-type is
    $O(2^n\times m_\ell+2^{2^n}\times 2^{n}\times m_\ell)$.
    
\end{enumerate}
Since the size of $\sep_i^+(\cmf^+,t)$ is bounded by the product of the bounds in Items~1 and~2 above, we obtain that
$m_{\ell+1}\leq m_\ell\times f(n)$ for some triply exponential function $f$. Since $\ell$ is at most doubly exponential, a triply exponential upper bound on $m_\ell$ follows. The number of cases can be bound in a similar fashion. 

% It remains to note that the separators constructed in Lemma~\ref{lem:hypersep2sep} and Equation~\eqref{eq:compile} grow only by a single exponential. 
It remains to note that the final step with Lemma~\ref{lem:hypersep2sep} and Equation~\eqref{eq:compile} only adds a single exponential blowup.
\end{proof}

We conclude this section with a comment on how to adapt the construction to hybrid logics with $@$ and $U$. 
%The modification is actually very simple. 
The only change is in the definition of $\Sep_\cmf$-star-types:
\begin{itemize}

    \item With $@$, a $\Sep_\cmf$-star-type additionally specifies $\bigwedge_{a\in\sigma}@_a d_a$ for $d_a\in\SepCaseType$.
    
    \item With $U$, we treat it as a relation in $\sigma$, that is, a $\Sep_\cmf$-star-type additionally specifies $\nabla^U\mathbb{D}^U$ for $\mathbb{D}^U\subseteq \SepCaseType$ and $\nabla^U$ defined as with any other $R$ in place of $U$.
    
\end{itemize}
The rest of the construction works without any further changes.
The proof of its correctness is the same, with Claim~\ref{cl:construction claim} additionally asserting that $H'$ is $@$-consistent or $U$-consistent, respectively.

\section{Separators in Graded Hybrid Logics}\label{sec:graded}

We extend our hybrid logics with \emph{graded modalities} $\langle \bowtie n\  R\rangle\varphi$ for ${\bowtie}\in\{\leq, \geq\}$ and $n\geq 0$ with the following semantics: 
\begin{align*}
\Mmc,w & \models \langle \bowtie n\ R\rangle \phi &&\text{if\ \ \ \  }|\{v\mid (w,v)\in R^\Mmc,\Mmc,v\models\phi\}|\bowtie n.
\end{align*}
%
% \[\Mmc,w\models \langle \bowtie n\ R\rangle \phi \quad \text{ if }|\{v\mid (w,v)\in R^\Mmc,\Mmc,v\models\phi\}|\bowtie n.\]
We refer with $\Gmc,\Gat,\Gfull$ to the extensions of 
$\Hmc,\Hat,\Hfull$, respectively, with these modalities, and set $\GL=\{\Gmc,\Gat,\Gfull\}$. 
The combination of nominals and graded modalities in the logics of 
\GL is widely used in description logic research and practice.
It underpins standards such as OWL and is supported by most reasoning systems~\cite{DL-Textbook}. Following standard description logic conventions, in the case of 
\Gfull we only allow for ungraded use of $U$, so for example~$\mydiam{U}\phi$ is allowed but $\langle \leq n\ U\rangle \phi$ is not.

As before, separator existence can be model-theoretically characterized using the appropriate  bisimulations~\cite{DBLP:journals/corr/abs-1910-00039,DBLP:journals/sLogica/Rijke00}. Let $\Lmc\in \GL$ and $\Lmc'$ the corresponding language in \HL. Then $\Lmc(\sigma)$-bisimulations are defined just as $\Lmc'(\sigma)$-bisimulations with conditions \emph{(back)} and \emph{(forth)} replaced by the following:  
% %
% \begin{definition}\label{def:bisi}
% Let $\M,\N$ be models, and $\sigma$ a signature. A \emph{$\Gmc(\sigma)$-bisimulation} between \M and \N 
% is a relation $Z\subseteq \dom(\M)\times \dom(\N)$ such that 
% every $(w,w')\in Z$ \emph{(atom)} and \emph{(nom)} from $\Hmc(\sigma)$-bisimulations and the following variants of the \emph{(back)} and \emph{(forth)} conditions for all $R\in \sigma\cap\Rel$: 
\begin{description}%[\IEEEsetlabelwidth{abc}]

    \item[(gforth)]  
     for every finite $S\subseteq \{v\mid(w,v)\in R^\M\}$, there is $S'\subseteq \{v'\mid (w',v')\in R^\N\}$ and a bijection $\pi:S\to S'$ such that $(v,\pi(v))\in Z$ for all $v\in S$, for all $R\in\sigma\cap\Rel$;
    
    \item[(gback)] for every finite $S'\subseteq \{v'\mid(w',v')\in R^\N\}$, there is $S\subseteq \{v\mid (w,v)\in R^\M\}$ and a bijection $\pi:S\to S'$ such that $(v,\pi(v))\in Z$ for all $v\in S$, for all $R\in\sigma\cap\Rel$;
\end{description}    
%
% A $\Gmc_@(\sigma)$-bisimulation is a $\Gmc(\sigma)$-bisimulation $Z$ that additionally satisfies Condition~\emph{(@)}, and a $\Gmc_U(\sigma)$-bisimulation is a $\Gmc(\sigma)$-bisimulation for which \emph{(gback)} and \emph{(gforth)} also hold for universal modality~$U$.
% 
% For $\Lmc\in \GL$, we write $\M,w\sim_{\Lmc,\sigma} \N,w'$ if there is an $\Lmc(\sigma)$-bisimulation $Z$ between \M and \N such that $(w,w')\in Z$. 
% \end{definition}
%
% Note that an $\Lmc(\sigma)$-bisimulation is surjective and inverse surjective / total whenever $U\in \sigma$. 
% We call formulae $\varphi,\psi\in\Lmc$ \emph{jointly consistent under $\Lmc(\sigma)$-bisimulations} if there are pointed models $\M,w$ and $\N,w'$ with $\M,w\models\varphi$, $\N,w'\models \psi$, and $\M,w\sim_{\Lmc,\sigma} \N,w'$.
The characterization of separator existence in Lemma~\ref{lem:joint-consistent} for \HL extends to all logics in $\GL$.
\begin{lemma}\label{lem:joint-consistent-graded}
Let $\Lmc\in\GL$. Then, for every $\varphi_1,\varphi_2\in \Lmc$ and every signature $\sigma$, the following are equivalent:
  \begin{enumerate}
  
      \item there is an $\Lmc(\sigma)$-separator for $\varphi_1,\varphi_2$;
      
      \item $\varphi_1,\varphi_2$ are not jointly $\Lmc(\sigma)$-consistent.
      
  \end{enumerate}
\end{lemma}

Based on this lemma, one can adapt the techniques introduced in~\cite{DBLP:journals/tocl/ArtaleJMOW23} for \HL (as discussed at the beginning of Section~\ref{sec:hyper}) to show that
separator existence is in \TwoExpTime for $\Gfull$ and in \coNExpTime for \Gmc and \Gat.
The witness conditions for mosaic elimination can be obtained from the witness conditions developed below for hypermosaics. Matching lower bounds can also be obtained by adapting the lower bound proofs in~\cite{DBLP:journals/tocl/ArtaleJMOW23}. In fact, the \TwoExpTime lower bound carries over directly, since the $\Hfull(\sigma)$-bisimulations used in the reduction are already $\Gfull(\sigma)$-bisimulations (see also the size lower bound proof below). Slightly more care is required for the \coNExpTime lower bound, since the $\Hmc(\sigma)$-bisimulation depicted in~\cite[Figure~10]{DBLP:journals/tocl/ArtaleJMOW23}
is not a $\Gmc(\sigma)$-bisimulation. It can, however, be transformed into one by taking a binary tree also on the left-hand side and then adding an additional layer of leaves: a single $R$-successor interpreting nominal $a$ on the left, and a fresh $R$-successor for each leaf on the right.

Unfortunately, as before, the upper bound proofs are not constructive. Hence, we extend the hypermosaic elimination procedure from the previous section to the logics in \GL. 

\subsection{Hypermosaic Elimination}

As in the ungraded case, we focus on the procedure for the basic case without $@$ or $U$. The adaptations needed for these additional features are the same as in the ungraded case and will not be discussed in any detail again. 

Let $\phi_1,\phi_2$ be input formulae, and $\sigma$ be a signature. The notions of set of subformulae $\sub(\phi_1,\phi_2)$, ($\sub(\phi_1,\phi_2)$-)types, mosaics, hypermosaics, extensions of hypermosaics, and realizability in bisimilar models are all defined as before. 
To define the elimination condition, we need to adapt the notion of a witness. 

Let $H,H'$ be hypermosaics, and let $\mos\in H$ and $R\in\Rel$. Then $H'$ is an \emph{$R$-witness for $\mos$} if there is a function
\[
\Rmc:\TP\times\{1,2\}\times H'\times\TP\to\mathbb{N}\cup\{\infty\}
\]
that satisfies the following conditions:
\begin{enumerate}[label=(\roman*)]
    
    \item\label{it:R-wit 1} For all $t,s\in\TP$, $i\in\{1,2\}$, $\mos'\in H'$, if $\Rmc(t,i,\mos',s)>0$, then $t\in \mos_i$ and $s\in \mos'_i$;

    \item\label{it:R-wit 2} For all $t,s\in\TP$, $i\in\{1,2\}$, $\mos'\in H'$, if $a\in s$ for some nominal $a$ then $\Rmc(t,i,\mos',s)\leq 1$;
    
    \item\label{it:R-wit 3} For all $i\in\{1,2\}$, $t\in \mos_i$, and $\langle \bowtie n\ R\rangle\chi\in\sub(\phi_1,\phi_2)$:
    \[
    \langle \bowtie n\ R\rangle\chi\in t\quad \text{ iff } \quad \sum_{\substack{\mos'\in H'\\ \chi\in s\in \TP}} \Rmc(t,i,\mos',s)\bowtie n.
    \]
    
    \item\label{it:R-wit 4} If $R\in\sigma$, then for all $\mos'\in H'$, all $i,j\in\{1,2\}$, and all $t\in \mos_i,t'\in \mos_j$:
    \[
    \sum_{s\in \mos'_i}\Rmc(t,i,\mos',s)=\sum_{s'\in \mos'_j}\Rmc(t',j,\mos',s').
    \]

\end{enumerate}
The equations employ the standard extension of arithmetic to $\mathbb{N}\cup\{\infty\}$ with $\infty+x=x+\infty=\infty$ for every $x$.
We call a function $\Rmc$ as above an \emph{$R$-counting function for $\mos,H'$}.
Intuitively, $\Rmc(t,i,\mos',s)$ says how many copies of the type $s$ from within mosaic $\mos'$ are $R$-successors of the type $t$ on the side $i$.
Condition~\ref{it:R-wit 1} states that only types from $\mos$ and copies of types from the respective side of $\mos'$ are linked.
Condition~\ref{it:R-wit 2} states that nominals are never copied more than once. Condition~\ref{it:R-wit 3} says that all constraints imposed by graded modalities are respected.
Condition~\ref{it:R-wit 4} ensures that the picked successor types can be realized in a bisimilar way in $H'$.

We say that $H'$ is a \emph{witness} for $\mos\in H$ if $H'$ extends $H$ and is an $R$-witness for $\mos$ for every $R\in\Rel$. We now give the elimination procedure for $\Gmc$.

\emph{Elimination procedure for $\Gmc$}. Let $\Zmc_0$ denote the set of all nominal-clean and $\sigma$-consistent hypermosaics. We obtain $\Zmc_{i+1}$ from $\Zmc_{i}$ by dropping all $H$ that contain a mosaic $\mos$ that does not have a witness in $\Zmc_{i}$. Let $\Zmc^{\ast}$ be the set of hypermosaics where this stabilizes. We have a graded variant of Theorem~\ref{thm:correctness-elimination}.

\begin{restatable}{theorem}{thmcorrectnesseliminationgraded}
\label{thm:correctness-elimination-graded}
For a hypermosaic $H$, the following are equivalent: 
\begin{itemize}
    \item $H$ is realizable by a $\Gmc(\sigma)$-bisimulation;
    \item $H\in \Zmc^*$.
\end{itemize}
\end{restatable}

\begin{proof}(Sketch)
The direction from top to bottom is almost the same as without grades. Given a realized $H$ we get its appropriate extension $H'\in\Zmc^*$ exactly as before. The only difference is that proving that $H'$ witnesses each $\mos\in H'$ requires quantitative reasoning.
% \nb{JK: some more in appendix?}

The opposite direction is also similar. We take $H^*\in\Zmc^*$ as before and construct models $\M_1,\M_2$ and a bisimulation $Z$. The difference is that now the domains $W^{\M_{i}}$ contain multiple copies of each nominal-free type in each mosaic:
\begin{align*}
W^{\M_{i}}=~& \{(t,j,\mos)\mid t\in \mos_{i}\text{ nom-free},\mos\in H^*, j\in\mathbb{N}\setminus\{0\}\} \cup{} \\
& \{(t,1,\mos)\mid t\in \mos_{i}\text{ has nom},\mos\in H^*\}\ \ \cup\ \ \text{Nom}_{i}
\end{align*}
with $i\in \{1,2\}$ and Nom$_{i}$ the nominals absent on the $i$-th side.
% % We next define $R^{\M_{i}}$ for relation symbols $R\in\sig(\phi_1)\cup\sig(\phi_2)$. $H^*$ is a witness for every $\mos\in H^*$. Let us fix, for every $\mos\in H^*$, an $R$-counting function $\Rmc_\mos$. We set:
% % \begin{gather*}
% % (t,j,\mos)\ R^{\Mmc_i}\ (t',j',\mos') \hspace{0.5cm}
% % \iff \hspace{0.5cm}
% % \Rmc_\mos(t,i,\mos',t')\geq j'.
% % \end{gather*}
% % The elements of $\Nom_i$ are not linked by $R^{\Mmc_i}$ with any element.

% Towards the definition of $R^{\M_{i}}$, for every $\mos\in H^*$ fix an $R$-counting function $\Rmc_\mos$ for $\mos,H^*$. We set:
% \begin{gather*}
%     (t,j,\mos)\ R^{\Mmc_i}\ (t',j',\mos') \hspace{0.5cm}
%     \iff \hspace{0.5cm}
%     \Rmc_\mos(t,i,\mos',t')\geq j'.
% \end{gather*}
% The elements of $\Nom_i$ are not linked by $R^{\Mmc_i}$ with any element.
The relation $R^{\M_{i}}$ links a point to an appropriate number of these copies, depending on $R$-counting functions for mosaics in $H^*$. The rest is as before: propositional variables and nominals are defined based on the type component of a point, and the bisimulation $Z$ links points with the same mosaic component.
\end{proof}
We obtain the following consequence of Theorem~\ref{thm:correctness-elimination-graded}. 
\begin{corollary}\label{cor:hypermosaics-graded}
    $\varphi_1$ and $\varphi_2$ are jointly $\Gmc(\sigma)$-consistent iff there exists a hypermosaic $\{(\{t_{1}\},\{t_{2}\})\}\in\Zmc^*$ with $\varphi_1\in t_{1}$ and $\varphi_2\in t_{2}$.
\end{corollary}
% 
% \begin{lemma}
%     Let $H,H'$ be hypermosaics, $\mos\in H$ and $R\in\Rel$. The following conditions are equivalent:
%     \begin{enumerate}
%         \item there is a witnessing function $\Rmc$ for $R$ into $\mathbb{N}\cup\{\infty\}$;
%         \item there is a witnessing function $\Rmc$ for $R$ bounded by $\kappa 2^{2(|\varphi_{1}|+|\varphi_{2}|)}$;
%         \item there is a function $\Rmc$ satisfying the conditions of a witnessing function for $R$
%         with Condition~\ref{it:R-wit 4} weakened as follows: if $R\in\sigma$, then for all $\mos'\in H'$, all $i,j\in\{1,2\}$, and all $t\in \mos_i,t'\in \mos_j$:
%     \[
%     \sum_{s\in \mos'_i}\Rmc(t,i,\mos',s)=\sum_{s'\in \mos'_j}\Rmc(t',j,\mos',s')
%     \]
%     if $\sum_{s\in \mos'_i}\Rmc(t,i,\mos',s)\leq \kappa 2^{|\varphi_{1}|+|\varphi_{2}|}$ or
%     $\sum_{s'\in \mos'_j}\Rmc(t',j,\mos',s')\leq \kappa 2^{|\varphi_{1}|+|\varphi_{2}|}$.
%     \end{enumerate}
% \end{lemma}

The correctness proof for interpolant construction below will rely on the observation that values of $R$-counting functions only matter up to a threshold. Let $\kappa$ be the greatest grade appearing in $\phi_1,\phi_2$ and denote $\lambda=\kappa\times|\TP|$. Consider the equivalence relation $=^\lambda$ on $\mathbb{N}\cup\{\infty\}$ such that $x=^\lambda y$ iff either $x,y\leq\lambda$ and $x=y$, or $x,y>\lambda$. A \emph{weak $R$-counting function} is defined in the same way as an $R$-counting function, except that equality $=$ in the last condition~\ref{it:R-wit 4} is replaced by $=^\lambda$. The following lemma provides bounds for the values in (weak) $R$-counting functions. 

% To analyze the complexity of the construction, we also give bounds for the $R$-counting function.  

\begin{restatable}{lemma}{lemRcountingweak}
\label{lem:R-counting weak}
    Assume a mosaic $\mos$ and a hypermosaic $H'$. 
\begin{enumerate}
\item $H',\mos$ admit an $R$-counting function iff they admit a weak $R$-counting function. 
\item Moreover, if either condition holds for $H',\mos$, then they admit an $R$-counting function 
with codomain $\{0,\ldots,\lambda \times |\TP|\}$ and a weak $R$-counting function with codomain $\{0,\ldots,\lambda\}$. 
%respectively.\nb{J: is this a respectively or just conjunction?}
\end{enumerate}
\end{restatable}

Observe that it follows from the second part of Lemma~\ref{lem:R-counting weak} that it can be decided in triply exponential time (in the size of $\varphi_{1},\varphi_{2}$) whether a hypermosaic $H'$ is an $R$-witness for a mosaic $\mos$. Hence the elimination procedure runs in triply exponential time overall. An elimination procedure that eliminates mosaics rather than hypermosaics that runs in doubly exponential time and decides separator existence is easily derived from the procedure above.

%The hypermosaic elimination procedure developed above can be adapted to $\Gmc_@$ and $\Gfull$ in precisely the same way  that the elimination procedure for $\Hmc$ was adapted to $\Hmc_{@}$ and $\Hfull$, respectively. Note that this works for $\Hfull$ since we only allow for ungraded use of $U$.
\subsection{Separator Construction with Grades}

In this subsection we show how to construct separators in the graded case. Our construction follows the ideas from the grades-free setting. The notions of a case distinction and a hyperseparator are defined as before, only with $\Gmc$ in place of $\Hmc$. Lemma~\ref{lem:hypersep2sep} remains valid as it does not depend on the logic, but only on the properties of a hyperseparator. We inductively prove the following graded variant of Lemma~\ref{lem:sep construction}.

\begin{lemma}\label{lem:sep graded construction}
For every round $\ell$ of the elimination procedure for $\Gmc$, we can compute a hyperseparator $\Case,\sep_1,\sep_2$ that separates every hypermosaic eliminated up to that round.
\end{lemma}
\begin{proof}
We prove the lemma by induction. The core argument is the same as in the grades-free case. The base case is exactly the same. Let us deal with the inductive case.

Assume a hyperseparator $\Case,\sep_1,\sep_2$ separating all hypermosaics eliminated up to round $\ell$. We construct $\Case^+,\sep^+_1,\sep^+_2$ for all hypermosaics eliminated up to round $\ell+1$. As with $\Hmc$, $\Case^+$ is the result of refining $\Case$: we split every case $\cmf\in\Case$ by determining the $\SepCaseType$ of every nominal on each side.

The key modification needed for $\Gmc$ is in the definition of a star-type. As with $\Hmc$, a \emph{$\Sep_\cmf$-star-type} of a given point $v$ is the full description of $v$ and its neighbors with respect to $\SepCaseType$. What is different is that \emph{full description} is now quantitative: for every $R$ and $d\in\SepCaseType$ we specify exactly (up to $\lambda$) how many $R$-children of $v$ satisfy $d$. Formally, the $\SepCaseStarType$ is the set 
\[
%\SepCaseStarType =
\left\{
  d\wedge\hspace{-0.4cm}\bigwedge_{\substack{R\in\sigma\\ {d'}\in\SepCaseType}}\hspace{-0.3cm}\mydiam{= f_R(d')\ R} d'\ \;\middle|\;
  \begin{aligned}
  d\in\SepCaseType,
  f_R:\SepCaseType\to\{0,...,\lambda\}
  \end{aligned}\right\}
\]
% \[
% \OneSepCaseType = \{d\wedge\bigwedge_{\substack{R\in\sigma\\ {d'}\in\SepCaseType}}\mydiam{= f_R(d')\ R} d'\ |\ d\in\SepCaseType,\ f_R:\SepCaseType\to\{0,...,\kappa\}\}
% \]
where
\[
\mydiam{= n\ R} \theta = 
\begin{cases}
    \mydiam{\leq n\ R} \theta \wedge \mydiam{\geq n\ R} \theta & \text{ if } n<\lambda\\
    \mydiam{\geq n\ R} \theta & \text{ otherwise.}
\end{cases}
\]

As with $\Hmc$, for every type $t$ and side $i\in \{1,2\}$ we define:
\[
\sep^+_i(\cmf^+,t) = \bigvee
\left\{
  \delta \in\SepCaseStarType\mid
  \begin{aligned}
  %& \delta\in\OneSepCaseType\\
  %&
  \case(\cmf^+_i)\wedge\delta\wedge t\not\models\bot
  \end{aligned}
\right\}.
\]
The proof that $\Case^+,\sep^+_1,\sep^+_2$ is a hyperseparator is the same as with $\Hmc$. The proof that it separates every $H$ eliminated up to round $\ell+1$ goes along the same lines, which we outline now. Assume the hyperseparator does not separate $H$. This means that there is a case $\cmf^+\in\Case^+$ such that
for every $\mos\in H$ there is a star-type $\delta_\mos$ which is a disjunct of all $\sep^+_i(\cmf^+,t)$ with $i\in \{1,2\}$ and $t\in\mos_i$. We define a mosaic $H'$ by setting 
\[
\mos^d_i=\{t\in\TP\ |\ \case(\cmf^+_i)\wedge t\wedge d\not\models\bot\}.
\]
and:
\[
H' = \{(\mos^d_1,\mos^d_2)\ |\ d\in\SepCaseType\},
\]
all as in the ungraded case, and prove a graded variant of Claim~\ref{cl:construction claim}.
\begin{restatable}{claim}{clconstructiongradedclaim}
\label{cl:construction graded claim}
The following hold:
\begin{enumerate}
  \item $H'$ is nominal clean.
  \item $H'$ is a witness for every $\mos\in H$.
\end{enumerate}
\end{restatable}

\begin{proof}(Sketch)
The first item, and the proof that $H'$ extends $H$ are easy, as in the grades-free setting. We show that for every $R$ and $\mos\in H$, $H'$ is an $R$-witness for $\mos$. Fix $R$ and $\mos\in H$. We need to construct an $R$-counting function \Rmc. By Lemma~\ref{lem:R-counting weak}, it is enough to construct a \emph{weak} $R$-counting function, which we do now. 

Assume $i\in \{1,2\}$ and $t\in\TP$. If $t\notin \mos_i$ then $\Rmc(t,i,\mos',s)=0$ for every $\mos'$ and $s$. In the remaining case fix a model
\[
\M,w\models\case(\cmf^+_i)\wedge\delta_\mos\wedge t
\]
and denote by $V$ the set of $R$-children of the root. We set:
\[
\Rmc(t,i,\mos',s) = 
\left|
\left\{
  v\in V\ \;\middle|\;
  \begin{aligned}
  d\in\SepCaseType,\ \mos'=(\mos_1^d,\mos_2^d),\ \M,v\models s\wedge d
  \end{aligned}
\right\}
\right|
\]
for all $\mos',s$. This completes the construction of $\Rmc(t,i,\ \cdot\ ,\ \cdot\ )$. Running it for every possible values of $t,i$ we construct the entire $\Rmc$.
Checking that \Rmc is a weak $R$-counting function boils down to checking conditions~\ref{it:R-wit 1}, \ref{it:R-wit 2}, and~\ref{it:R-wit 3}, and the weak variant of~\ref{it:R-wit 4}. This is not very difficult, but it is worth to point out that weakness of the last condition~\ref{it:R-wit 4} is essential because the constructed formulae $\sep^+_i(\cmf^+,t)$ only count up to the threshold $\lambda$.
\end{proof}

Lemma~\ref{lem:sep graded construction} follows from Claim~\ref{cl:construction graded claim} exactly the same way as Lemma~\ref{lem:sep construction} from Claim~\ref{cl:construction claim}.
\end{proof}

\begin{restatable}{theorem}{thmsumminguptwo}\label{thm:summing-up-two}
  If $\varphi_1,\varphi_2$ are not jointly $\Gmc(\sigma)$-consistent, then we can compute in fourfold exponential time a $\Gmc(\sigma)$-separator (of fourfold exponential size) for $\varphi_1,\varphi_2$.
\end{restatable}
\begin{proof}
   If $\varphi_1,\varphi_2$ are not jointly $\Gmc(\sigma)$-consistent, we can construct separators based on Corollary~\ref{cor:hypermosaics-graded} and Lemma~\ref{lem:hypersep2sep}. We have already argued that the elimination runs in triply exponential time. The only difference in the size analysis compared to Theorem~\ref{thm:summing-up} is in the number and size of $\Sep_\cmf$-star-types, which depend on $\lambda$ now and are in $O(2^{2^n}\times {\lambda^{2^{2^n}}})$ and $O(2^n\times m_\ell+\lambda\times 2^{2^n}\times 2^n\times m_\ell)$, respectively. Since $\lambda\leq n\times 2^n$, the rest of the analysis remains the same.
\end{proof}

% {\color{red} TODO briefly about $@$ and $U$}

\section{Separator Size: Lower Bound}\label{sec:lowerbounds}
In this section, we establish lower bounds on the size of separators in $\Hfull$ and $\Gfull$. 

\begin{restatable}{theorem}{thmlowerbound} \label{thm:size-lower-bound} Let $\Lmc\in\{\Hfull,\Gfull\}$. There is a signature $\sigma$ such that
% Let $\sigma = \{R,S,r\}$ for $R,S\in\Rel$ and $r\in \Var$.
for every $n\in\mathbb{N}$ there are $\phi_1,\phi_2\in\Lmc$ of size polynomial in $n$ which admit an $\Lmc(\sigma)$-separator, but every such separator has size at least triply exponential in $n$ and DAG-size doubly exponential in $n$.
\end{restatable}

\begin{proof}
The proof builds on ideas used to establish \TwoExpTime-hardness of deciding separator existence in \Hfull~\cite{DBLP:journals/tocl/ArtaleJMOW23}. We focus on $\Hfull$ but almost the same proof applies to \Gfull. 

Set $\sigma=\{R,S,r\}$ with $R,S\in\Rel$, $r\in\Var$, and let $n\in\mathbb{N}$.
We define jointly $\Hfull(\sigma)$-inconsistent formulae
$\varphi_1,\varphi_2$ which only have separators of large modal depth, that is, nesting of the modal operators. 
The main work is done by $\varphi_1$; we describe its behavior by describing its class of models up to $\Hmc(\sigma)$-bisimulations. A pointed model $\M,w$ is an \emph{$SR$-tree} if it is a tree that has an infinite $S$-path starting in the root, no $R$-edge is followed by an $S$-edge, and every point has an $R$-successor. We associate with each $R$-path $\pi$ of length $2^n$ in \Mmc the binary number $r(\pi)<2^{2^n}$ encoded by the valuation of variable $r\in\sigma$ along $\pi$.
We call an $SR$-tree \emph{good} if for every path $\pi_{SR}=\pi_S\pi_R$ starting at the root consisting of an $S$-part $\pi_S$ of length $i<2^n$, and an $R$-part $\pi_R$, and for every $u\in \pi_R$ whose position in $\pi_R$ is a multiple of $2^n$ we have:
\begin{enumerate}

    \item For every $R$-path $\pi\pi'$ starting in $u$ with $|\pi|=|\pi'|=2^n$, $r(\pi)=x$, $r(\pi')=x'$ and $y=x+1$, the values $x'$ and $y$ agree on the $(i+1)$-th bit.

    \item No $R$-path $\pi$ of length $2^n$ starting in $u$ has $r(\pi)= 2^{2^n}-1$.

\end{enumerate}
Intuitively, Property~(1) synchronizes a single bit in values of consecutive $R$-paths, and Property~(2) expresses that the maximal value $2^{2^n}-1$ is never reached. 

The class of all good $SR$-trees is denoted $\mathfrak{M}$. It is standard to give, using the universal modality and variables not in $\sigma$ to count modulo $2^n$, a formula $\varphi_1\in\Hfull$ of size polynomial in $n$ such that:
\begin{enumerate}[label=(\roman*)]
    \item If $\M,w\models\phi_1$ then $\M,w\sim_{\Hmc,\sigma}\N,v$ for some $\N,v\in\mathfrak{M}$.
    \item Every $\M,w\in\mathfrak{M}$ can be made a model of $\phi_1$ by changing the valuation of variables not in $\sigma$.
\end{enumerate}
We claim that $\varphi_1$ and $\varphi_2=a\wedge[S]a$, for a nominal $a$, are jointly $\Hmc(\sigma)$-inconsistent. With the aim of showing a contradiction, let $\M,w$ be a model of $\varphi_2$ which is $\Hmc(\sigma)$-bisimilar to a model of $\varphi_1$, and due to Point~(i), also $\Hmc(\sigma)$-bisimilar to a good $SR$-tree $\N,v$. Then, $a^\M$ is $\Hmc(\sigma)$-bisimilar to all elements on the $S$-path starting in $v$ and roots an infinite $R$-path $\pi_R$ such that every $u\in \pi_R$ whose position in $\pi_R$ is a multiple of $2^n$ satisfies Property~(1) \emph{for every} $i<2^n$. Hence, any $\pi\pi'$ starting in such $u$ satisfies $r(\pi')=r(\pi)+1$ modulo $2^{2^n}$. This results in a perfectly synchronized counter along $\pi_R$, in contradiction to Property~(2). 

Suppose now that $\varphi_1,\varphi_2$ are separated by $\theta\in\Hfull(\sigma)$ with modal depth $d<2^{2^n}$. In order to show a contradiction for (the stronger assumption) $\theta\in\Hmc(\sigma)$, consider the model $\N,v$ of $\varphi_2$ which consists of an $S$-loop around $v=a^\N$ and an infinite $R$-path starting in $v$ along which the $r$-counter starts in $0$ and is correctly incremented throughout. Consider its unraveling $\N^v,v$ into a tree. It is not a good $SR$-tree, but one can turn it into a good $SR$-tree $\M^0,v$ by modifying only the valuation of $r$ in positions of depth $>d$. By Property~(ii), this $\M^0,v$ can be turned into a model $\M,v$ of $\varphi_1$ by changing non-$\sigma$ symbols. We have $\M,v\models\theta$. Since $\N,v\models \neg \theta$, also $\N^v,v\models\neg\theta$. But this is a contradiction since the respective reducts of $\M,v$ and $\N^v,v$ to $\sigma$ are isomorphic on the first $d$-levels, and $\theta$ has modal depth $d$.
For $\theta\in\Hfull(\sigma)$, one has to be a bit more careful, but the argument is similar. We note that all this works for $\theta\in \Gfull(\sigma)$ as well since $R,S$ are interpreted as total functions in $\N,\M$.

Clearly, a doubly exponential lower bound on the modal depth of separators implies the same lower bound for formula and DAG-size. To achieve a triply exponential lower bound for formula size, we replace $R$ with a union of two relations $R_1,R_2$. 
\end{proof}
\section{Uniform Interpolants}
\label{sec:uniform}
We show that uniform interpolants behave very differently from Craig interpolants by proving that the existence of uniform interpolants is undecidable for all fragments of FO containing $\Hmc$, and hence for all languages considered in this article.  %
\begin{definition}
    Let $\Lmc_{1},\Lmc_{2}$ be fragments of FO and $\varphi$ a formula in FO. Let $\sigma$
    be a signature. Then a formula $\chi$ is called a \emph{uniform $\Lmc_{1}(\sigma)$-interpolant for $\varphi$ in $\Lmc_{2}$} if $\chi \in \Lmc_{2}$, $\sig(\chi)\subseteq \sigma$, and
\begin{itemize}
    \item $\varphi\models\chi$ and
    \item if $\varphi\models \psi$ and $\psi\in \Lmc_{1}$ and $\sig(\psi)\cap \sig(\varphi)\subseteq \sigma$, then $\chi \models \psi$.
\end{itemize}
\end{definition}

Typically, one assumes that $\varphi\in \Lmc_{1}=\Lmc_{2}$. The \emph{uniform interpolant existence problem for $\Lmc$} is the problem to decide for $\varphi\in \Lmc$ and signature $\sigma$ whether a uniform $\Lmc(\sigma)$-interpolant for $\varphi$ in $\Lmc$ exists. Observe that for logics with the uniform interpolation property (such as basic modal logic, the modal $\mu$-calculus, and propositional intuitionistic logic~\cite{visser1996uniform,DBLP:conf/aiml/DAgostinoH96,DBLP:journals/jsyml/Pitts92}), the uniform interpolant existence problem is trivially decidable.
\begin{restatable}{theorem}{thmundecunform}
\label{thm:undecunform} 
Let $\Lmc_{1},\Lmc_{2}$ be fragments of FO containing $\Hmc$. Then it is undecidable whether for a formula $\varphi\in \Hmc$ and signature $\sigma$ there exists a uniform $\Lmc_{1}(\sigma)$-interpolant for $\varphi$ in $\Lmc_{2}$.

In particular, for any fragment $\Lmc$ of FO containing $\Hmc$, the uniform interpolant existence problem for $\Lmc$ is undecidable.
\end{restatable}
The proof uses the fact that in $\Hmc$ one can use nominals to describe finite structures almost up to isomorphism. To illustrate the idea, consider the formula $\phi_{even}$ in $\Hmc$ defined as the conjunction of 
\begin{align*}
\mydiam{E}(\mathsf{min} \wedge \mathsf{even}),\quad
    \mybox{E}(\mathsf{even} \rightarrow \mybox{S}\neg \mathsf{even}), \quad \mybox{E}(\neg \mathsf{even} \rightarrow \mybox{S}\mathsf{even}),\\
    \mybox{E}(\mybox{S}\bot \rightarrow \neg \mathsf{even}), \quad \mybox{E}\mybox{S}\neg\mathsf{min}
\end{align*}
There does not exist a uniform $\Hmc(\{E,S,\mathsf{min}\})$-interpolant for $\varphi_{even}$ in FO because such an
interpolant would provide a definition of finite linear orders of even length (with $S$ being the successor relation).
It is known that such a definition does not exist~\cite[Example~2.3.6]{DBLP:books/daglib/0082516}. To show this, consider $\mathcal{M}_{n}=(W_{n},S^{\Mmc_{n}},E^{\Mmc_{n}},\mathsf{min}^{\Mmc_{n}})$ with $W_{n}=\{a,a_{0},\ldots,a_{n}\}$, $(a_{i},a_{i+1})\in S^{\Mmc_{n}}$ for $i<n$, $(a,a_{i})\in E^{\Mmc_{n}}$ for $i\leq n$,
and $\mathsf{min}^{\Mmc_{n}}=\{a_{0}\}$. Regard $a_{0},\ldots,a_{n}$ as nominals and describe $\Mmc_{n},a$
using the conjunction $\chi_{n}\in \Hmc$ of the following formulae:
\begin{align*}
     \mydiam{E}(\mathsf{min}\wedge a_{0}) \wedge \bigwedge_{i=1}^{n}\mydiam{E}a_{i},\hspace{0.5cm}
\allowdisplaybreaks
     \mybox{E} \bigvee_{i=0}^{n}a_{i} \wedge [E]\bigwedge_{i\not=j}\neg (a_{i} \wedge a_{j}),\\
     \mybox{E}\bigwedge_{i=0}^{n-1}(a_{i}\rightarrow \mydiam{S}a_{i+1}),
     \quad \mybox{E}\bigwedge_{i=0}^{n-1}(a_{i} \rightarrow [S]a_{i+1}),\quad 
     \mybox{E}(a_{n}\rightarrow [S]\bot).
\end{align*}
Then $\Mmc_{n},a\models\chi_{n}$ and one can easily show that $\phi_{even}\models \neg \chi_{n}$ iff $n$ is an even number. Hence a uniform $\Hmc(\{E,S,\mathsf{min}\})$-interpolant $\chi$ would have the property that $\chi\models \neg \chi_{n}$ iff $n$ is an even number. This contradicts non-FO-definability of finite linear orders of even length. 

The proof technique introduced above is extended in a rather straightforward way in the full version on arXiv to provide a reduction of an undecidable finite tiling problem to uniform interpolant existence, thus proving Theorem~\ref{thm:undecunform}. In the reduction, one encodes $n\times m$-grids in the same way as linear orderings are encoded above.

\section{Discussion}
We have established the first elementary upper bounds for computing Craig interpolants in a family of logics for which the CIP cannot be readily recovered. A central open question is to determine the exact complexity of the problem, and thus to close the gap between our fourfold exponential upper bounds and threefold exponential lower bounds. Beyond these, our results prompt several broader research questions. We discuss the robustness of our interpolant construction method when the input languages are extended and when the underlying notion of interpolant is modified.

% {\bf Extending the Input Language}. 
\subparagraph*{Extending the Input Language.}
Many different constraints on accessibility relations have been considered in both modal and description logic~\cite{areces200714,handbookDL}. The extension of our results to role inclusions constraints, $R\subseteq S$, is straightforward. Also, for logics in $\HL$ our approach extends in a straightforward way to fundamental constraints such as transitivity and reflexivity. For instance, to deal with transitive relations, it suffices to strengthen the definition of $t \rightsquigarrow_{R} t'$ by demanding that also $\mydiam{R}\chi \in t'$ implies $\mydiam{R}\chi\in t$ and to observe that a bisimulation along a relation $R$ is also a bisimulation along the transitive closure of $R$. On the other hand, more substantial adaptations appear to be necessary if one adds inverse relations to the language. Whether this is possible at all remains open. Further examples of logics lacking CIP for which the interpolant existence problem is decidable but no elementary bound on the size of Craig interpolants is known include the guarded and two-variable fragments of FO, weak K4, and various fragments of first-order modal logics~\cite{DBLP:conf/lics/JungW21,DBLP:conf/aiml/KuruczWZ24,DBLP:journals/lmcs/KuruczWZ25}. An orthogonal, and likely challenging question is whether one can decide existence of interpolants in the \emph{hybrid $\mu$-calculus}, that is, the extension of $\Hmc$ with fixed points, and if so, how to compute them. Note that, while the $\mu$-calculus enjoys the CIP (it even supports uniform interpolation~\cite{DBLP:conf/aiml/DAgostinoH96}), the formula $(a \wedge \mydiam{R}a) \rightarrow (b \rightarrow \mydiam{R} b)$ from the introduction shows that the hybrid $\mu$-calculus does not.

\subparagraph*{Modifying the Notion of Interpolant.}
% {\bf Modifying the Notion of Interpolant}. 
In many applications of interpolants in description logic, it is desirable to obtain interpolants in languages that are strictly weaker than the input language. For example, given graded input formulae one may seek grades-free interpolants, as e.g. in~\cite{DBLP:conf/lics/KuijerTWZ25}. Our method is highly robust under this kind of restriction. Indeed, for any $\Lmc_{1}\subseteq \Lmc_{2}$ with $\Lmc_{1},\Lmc_{2}\in \HL\cup\GL$, straightforward adaptations of our algorithms show that interpolants in $\Lmc_{1}$ 
between formulae in $\Lmc_{2}$ can be computed, whenever they exist, in fourfold exponential time. We note, however, that further extensions of these results may not be feasible. For instance, once inverse relations are added, the existence of interpolants in $\Hmc$ between formulae in $\Gmc$ becomes undecidable~\cite{DBLP:conf/lics/KuijerTWZ25}.

% Maybe not:
% \begin{itemize}
%    \item first order separators in graded case
% \end{itemize}

\clearpage

\bibliographystyle{plainurl}
\bibliography{main}

\clearpage
\appendix

\section{Proofs for Section~\ref{sec:first-order}}

\lemjointconsone*

\begin{proof} 
Assume $\mathsf{m}=(\mos_{1},\mos_{2})$ and $\mathcal{M}_1,w_{1}\sim_{\Hmc,\sigma}\mathcal{M}_{2},w_{2}$ with
    $\mathcal{M}_i,w_{i}\models \mos_{i} \wedge \case(\cmf_{i})$, for both $i\in \{1,2\}$. We show that $\mathsf{m}$ is not eliminated in the elimination procedure for $\cmf=(\cmf_{1},\cmf_{2})$.
    To this end, let 
\[
\mathcal{X}=\{(\tp_{\Mmc_{1}}(d_{1}),\tp_{\Mmc_{2}}(d_{2})) \mid \Mmc_{1},d_{1} \sim_{\Hmc,\sigma}\mathcal{M}_{2},d_{2}\}\\
\]
It is easy to see that $\mathcal{X}\subseteq \mathcal{Y}^{\ast}$ and so the claim follows.
    
\medskip

Conversely, we construct models $\Mmc_{1},\Mmc_{2}$ witnessing that $\mos_{1}\wedge \case(\cmf_{1}),\mos_{2}\wedge \case(\cmf_{2})$ are jointly $\Hmc(\sigma)$-consistent for every $\mos\in \mathcal{Y}^{\ast}$.
We construct models $\Mmc_{i}$ by merging two models $\Mmc_{i}'$ and $\Nmc_{i}$, for both $i\in \{1,2\}$. 

Observe that for both $i\in \{1,2\}$ and any $\mos_{i}$ with $\mos\in \Ymc^{\ast}$ and nominal $a$, we have $a\in \mos_{i}$ iff $\mos_{i}=t_{a}^{i}$.
We begin with constructing $\Mmc_{i}'$. Let $W^{\Mmc_{i}'} = \mathcal{Y}^{\ast}$ and set for all $\mos,\mos'\in W^{\Mmc_{i}'}$: 
\begin{itemize}
    \item $\mos\in p^{\Mmc_{i}'}$ if $p\in \mos_{i}$, for every propositional variable $p$,
    \item $\mos = a^{\Mmc_{i}'}$ if $t_{a}^{i}=\mos_{i}$, for every nominal $a$,
    \item $(\mathsf{m},\mathsf{m}') \in R^{\Mmc_{i}}$ if $\mathsf{m}\rightsquigarrow_{R}\mathsf{m}'$, for every $R\in \sigma$.
\end{itemize}
We next define $\Nmc_{i}$ as any model satisfying $\case(\cmf_{i})$ and containing for every $\mos\in \mathcal{Y}^{\ast}$ and $\mydiam{R}\chi\in\mos_{i}$ with $R\not\in \sigma$ for some type $t_{\chi}$ a world
$w_{t_{\chi}}$ satisfying $t_{\chi}$ with $\chi\in t_{\chi}$ and $\mathsf{m}_{i}\rightsquigarrow_{R}t_{\chi}$. 
Note that such a model exists since $\mos_{i}\wedge \case(\cmf_{i})$ is satisfiable for both $i\in \{1,2\}$
for every $\mos\in \mathcal{Y}^{\ast}$. 

Now merge $\Mmc_{i}'$ and $\Nmc_{i}$ by taking their disjoint union and then doing the following: (i) if $a$ is a nominal for which $a^{\Mmc_{i}'}$ is defined, then identify $a^{\Nmc_{i}}$ with $a^{\Mmc_{i}}$ and drop all pairs $(a^{\Nmc_{i}},d)$ from $R^{\Nmc_{i}}$ with $R\in\sigma$ and (ii) add all $(\mos,w_{t_{\chi}})$ with
$\mathsf{m}_{i}\rightsquigarrow_{R}t_{\chi}$ to $R^{\Mmc_{i}'}$ for $R\not\in\sigma$.

It is not difficult to show using the construction that for all $\mos\in \mathcal{Y}^{\ast}$:
\begin{itemize}
\item $\Mmc_{1},\mos\sim_{\Hmc,\sigma} \Mmc_{2},\mos$;
\item $\Mmc_{i},\mos\models \mos_{i}$, for both $i\in\{1,2\}$;
\item $\Mmc_{i},w\models t$ iff $\Nmc_{i},w\models t$ for all $w$ in $\Nmc_{i}$ and types $t$, for both $i\in\{1,2\}$.
\end{itemize}
It follows that $\mos_{1}\wedge \case(\cmf_{1}),\mos_{2}\wedge \case(\cmf_{2})$ are jointly $\Hmc(\sigma)$-consistent for every $\mos\in \mathcal{Y}^{\ast}$.
%It follows that $s_{1}\wedge \case(\bar{t}_{1}),s_{2}\wedge \case(\bar{t}_{2})$ are jointly $\Hmc(\sigma)$-consistent.
\end{proof}
%The following lemma is a direct consequence of Lemma~\ref{lem:jointconsone}.
%\begin{restatable}{lemma}{lemjointcons}
%\label{lem:jointcons} 
%    $\phi_1,\phi_2$ are jointly $\Hmc(\sigma)$-consistent iff there are $t_{1},t_{2}$
%    with $\phi_1\in t_{1}$ and $\phi_2\in t_{2}$ and basic nominal sequences $\overline{t}_{1}$ and $\overline{t}_{2}$ such %that $(\{t_{1}\},\{t_{2}\})$ is not eliminated in the elimination procedure for $(\overline{t}_{1}$, $\overline{t}_{2})$.
%
\begin{lemma} If $\mos$ is eliminated, then $\mos_{1}\wedge \case(\cmf_{1})\models \sep_{\cmf}(\mos)$
and $\mos_{2}\wedge \case(\cmf_{2})\models \neg\sep_{\cmf}(\mos)$.
\end{lemma}
\begin{proof}
    If $\mos$ is eliminated in the base case, then this follows by definition. 
    Now assume $\mydiam{R}\chi\in \mos_{1}$, $R\in \sigma$, and $\mos$ is eliminated in round $\ell+1$ because there is no witness in $\mathcal{Y}_{\ell}$. Then all $(t,t')$ with $t\in W_{1}=\{ t \in \TP\mid \chi\in t, \mos_{1} \rightsquigarrow_{R} t\}$
and $t'\in W_{2}= \{t'\in \TP\mid \mos_{2} \rightsquigarrow_{R} t'\}$ have been eliminated. By induction hypothesis,
the lemma holds for these $(t,t')$. To show that $\mos_{1}\wedge \case(\cmf_{1})\models \sep_{\cmf}(\mos)$,
assume $\Mmc_{1},w_{1}\models \mos_{1} \wedge \case(\cmf_{1})$. Then there is $v_{1}$ with $(w_{1},v_{1})\in R^{\Mmc_{1}}$
and $\Mmc_{1},v_{1}\models t$ for some $t\in W_{1}$. By induction hypothesis, $\Mmc_{1},v_{1}\models \sep_{\cmf}(t,t')$ for all $t'\in W_{2}$. Hence $\Mmc_{1},w_{1}\models \sep_{\cmf}(\mos)$. To show that $\mos_{2}\wedge \case(\cmf_{2})\models \neg\sep_{\cmf}(\mos)$,
assume $\Mmc_{2},w_{2}\models \mos_{2} \wedge \case(\cmf_{2}) \wedge \sep_{\cmf}(\mos)$.
Then there are $t\in W_{1}$ and $v_{2}$ with $(w_{2},v_{2})\in R^{\Mmc_{2}}$
and $\Mmc_{2},v_{2}\models \bigwedge_{t'\in W_{2}}\sep_{\cmf}(t,t')$.
By induction hypothesis, $\Mmc_{2},v_{2}\models \neg t'$ for all $t'\in W_{2}$. This contradicts the definition
of $W_{2}$ since $\mos_{2}\models \mybox{R}\bigvee_{t'\in W_{2}} t'$.
\end{proof}
We next consider $\Hfull$. Fix $\varphi_{1},\varphi_{2}$ in $\Hfull$ and a signature $\sigma$.
For $\Hfull$, we assume without loss of generality that $\mydiam{U}a\in \sub(\varphi_1,\varphi_2)$ for $a\in \sigma$.
We say that types $t,t'$ \emph{coincide for $U$}, in symbols $t\equiv_{U}t'$,
if $\mydiam{U}\chi\in t$ iff $\mydiam{U}\chi \in t'$ for all formulae $\mydiam{U}\chi$. Observe that types satisfied in the same model coincide for $U$.
Mosaics $\mos,\mos'$ coincide for $U$, in symbols $\mos\equiv_{U}\mos'$,
if $\mos_{i}\equiv_{U}\mos_{i}$ for both $i\in \{1,2\}$.
Types $t,t'$ are called \emph{$R$-coherent for $\Hfull$}, in symbols $t \rightsquigarrow_{R}^{U} t'$, if $t\equiv_{U}t'$ and $\chi\in t'$ and $\mydiam{R}\chi \in \sub(\varphi_1,\varphi_2)$ imply $\mydiam{R}\chi\in t$.
Singleton mosaics $\mos$ and $\mos'$ are \emph{$R$-coherent for $\Hfull$},
in symbols $\mos\rightsquigarrow_{R}^{U} \mos'$, if $\mos_{i}\rightsquigarrow_{R}^{U} \mos_{i}'$ for both $i\in\{1,2\}$.
Observe that the relations $\rightsquigarrow_{U}^{U}$ and $\equiv_{U}$ coincide.

Let $\mos$ be a singleton mosaic and $i\in \{1,2\}$. 
We next define the notion of a witness. If $R\in \Rel\setminus\sigma$, then we say that a singleton mosaic $\mos'$ is an \emph{$\Hfull$-witness} for $\mydiam{R}\chi\in \mos_{i}$ if $\chi\in \mos_{i}'$, $\mos_{i} \rightsquigarrow_{R}^{U} \mos_{i}'$, and $\mos \equiv_{U} \mos'$. If $R\in \sigma\cup \{U\}$ we additionally require that
$\mos \rightsquigarrow_{R}^{U} \mos'$. Observe that $\mos \rightsquigarrow_{R}^{U} \mos'$
implies $\mos_{i} \rightsquigarrow_{R}^{U} \mos_{i}'$ and $\mos \equiv_{U} \mos'$ but that the 
converse implication only holds for $R=U$. 

%\item[(W$_{\bar{\sigma}}$)] if $R\not\in\sigma$, then $\chi\in t_{i}'$ and $t_{i}\rightsquigarrow_{R} t_{i}'$.
%\end{description}
Fix a basic nominal case $\cmf = (\cmf_{1},\cmf_{2})$. 

\medskip
\noindent
\emph{Elimination procedure for $\cmf$ in $\Hfull$}. Let $\mathcal{Y}$ be the set of all singleton mosaics. Obtain $\mathcal{Y}_{0}$ from $\mathcal{Y}$ by dropping 
\begin{itemize}
    \item all $\mos$ such that $\mos_{i}\wedge \case(\cmf_{i})$ is not satisfiable for some $i\in \{1,2\}$, 
%    \item all $\mos$ such that $a\in \mos_{i}$ and $\mos_{i}\not=t_{a}^{i}$ for some $a$ and 
%    $i\in \{1,2\}$. 
    \item all $\sigma$-inconsistent $\mos$.
\end{itemize}
Next obtain by induction $\mathcal{Y}_{\ell+1}$ from $\mathcal{Y}_{\ell}$ by dropping all $\mos$ 
such that there is $\mydiam{R}\chi\in \mos_{i}$ without any $\Hfull$-witness in $\mathcal{Y}_{\ell}$.
Let $\mathcal{Y}^{\ast}$ denote the set of mosaics where this procedure stabilises.
\begin{restatable}{lemma}{lemjointconsonetwo}
\label{lem:jointconsonetwo} 
%Let $\overline{t}_{1}$ and $\overline{t}_{2}$ be nominal sequences
%and $\mos$ a singleton mosaic.
The following conditions are equivalent for all $\mos$:
\begin{itemize}
    \item $\mos_{1} \wedge \case(\cmf_{1}),\mos_{2}\wedge \case(\cmf_{2})$ are jointly $\Hfull(\sigma)$-consistent; 
    \item $\mathsf{m}\in \mathcal{Y}^{\ast}$.
\end{itemize}
\end{restatable}
\begin{proof}
    The proof is similar to the proof of Lemma~\ref{lem:jointconsone}.
    Assume $\mathsf{m}=(\mos_{1},\mos_{2})$ and $\mathcal{M}_1,w_{1}\sim_{\Hfull,\sigma}\mathcal{M}_{2},w_{2}$ with
    $\mathcal{M}_i,w_{i}\models \mos_{i} \wedge \case(\cmf_{i})$, for both $i\in\{1,2\}$. We show that $\mathsf{m}$ is not eliminated in the elimination procedure for $\cmf$.
    To this end, let 
\[
\mathcal{X}=\{(\tp_{\Mmc_{1}}(d_{1}),\tp_{\Mmc_{2}}(d_{2})) \mid \Mmc_{1},d_{1} \sim_{\Hfull,\sigma}\mathcal{M}_{2},d_{2}\}\\
\]
It is easy to see that $\mathcal{X}\subseteq \mathcal{Y}^{\ast}$ and so the claim follows.
    
\medskip

Conversely, given $\mos^{\ast}\in \mathcal{Y}^{\ast}$ we construct models $\Mmc_{1},\Mmc_{2}$ witnessing that $\mos^{\ast}_{1}\wedge \case(\cmf_{1}),\mos^{\ast}_{2} \wedge \case(\cmf_{2})$ are jointly $\Hfull(\sigma)$-consistent. Let $\mathcal{Y}' = \{\mos'\in \mathcal{Y}^{\ast}\mid \mos'\equiv_{U}\mos^{\ast}\}$.
Now construct models $\Mmc_{i}$ by setting $W^{\Mmc_{i}} = \mathcal{Y}'$ and for all $\mos,\mos'\in W^{\Mmc_{i}}$: 
\begin{itemize}
    \item $\mos\in p^{\Mmc_{i}}$ if $p\in \mos_{i}$, for every propositional variable $p$;
    \item $\mos = a^{\Mmc_{i}}$ if $t_{a}^{i}=\mos_{i}$, for every nominal $a$;
    \item $(\mathsf{m},\mathsf{m}') \in R^{\Mmc_{i}}$ if $\mathsf{m}\rightsquigarrow_{R}^{U}\mathsf{m}'$, for every $R\in \sigma\cup \{U\}$;
    \item $(\mathsf{m},\mathsf{m}') \in R^{\Mmc_{i}}$ if $\mathsf{m}_{i}\rightsquigarrow_{R}^{U}\mathsf{m}_{i}'$, for every $R\in \Rel\setminus\sigma$.
\end{itemize}
It is not difficult to show using the construction that for all $\mos\in \mathcal{Y}'$:
\begin{itemize}
\item $\Mmc_{1},\mos\sim_{\Hfull,\sigma} \Mmc_{2},\mos$;
\item $\Mmc_{1},\mos\models \mos_{1}$;
\item $\Mmc_{2},\mos\models \mos_{2}$.
\end{itemize}
It follows that $\mos_{1}^{\ast}\wedge \case(\cmf_{1}),\mos_{2}^{\ast}\wedge \case(\cmf_{2})$ are jointly $\Hmc(\sigma)$-consistent for every $\mos\in \mathcal{Y}^{\ast}$.
\end{proof}
We now construct, similar to $\Hmc$, by induction for each eliminated $\mos$ a $\Hfull(\sigma)$-formula $\sep_{\cmf}(\mos)$ 
with 
\begin{equation*}
\mos_{1} \wedge \case(\cmf_{1}) \models \sep_{\cmf}(\mathsf{m}),
\hspace{0.1cm} \mos_{2} \wedge \case(\cmf_{2}) \models \neg\sep_{\cmf}(\mathsf{m})
\end{equation*}
The base case is exactly as before. If $\mos$ is eliminated in round $\ell+1$, we now have to also consider the case $R\not\in\sigma$ and $R\in U$. Assume $\mos$ is eliminated because there is 
$\mydiam{R}\chi\in \mos_{i}$ with $R \not\in \sigma$ or $R=U$ and there is no witness 
left in $\mathcal{Y}_{\ell}$. Assume $i=1$ (the case $i=2$ is dual).
Set $W_{1}=\{ t \in \TP\mid \chi\in t, \mos_{1} \rightsquigarrow_{R}^{U} t\}$
and $W_{2}= \{t'\in \TP\mid \mos_{2} \equiv^{U} t'\}$. Note that all mosaics in $W_{1}\times W_{2}$  have been eliminated before round $\ell+1$. Then we set
\[
\sep_{\cmf}(\mos) = \mydiam{U}(\bigvee_{t\in W_{1}}\bigwedge_{t'\in W_{2}}\sep_{\cmf}(t,t'))
\]
(Here, if $W_{1}=\emptyset$, then $\sep_{\cmf}(\mos)=\bot$. Note that $W_{2}\not=\emptyset$ since $\mos_{2}\in W_{2}.)$ 

If $\mos$ is eliminated because there is 
$\mydiam{R}\chi\in \mos_{i}$ with $R \in \sigma$ and there is no witness 
left in $\mathcal{Y}_{\ell}$ then we define $\sep_{\cmf}(\mos)$
exactly as for $\Hmc$ but with $\rightsquigarrow_{R}$ replaced by $\rightsquigarrow_{R}^{U}$ in the definition of
$W_{i}$, for both $i\in\{1,2\}$. It follows from Lemma~\ref{lem:jointconsonetwo} that $\sep_{\cmf}(\mos)$ is as required. The construction of the $\Hfull(\sigma)$-separator for $\varphi_{1},\varphi_{2}$ is now exactly the same as for $\Hmc$.  

\bigskip

We next consider $\Hmc_{@}$. Fix $\varphi_{1},\varphi_{2}$ in $\Hmc_{@}$ and a signature $\sigma$.
For $\Hmc_{@}$, we assume without loss of generality that $@_{a}a\in \sub(\varphi_1,\varphi_2)$ for $a\in \sigma$. 
We say that types $t,t'$ \emph{coincide for $@$}, in symbols $t\equiv_{@} t'$, if for all $@_{a}\chi\in \sub(\varphi_,\varphi_{2})$, $@_{a}\chi\in t$ iff $@_{a}\chi\in t'$. Mosaics $\mos,\mos'$ coincide for $@$,
in symbols $\mos \equiv_{@} \mos'$, if $\mos_i \equiv_{@} \mos_{i}'$ for both $i\in \{1,2\}$.
Types $t,t'$ are called \emph{$R$-coherent for $\Hmc_{@}$}, in symbols $t \rightsquigarrow_{R}^{@} t'$, if $t\equiv_{@}t'$ and $\chi\in t'$ and $\mydiam{R}\chi \in \sub(\varphi_1,\varphi_2)$ imply $\mydiam{R}\chi\in t$.
Singleton mosaics $\mos$ and $\mos'$ are \emph{$R$-coherent for $\Hmc_{@}$},
in symbols $\mos\rightsquigarrow_{R}^{@} \mos'$, if $\mos_{i}\rightsquigarrow_{R}^{@} \mos_{i}'$ for both $i\in \{1,2\}$.

Let $\mos$ be a singleton mosaic and $i\in \{1,2\}$. 
We next define the notion of a witness. For $R\in \sigma$, we say that a singleton mosaic $\mos'$ is an \emph{$\Hmc_{@}$-witness} for $\mydiam{R}\chi\in \mos_{i}$ if $\chi\in \mos_{i}'$ and 
$\mos \rightsquigarrow_{R}^{@} \mos'$. We say that a singleton mosaic $\mos'$ is an \emph{$\Hmc_{@}$-witness} for $@_{a}\chi\in \mos_{i}$ if $a,\chi\in \mos_{i}$ and $\mos \equiv_{@} \mos'$.

%\item[(W$_{\bar{\sigma}}$)] if $R\not\in\sigma$, then $\chi\in t_{i}'$ and $t_{i}\rightsquigarrow_{R} t_{i}'$.
%\end{description}
Fix a basic nominal case $\cmf =(\cmf_{1},\cmf_{2})$. Note that then $\chi\in \cmf_{i}(a)$ iff $@_{a}\chi\in \cmf_{i}(a)$, for all $@_{a}\chi\in \sub(\varphi_{1},\varphi_{2})$.

\medskip
\noindent
\emph{Elimination procedure for $\cmf$ in $\Hmc_{@}$}. Let $\mathcal{Y}$ be the set of all singleton mosaics. Obtain $\mathcal{Y}_{0}$ from $\mathcal{Y}$ by dropping 
\begin{itemize}
    \item all $\mos$ such that $\mos_{i}\wedge \case(\cmf_{i})$ is not satisfiable for some $i\in \{1,2\}$, 
%    \item all $\mos$ such that $a\in \mos_{i}$ and $\mos_{i}\not=t_{a}^{i}$ for some $a$ and 
%    $i\in \{1,2\}$. 
    \item all $\sigma$-inconsistent $\mos$.
\end{itemize}
Next obtain by induction $\mathcal{Y}_{\ell+1}$ from $\mathcal{Y}_{\ell}$ by dropping all $\mos$ 
such that there is $\mydiam{R}\chi\in \mos_{i}$ or $@_{a}\chi\in \mos_{i}$ without any $\Hmc_{@}$-witness in $\mathcal{Y}_{\ell}$.
Let $\mathcal{Y}^{\ast}$ denote the set of mosaics where this procedure stabilises.
The following is shown similarly to Lemma~\ref{lem:jointconsone}.
\begin{restatable}{lemma}{lemjointconsthree}
\label{lem:jointconsthree} 
%Let $\overline{t}_{1}$ and $\overline{t}_{2}$ be nominal sequences
%and $\mos$ a singleton mosaic.
The following conditions are equivalent for all $\mos$:
\begin{itemize}
    \item $\mos_{1} \wedge \case(\cmf_{1}),\mos_{2}\wedge \case(\cmf_{2})$ are jointly $\Hmc_{@}(\sigma)$-consistent;
    \item $\mathsf{m}\in \mathcal{Y}^{\ast}$.
\end{itemize}
\end{restatable}
We now construct, similar to $\Hmc$, by induction for each eliminated $\mos$ a $\Hmc_{@}(\sigma)$-formula $\sep_{\cmf}(\mos)$ 
with 
\begin{equation*}
\mos_{1} \wedge \case(\cmf_{1}) \models \sep_{\cmf}(\mathsf{m}),
\hspace{0.1cm} \mos_{2} \wedge \case(\cmf_{2}) \models \neg\sep_{\cmf}(\mathsf{m})
\end{equation*}
The base case is exactly as before. If $\mos$ is eliminated in round $\ell+1$, we now have to also consider the case $@_{a}\chi$. Assume $\mos$ is eliminated because there is 
$@_{a}\chi\in \mos_{i}$ and there is no witness 
left in $\mathcal{Y}_{\ell}$. Assume $i=1$ (the case $i=2$ is dual).
Set $W_{1}=\{ t \in \TP\mid a,\chi\in t, \mos_{1} \equiv_{@} t\}$
and $W_{2}= \{t'\in \TP\mid \mos_{2} \equiv_{@} t'\}$. Note that all mosaics in $W_{1}\times W_{2}$ have been eliminated before round $\ell+1$. Then we set
\[
\sep_{\cmf}(\mos) = @_{a}(\bigvee_{t\in W_{1}}\bigwedge_{t'\in W_{2}}\sep_{\cmf}(t,t'))
\]
(Here, if $W_{1}=\emptyset$, then $\sep_{\cmf}(\mos)=\bot$. Note that $W_{2}\not=\emptyset$ since $\mos_{2}\in W_{2}.)$ 

If $\mos$ is eliminated because there is 
$\mydiam{R}\chi\in \mos_{i}$ with $R \in \sigma$ and there is no witness 
left in $\mathcal{Y}_{\ell}$ then we define $\sep_{\cmf}(\mos)$
exactly as for $\Hmc$ but with $\rightsquigarrow_{R}$ replaced by $\rightsquigarrow_{R}^{@}$ in the definition of
$W_{i}$, for both $i\in\{1,2\}$. It follows from Lemma~\ref{lem:jointconsthree} that $\sep_{\cmf}(\mos)$ is as required. The construction of the $\Hmc_{@}(\sigma)$-separator for $\varphi_{1},\varphi_{2}$ is now exactly the same as for $\Hmc$.  

\section{Proofs for Section~\ref{sec:hyper}}

\thmhypermosaics*

\begin{proof}

Assume first that $H$ is realized in $\M_{1}, \M_{2}$ by $\Hmc(\sigma)$-bisimulation $Z$. 
% Without loosing generality $Z$ is the \emph{maximal} $\Hmc(\sigma)$-bisimulation between $\M_{1}$ and $\M_{2}$.
Define $H'$ as the set of maximal (w.r.t.~set-inclusion on both coordinates) mosaics $\mos$ realized by $Z$. By definition, $H'$ is nominal-clean and $\sigma$-consistent. Moreover, $H'$ is a witness for all $\mos\in H\cup H'$. So neither $H$ nor $H'$ are eliminated.

Conversely, assume that $H$ is not eliminated. We construct models $\M_{1}$ and $\M_{2}$ realizing $H$ by $\Hmc(\sigma)$-bisimulation $Z$. We first construct a sequence $H_0,H_1,\ldots$ of hypermosaics in $\Zmc^*$ as follows:
\begin{itemize}
\item We start with $H_0:=H$.
\item For $i\geq 0$, try to pick some $\mos\in H_i$ such that $H_i$ itself is not a witness for $\mos$. If there is no such $\mos$, then $H_i$ is the last element of the sequence. Otherwise, $H_{i+1}$ is some $H'\in \Zmc^*$ that is a witness for $\mos$. (Such $H'$ always exists as $H_i$ is not eliminated.)
\end{itemize}
Note that the sequence is finite since each witness $H_{i+1}$ (strictly) extends $H_i$. Moreover, by definition, the last element $H^*$ of the sequence is a witness for all mosaics $\mos\in H^*$. Clearly, $H^*$ is $\sigma$-consistent and nominal-clean. 

We are now in a position to define $\M_{i}$ and $Z$. For $i\in \{1,2\}$, let Nom$_{i}$ be the set of nominals $a\in \Nom \cap (\sig(\varphi_1)\cup\sig(\varphi_2))$ that do not occur in any $t$ with $t\in \mos_{i}$ 
for some $\mos\in H^*$. Take
\[
W^{\M_{i}}= \{(t,\mos)\mid t\in \mos_{i}, \mos\in H^*\} \cup \text{Nom}_{i}
\]
as the domain of $\M_i$.
We next define $R^{\M_{i}}$ for relation symbols $R\in\sig(\varphi_1)\cup\sig(\varphi_2)$. Let $(t,\mos),(t',\mos')\in W^{\M_{i}}$. Then $((t,\mos),(t',\mos'))\in R^{\M_i}$ iff $t\rightsquigarrow_R t'$ and, if $R\in \sigma$, additionally $\mos \rightsquigarrow_R \mos'$.
For a propositional variable $p\in \sig(\varphi_1)\cup\sig(\varphi_2)$, let $p^{\M_{i}}=\{ (t,\mos)\in W_{i} \mid p\in t\}$.

For $a\in \Nom \cap (\sig(\varphi_1)\cup\sig(\varphi_2)$ we make a case distinction.
If there is $(t,\mos)\in W_{i}$ with $a\in t$, then let $a^{\M_{i}}= (t,\mos)$.
Otherwise $a\in \text{Nom}_{i}$ and we set $a^{\M_{i}}=a$.

\begin{claim}
    For every $i\in\{1,2\}$, $\psi\in\sub(\varphi_1,\varphi_2)$ and every $(t,\mos)\in W^{\Mmc_i}$, we have: 
    \[
    \psi\in t\quad\text{ iff }\quad\Mmc_i,(t,\mos)\models \psi.
    \]
\end{claim}

\begin{proof}
We prove the claim by induction on $\psi$. The only non-trivial case are formulae of shape $\psi=\mydiam{R}\chi$. 

($\Rightarrow$) Suppose $\mydiam{R}\chi\in t$. Since $H^*$ is a witness for $\mydiam{R}\chi\in t\in \mos$, there is $\mos'\in H^*$ and a type $t'\in \mos'_i$ such that $\chi\in t'$, $t\rightsquigarrow_Rt'$ and, in case $R\in\sigma$, also $\mos\rightsquigarrow_R \mos'$. By induction hypothesis, we have $\M_i,(t',\mos')\models \chi$. By construction, we have $((t,\mos),(t',\mos'))\in R^{\M_i}$. Hence, $\M_i,(t,\mos)\models \mydiam{R}\chi$.

($\Leftarrow$) Suppose $\M_i,(t,\mos)\models \mydiam{R}\chi$. Then there is $(t',\mos')$ with $\M_i,(t',\mos')\models\chi$ and $((t,\mos),(t',\mos'))\in R^{\M_i}$. By induction hypothesis, $\chi\in t'$. Since $t\rightsquigarrow_R t'$, we have $\mydiam{R}\chi\in t$.
\end{proof}

Finally, define $Z\subseteq W^{\M_1}\times W^{\M_2}$ by setting
\begin{gather*}
(t,\mos)\ Z\ (t',\mos') \iff \mos=\mos'.
\end{gather*}

\begin{claim}
    $Z$ is an $\Hmc(\sigma)$-bisimulation between $\Mmc_1$ and $\Mmc_2$, and each $\mos\in H^*$ is realized in $\Mmc_1,\Mmc_2$ by $Z$.
\end{claim}
\begin{proof}
The atomic and nominal conditions are immediate. For the back and forth conditions it is enough to show that for every $R\in\sigma$, points $(t,\mos)$, $(t',\mos)$, and mosaics $\mos'$, either both $(t,\mos),(t',\mos)$ have an $R$-successor with second component $\mos'$ or none has. Suppose $((t,\mos),(\hat t,\mos'))\in R^{\M_i}$ for some $i\in\{1,2\}$. By definition of $R^{\M_i}$, $\mos\rightsquigarrow_R \mos'$. That is, for every $j\in\{1,2\}$ and every $s\in\mos_j$, we have some $s'\in\mos'_j$ with $s\rightsquigarrow_R s'$. Hence, there is $\hat t'\in\mos_{3-i}$ with $t'\rightsquigarrow_R \hat t'$. By definition of $R^{\M_i}$, $((t',\mos),(\hat t',\mos'))\in R^{\M_{3-i}}$.
\end{proof}

It remains to note that $H$ is realized in $\Mmc_1,\Mmc_2$ by $Z$ since each $\mos\in H$ is (component-wise) contained in some $\mos^*\in H^*$ and each $\mos^*$ in $H^*$ is realized in $\Mmc_1,\Mmc_2$ by $Z$, by the Claim.
\end{proof}

\subsection{Elimination Procedure for \Hfull}

We let $\Zmc^*$ denote the result of the hypermosaic elimination procedure with the changes needed for $\Hfull$ described in the main body of paper. Then, we have the following equivalent of Theorem~\ref{thm:correctness-elimination}.

\begin{theorem}\label{thm:correctness-elimination-full}
For a hypermosaic $H$, the following are equivalent: 
\begin{itemize}

    \item  $H$ is realizable by an $\Hfull(\sigma)$-bisimulation;
    
    \item $H\in \Zmc^*$.
    
\end{itemize}
\end{theorem}

\begin{proof} 
The proof is similar to the proof of Theorem~\ref{thm:correctness-elimination}.

Assume first that $H$ is realized in $\M_{1}, \M_{2}$ by $\Hmc(\sigma)$-bisimulation $Z$. 
% Without loosing generality $Z$ is the \emph{maximal} $\Hmc(\sigma)$-bisimulation between $\M_{1}$ and $\M_{2}$.
Define $H'$ as the set of maximal (w.r.t.~set-inclusion on both coordinates) mosaics $\mos$ realized by $Z$. By definition, $H'$ is nominal-clean and $\sigma$-consistent. Moreover, $H'$ is a witness for all $\mos\in H\cup H'$. So neither $H$ nor $H'$ are eliminated.

Conversely, assume that $H$ is not eliminated. We construct models $\M_{1}$ and $\M_{2}$ realizing $H$ by $\Hmc(\sigma)$-bisimulation $Z$. We first construct a sequence $H_0,H_1,\ldots$ of hypermosaics in $\Zmc^*$ as follows:
\begin{itemize}
\item We start with $H_0:=H$.
\item For $i\geq 0$, try to pick some $\mos\in H_i$ such that $H_i$ itself is not a witness for $\mos$. If there is no such $\mos$, then $H_i$ is the last element of the sequence. Otherwise, $H_{i+1}$ is some $H'\in \Zmc^*$ that is a witness for $\mos$. (Such $H'$ always exists as $H_i$ is not eliminated.)
\end{itemize}
Note that the sequence is finite since each witness $H_{i+1}$ (strictly) extends $H_i$. Moreover, by definition, the last element $H^*$ of the sequence is a witness for all mosaics $\mos\in H^*$. Clearly, $H^*$ is $\sigma$-consistent, $U$-consistent, and nominal-clean. 

We are now in a position to define $\M_{i}$ and $Z$. Take
\[
W^{\M_{i}}= \{(t,\mos)\mid t\in \mos_{i}, \mos\in H^*\}
\]
as the domain of $\M_i$ and set $U^{\M_i}=W^{\M_i}\times W^{\M_i}$.
We next define $R^{\M_{i}}$ for relation symbols $R\in\Rel\cap(\sig(\varphi_1)\cup\sig(\varphi_2))$. Let $(t,\mos),(t',\mos')\in W^{\M_{i}}$. Then $((t,\mos),(t',\mos'))\in R^{\M_i}$ iff $t\rightsquigarrow_R t'$ and, if $R\in \sigma$, additionally $\mos \rightsquigarrow_R \mos'$.
For a propositional variable $p\in \sig(\varphi_1)\cup\sig(\varphi_2)$, let $p^{\M_{i}}=\{ (t,\mos)\in W_{i} \mid p\in t\}$.

For $a\in \Nom \cap (\sig(\varphi_1)\cup\sig(\varphi_2)$, we set $a^{\M_{i}}= (t_a,\mos)$ where $\mos\in H^*$ is the unique mosaic that mentions $a$ and $t_a\in \mos_i$ is the unique type in $\mos_i$ that contains $a$.

\begin{claim}
    For every $i\in\{1,2\}$, $\psi\in\sub(\varphi_1,\varphi_2)$ and every $(t,\mos)\in W^{\Mmc_i}$, we have: 
    \[
    \psi\in t\quad\text{ iff }\quad\Mmc_i,(t,\mos)\models \psi.
    \]
\end{claim}

\begin{proof}
We prove the claim by induction on $\psi$. The only non-trivial case are formulae of shape $\psi=\mydiam{R}\chi$. Here, $R$ can be either a relation symbol in $\Rel$ or $U$.

($\Rightarrow$) Suppose $\mydiam{R}\chi\in t$. Since $H^*$ is a witness for $\mydiam{R}\chi\in t\in \mos$, there is $\mos'\in H^*$ and a type $t'\in \mos'_i$ such that $\chi\in t'$, $t\rightsquigarrow_Rt'$ and, in case $R\in\sigma$, also $\mos\rightsquigarrow_R \mos'$. By induction hypothesis, we have $\M_i,(t',\mos')\models \chi$. By definition of $R^{\M_i}$ (recall that this is potentially $U^{\M_i}$), we have $((t,\mos),(t',\mos'))\in R^{\M_i}$. Hence, $\M_i,(t,\mos)\models \mydiam{R}\chi$.

($\Leftarrow$) Suppose $\M_i,(t,\mos)\models \mydiam{R}\chi$. Then there is $(t',\mos')$ such that $\M_i,(t',\mos')\models\chi$ and $((t,\mos),(t',\mos'))\in R^{\M_i}$. By induction hypothesis, $\chi\in t'$. If $R\in\Rel$, then $((t,\mos),(t',\mos'))\in R^{\M_i}$ implies $t\rightsquigarrow_R t'$, and hence $\mydiam{R}\chi\in t$. If $R=U$, then $\chi\in t'$ implies $\mydiam{U}\chi\in t'$. $U$-consistency of $H^*$ finally implies $\mydiam{U}\chi\in t$.
\end{proof}

Finally, define $Z\subseteq W^{\M_1}\times W^{\M_2}$ by setting
\begin{gather*}
(t,\mos)\ Z\ (t',\mos') \iff \mos=\mos'.
\end{gather*}

\begin{claim}
    $Z$ is an $\Hmc(\sigma)$-bisimulation between $\Mmc_1$ and $\Mmc_2$, and each $\mos\in H^*$ is realized in $\Mmc_1,\Mmc_2$ by $Z$.
\end{claim}
\begin{proof}
The atomic and nominal conditions are immediate. For the back and forth conditions it is enough to show that for every $R\in\sigma\cup\{U\}$, points $(t,\mos)$, $(t',\mos)$, and mosaics $\mos'$, either both $(t,\mos),(t',\mos)$ have an $R$-successor with second component $\mos'$ or none has. This is true for $R=U$ by our assumption that $\mos_1$ and $\mos_2$ are non-empty for every $\mos\in H^*$.
For $R\in \sigma$, suppose $((t,\mos),(\hat t,\mos'))\in R^{\M_i}$ for some $i\in\{1,2\}$. By definition of $R^{\M_i}$, $\mos\rightsquigarrow_R \mos'$. That is, for every $j\in\{1,2\}$ and every $s\in\mos_j$, we have some $s'\in\mos'_j$ with $s\rightsquigarrow_R s'$. Hence, there is $\hat t'\in\mos_{3-i}$ with $t'\rightsquigarrow_R \hat t'$. By definition of $R^{\M_i}$, $((t',\mos),(\hat t',\mos'))\in R^{\M_{3-i}}$.
\end{proof}

It remains to note that $H$ is realized in $\Mmc_1,\Mmc_2$ by $Z$ since each $\mos\in H$ is (component-wise) contained in some $\mos^*\in H^*$ and each $\mos^*$ in $H^*$ is realized in $\Mmc_1,\Mmc_2$ by $Z$, by the Claim.
\end{proof}

\begin{corollary}\label{cor:hypermosaics-full}
    $\varphi_1$ and $\varphi_2$ are jointly $\Hfull(\sigma)$-consistent iff there is a hypermosaic $\{(\{t_{1}\},\{t_{2}\})\}\in\Zmc^*$ with $\varphi_1\in t_{1}$ and $\varphi_2\in t_{2}$.
\end{corollary}

\subsection{Elimination Procedure for \Hat}

We let $\Zmc^*$ denote the result of the hypermosaic elimination procedure with the changes needed for $\Hat$ described in the main body of paper. Then, we have the following equivalent of Theorem~\ref{thm:correctness-elimination}.

\begin{theorem}\label{thm:correctness-elimination-at}
For a hypermosaic $H$, the following are equivalent: 
\begin{itemize}

    \item  $H$ is realizable by an $\Hat(\sigma)$-bisimulation;
    
    \item $H\in \Zmc^*$.
    
\end{itemize}
\end{theorem}

\begin{proof}
The proof is almost literally the same as the proof of Theorem~\ref{thm:correctness-elimination}. We provide it for the sake of completeness and point out the few the differences. 

Assume first that $H$ is realized in $\M_{1}, \M_{2}$ by $\Hmc(\sigma)$-bisimulation $Z$. 
% Without loosing generality $Z$ is the \emph{maximal} $\Hmc(\sigma)$-bisimulation between $\M_{1}$ and $\M_{2}$.
Define $H'$ as the set of maximal (w.r.t.~set-inclusion on both coordinates) mosaics $\mos$ realized by $Z$. By definition, $H'$ is nominal-clean and $\sigma$-consistent. Moreover, $H'$ is a witness for all $\mos\in H\cup H'$. So neither $H$ nor $H'$ are eliminated.

Conversely, assume that $H$ is not eliminated. We construct models $\M_{1}$ and $\M_{2}$ realizing $H$ by $\Hmc(\sigma)$-bisimulation $Z$. We first construct a sequence $H_0,H_1,\ldots$ of hypermosaics in $\Zmc^*$ as follows:
\begin{itemize}
\item We start with $H_0:=H$.
\item For $i\geq 0$, try to pick some $\mos\in H_i$ such that $H_i$ itself is not a witness for $\mos$. If there is no such $\mos$, then $H_i$ is the last element of the sequence. Otherwise, $H_{i+1}$ is some $H'\in \Zmc^*$ that is a witness for $\mos$. (Such $H'$ always exists as $H_i$ is not eliminated.)
\end{itemize}
Note that the sequence is finite since each witness $H_{i+1}$ (strictly) extends $H_i$. Moreover, by definition, the last element $H^*$ of the sequence is a witness for all mosaics $\mos\in H^*$. Clearly, $H^*$ is $\sigma$-consistent and nominal-clean. 

We are now in a position to define $\M_{i}$ and $Z$. For $i\in \{1,2\}$, let Nom$_{i}$ be the set of nominals $a\in \Nom \cap (\sig(\varphi_1)\cup\sig(\varphi_2))$ that do not occur in any $t$ with $t\in \mos_{i}$ 
for some $\mos\in H^*$. Take
\[
W^{\M_{i}}= \{(t,\mos)\mid t\in \mos_{i}, \mos\in H^*\} \cup \text{Nom}_{i}
\]
as the domain of $\M_i$.
We next define $R^{\M_{i}}$ for relation symbols $R\in\sig(\varphi_1)\cup\sig(\varphi_2)$. Let $(t,\mos),(t',\mos')\in W^{\M_{i}}$. Then $((t,\mos),(t',\mos'))\in R^{\M_i}$ iff $t\rightsquigarrow_R t'$ and, if $R\in \sigma$, additionally $\mos \rightsquigarrow_R \mos'$.
For a propositional variable $p\in \sig(\varphi_1)\cup\sig(\varphi_2)$, let $p^{\M_{i}}=\{ (t,\mos)\in W_{i} \mid p\in t\}$.

For $a\in \Nom \cap (\sig(\varphi_1)\cup\sig(\varphi_2)$ we make a case distinction.
If there is $(t,\mos)\in W_{i}$ with $a\in t$, then let $a^{\M_{i}}= (t,\mos)$.
Otherwise $a\in \text{Nom}_{i}$ and we set $a^{\M_{i}}=a$.

\begin{claim}
    For every $i\in\{1,2\}$, $\psi\in\sub(\varphi_1,\varphi_2)$ and every $(t,\mos)\in W^{\Mmc_i}$, we have: 
    \[
    \psi\in t\quad\text{ iff }\quad\Mmc_i,(t,\mos)\models \psi.
    \]
\end{claim}

\begin{proof}
We prove the claim by induction on $\psi$. Formulae of shape $\psi=\mydiam{R}\chi$ are treated as in the respective claim in the proof of Theorem~\ref{thm:correctness-elimination}. We consider here only formulae of shape $\psi=@_a\chi$.

($\Rightarrow$) Suppose $@_a\chi\in t$. Since $H^*$ is a witness for $\mos\in H^*$, there is $\mos'\in H^*$ and a type $t'\in \mos'_i$ such that $a,\chi\in t'$. By induction hypothesis, we have $\M_i,(t',\mos')\models a$ and $\M_i,(t',\mos')\models \chi$. Thus, $\M_i,(t,\mos)\models @_a\chi$.

($\Leftarrow$) Suppose $\M_i,(t,\mos)\models @_a\chi$. Then there is $(t',\mos')$ with $\M_i,(t',\mos')\models a$ and
$\M_i,(t',\mos')\models \chi$. By induction, $a,\chi\in t'$ and hence $@_a\chi\in t'$. By @-consistency of $H^*$, $@_a\chi\in t$.
\end{proof}

Finally, define $Z\subseteq W^{\M_1}\times W^{\M_2}$ by setting
\begin{gather*}
(t,\mos)\ Z\ (t',\mos') \iff \mos=\mos'.
\end{gather*}

\begin{claim}
    $Z$ is an $\Hmc(\sigma)$-bisimulation between $\Mmc_1$ and $\Mmc_2$, and each $\mos\in H^*$ is realized in $\Mmc_1,\Mmc_2$ by $Z$.
\end{claim}
\begin{proof}
The atomic and nominal conditions are immediate. For the back and forth conditions it is enough to show that for every $R\in\sigma$, points $(t,\mos)$, $(t',\mos)$, and mosaics $\mos'$, either both $(t,\mos),(t',\mos)$ have an $R$-successor with second component $\mos'$ or none has. Suppose $((t,\mos),(\hat t,\mos'))\in R^{\M_i}$ for some $i\in\{1,2\}$. By definition of $R^{\M_i}$, $\mos\rightsquigarrow_R \mos'$. That is, for every $j\in\{1,2\}$ and every $s\in\mos_j$, we have some $s'\in\mos'_j$ with $s\rightsquigarrow_R s'$. Hence, there is $\hat t'\in\mos_{3-i}$ with $t'\rightsquigarrow_R \hat t'$. By definition of $R^{\M_i}$, $((t',\mos),(\hat t',\mos'))\in R^{\M_{3-i}}$.
\end{proof}

It remains to note that $H$ is realized in $\Mmc_1,\Mmc_2$ by $Z$ since each $\mos\in H$ is (component-wise) contained in some $\mos^*\in H^*$ and each $\mos^*$ in $H^*$ is realized in $\Mmc_1,\Mmc_2$ by $Z$, by the Claim.
\end{proof}

\begin{corollary}\label{cor:hypermosaics-at}
    $\varphi_1$ and $\varphi_2$ are jointly $\Hat(\sigma)$-consistent iff there is a hypermosaic $\{(\{t_{1}\},\{t_{2}\})\}\in\Zmc^*$ with $\varphi_1\in t_{1}$ and $\varphi_2\in t_{2}$.
\end{corollary}

\subsection{Proofs for Separator Construction}

\claimconstruction*

\begin{proof}
We first prove that $H'$ is nominal-clean. Fix a nominal $a\in\Nom$ and side $i\in \{1,2\}$.
Note that $\cmf_i^+$ specifies the type $t$ and $\Sep_\cmf$-type $d$ of $a$, due to the inductive base of the separator construction. Hence, on the $i$-th side $a$ can only belong to the type $t$ in the mosaic $(\mos^d_1,\mos^d_2)$.

We now prove the other claim that $H'$ is a witness for every $\mos\in H$. Take $\mydiam{R}\chi\in t\in\mos_i$. Consider the $\Sep_\cmf$-star-type $\delta_\mos$. Since $\delta_\mos$ is a disjunct of $\sep^+_i(\cmf^+,t)$ there is a model
\[
\M,w\models\case(\cmf^+_i)\wedge\delta_\mos\wedge t.
\]
Since $\M,w\models t$, it follows that $w$ has an $R$-child $w'$ satisfying some $t'\wedge d$ with $\chi\in t'\in\TP$ and $d\in\SepCaseType$. In particular $t'\in \mos_i^d$. Denote $\mos'=(\mos^d_1,\mos^d_2)$. Since $wRw'$ as well as $\tp_\M(w)=t$ and $\tp_\M(w')=t'$, we have $t\rightsquigarrow_{R} t'$.

If $R\notin\sigma$ then we are done with showing that $\mos'$ witnesses $\mydiam{R}\chi\in t\in\mos_i$. Otherwise $R\in\sigma$ in which case we need to additionally prove $\mos\rightsquigarrow_{R}\mos'$. Take any $s\in \mos_j$. We need to find some $s'\in \mos^d_j$ such that $s\rightsquigarrow_{R} s'$. Because $\delta_\mos$ is a disjunct of $\sep^+_j(\cmf^+,s)$ there is a model
\[
\N,v\models\case(\cmf^+_j)\wedge\delta_\mos\wedge s.
\]
Since $\M,w\models\delta_\mos$ as well as $wRw'$ and $\M,w'\models d$, we get $\delta_\mos\models\mydiam{R}d$. Hence $v$ has an $R$-child $v'$ in $\N$ such that $\N,v'\models d$. Take the type $s'=\tp_\N(v')$ of this $v'$. Since $\N,v'\models\case(\cmf^+_j)\wedge d\wedge s'$, we get $s'\in\mos_j^d$. Moreover, $s\rightsquigarrow_{R} s'$ follows from $vRv'$. This completes the proof that $\mos'$ witnesses $\mydiam{R}\chi\in t\in\mos_i$.

It remains to prove that $H'$ extends $H$. Take $\mos\in H$. Take the unique $\Sep_\cmf$-type $d$ enforced by $\delta_\mos$ in the sense that $\delta_\mos\models d$. For every $i$ and $t\in \mos_i$ we know that $\case(\cmf^+_i)\wedge\delta_\mos\wedge t$ is consistent, and hence so is $\case(\cmf^+_i)\wedge d\wedge t$. This implies $t\in \mos_i^d$. In other words, $\mos_i\subseteq \mos^d_i$, for both $i$. This proves that $H'$ extends $H$ and hence completes the proof of Claim~\ref{cl:construction claim}.
\end{proof}

We remark that in the presence of $@$ or $U$, the Claim~\ref{cl:construction claim} additionally asserts that $H'$ is $@$-consistent or $U$-consistent, respectively. This property of $H'$ follows from the observation that in our construction a case $\cmf^+\in\Case^+$ always determines the types of all nominals on both sides.

\thmsummingup*

\begin{proof}
We add the missing details for the analysis of the size and computation time of the separators. 

Let $n=|\varphi_1|+|\varphi_2|$ and $s(n)=2^n\times 2^{2^n}$. We start with analyzing the number of rounds the hypermosaic elimination runs. A trivial triply exponential upper bound is the number of hypermosaics since in each round at least one is eliminated. However, exploiting the fact that witnesses $H'$ for a mosaic $m\in H$ have to extend $H$, one can show a better bound.  

\begin{lemma}\label{lem:runtime}
The elimination terminates after $s(n)$ rounds.
\end{lemma}

We can now analyze the size of the separators constructed during elimination. The fundamental observation is that, in each round of the elimination, one can bound the size of $\sep_{i}(\cmf,t)$ independently of both the number and the size of the cases in~$\Case$. 

Let $m_{\ell}$ be the maximal size of some $\sep_i(\cmf,t)$ and $n_\ell$ be the maximal number of cases $\Case$, after $\ell$ rounds of elimination. Clearly, $m_{0}\leq 2|\sigma|$ and $n_0\leq 2^{n^2}$ by the definition of the base case. 

In the inductive step, let us first consider $m_{\ell+1}$. %In Step~1 of the construction, the maximal size $m_\ell'$ of some $\sep_i(t,\overline\cmf)$ in the uniform separator $\Case_1,\Case_2,\sep_1,\sep_2$ is bounded by $m_\ell'\leq m_\ell\times q(n)$. Analyzing Step~2 of the construction,
We observe that, for any case $\cmf\in \Case$:
\begin{enumerate}[label=(\roman*)]
% \begin{enumerate}[label=(\roman*)]

    \item the number of $\Sep_\cmf$-types is bounded by $2^{2^{n}}$, 
    
    \item the number of $\Sep_\cmf$-star-types is bounded by $2^{2^n}\times2^{2^{2^n}}$,
    
    \item the size of a $\Sep_\cmf$-type is $O(2^n\times m_\ell)$, and

    \item the size of a $\Sep_\cmf$-star-type is
    $O(2^n\times m_\ell+2^{2^n}\times 2^{n}\times m_\ell)$.
    
\end{enumerate}

Since the size of $\sep_i^+(\cmf^+,t)$ is bounded by the product of the bounds in Items~(ii) and~(iv) above, we obtain that
\[m_{\ell+1}\leq m_\ell\times f(n)\] for some triply exponential function $f$. An easy induction yields an upper bound $m_\ell\leq f'(n)^\ell$ after $\ell$ rounds, for some triply exponential function $f'$. Since by Lemma~\ref{lem:runtime} the elimination terminates after $s(n)$ rounds, the maximal size of separators is bounded by a triply exponential function.

We next bound the number of cases $n_{\ell+1}$. 
% For this, let $H_1,\ldots,H_k$ be the minimal hypermosaics that witness $\mos\in H$, which were fixed in the beginning of the construction. It is routine to verify that the number of such minimal hypermosaics is bounded by $g(n)$ for some doubly exponential function $g$.\nb{J: I believe it is true, but better make argument}
Each case is refined by all possible functions $\Nom_0\to \SepCaseType$. By Item~(i) above, there are doubly exponentially many of those. Hence, 
\[n_{\ell+1}\leq n_\ell\times g(n),\] 
for some doubly exponential function $g$. Again, one shows by induction that $n_\ell\leq g'(n)^\ell$, for some doubly exponential function $g'$. Since the elimination terminates after $s(n)$ rounds by Lemma~\ref{lem:runtime}, the number of cases is bounded by a triply exponential function. 

The final combination of the constructed separators in Lemma~\ref{lem:hypersep2sep} leads to formulae of overall fourfold exponential size.
\end{proof}

\section{Proofs for Section~\ref{sec:graded}}

\lemRcountingweak*

\begin{proof}
We show (1). Only the right-to-left implication is nontrivial. Assume \Rmc is a weak $R$-counting function for $\mos,H'$.
We first obtain a weak $R$-counting function $\Rmc'$ bounded by $\lambda$ for $\mos, H'$ by setting
\[
\Rmc'(t,i,\mos',s) =
\begin{cases}
    \Rmc(t,i,\mos',s) & \text{if $\Rmc(t,i,\mos',s)\leq \lambda$,}\\
    \lambda & \text{otherwise,}
\end{cases}
\]
We now modify $\Rmc'$ further. Consider $\mos'\in H'$. 

\medskip
\noindent
Case 1. $\sum_{s\in \mos'_i}\Rmc'(t,i,\mos',s)\leq \lambda$ for some $t\in \mos_{i}$ and $i\in \{i,2\}$. 
Then, by definition, $\sum_{s\in \mos'_i}\Rmc'(t,i,\mos',s)\leq \lambda$ for all $t\in \mos_{i}$ and $i\in \{1,2\}$
and the required equations hold for $\Rmc'$ on $\mos'$.

\medskip
\noindent
Case 2. $\sum_{s\in \mos'_i}\Rmc'(t,i,\mos',s)> \lambda$ for some (equivalently, all) $t\in \mos_{i}$ and $i\in \{1,2\}$. We modify $\Rmc'(t,i,\mos',s)$ as follows. By the pigeonhole principle, for every $i\in \{1,2\}$ and $t\in \mos_{i}$ there exists $s_{\kappa(i,t)}\in \mos_{i}'$ such that $\Rmc(t,i,\mos',s_{\kappa(i,t)})>\kappa$. Then let
    $$
    n_{\mos'} = \max\{\sum_{s\in \mos'_i}\Rmc(t,i,\mos',s)\mid t\in \mos_{i}, i\in \{1,2\}\}
    $$
    and update $\Rmc'$ on $\mos'$ by setting 
    \[
\Rmc'(t,i,\mos',s) =
\begin{cases}
    n_{\mos'}-\sum_{s\in \mos'_i\setminus\{\kappa(t,i)\}}\Rmc(t,i,\mos',s) & \text{if $s=s_{\kappa(i,t)}$,}\\
    \Rmc(t,i,\mos',s) & \text{otherwise.}
\end{cases}
\]
Observe that now 
\[
\max\{\sum_{s\in \mos'_i}\Rmc'(t,i,\mos',s)\mid t\in \mos_{i}, i\in \{i,2\}\} = n_{\mos'}
\]
for all $t\in \mos_{i}$ and $i\in \{1,2\}$. Moreover, the remaining conditions for an $R$-counting function still hold by definition.

(2) follows directly from the proof of (1). 
\end{proof}

\thmcorrectnesseliminationgraded*

\begin{proof}
%Assume first that $H$ is realized in $\M_{1}, \M_{2}$ by $\Gmc(\sigma)$-bisimulation $Z$.
%Define $H'$ as the set of maximal (w.r.t.~set-inclusion on both coordinates) mosaics $\mos$ realized by $Z$. By %definition, $H'$ is nominal-clean and $\sigma$-consistent. Moreover, $H'$ is a witness for all $\mos\in H\cup H'$. So %neither $H$ nor $H'$ are eliminated.
% For $i\in \{1,2\}$ and $w\in W^{\M_i}$, denote with $T_j(w)$ the set of types
% %
% \begin{align*}
% %
%   T_j(w)=\{\tp_{\M_j}(v)\mid v\in W^{\M_j} \text{ with }\M_i,w\sim_{\Gmc,\sigma}\M_j,v\}.
% %
% \end{align*}
% Let $H'$ be the set of all mosaics $(T_1(w),T_2(w))$ for $w\in W^{\M_1}\cup W^{\M_2}$. By definition, $H'$ is a $\sigma$-consistent and nominal-clean hypermosaic. Moreover, $H'$ is a witness for every $\mos\in H\cup H'$. Since $H'$ extends $H$, neither $H$ nor $H'$ are eliminated.
We show the direction from bottom to top. Assume that $H$ is not eliminated. We construct models $\M_{1}$ and $\M_{2}$ realizing $H$ by $\Gmc(\sigma)$-bisimulation $Z$. We first construct a sequence $H_0,H_1,\ldots$ of hypermosaics in $\Zmc^*$ as follows:
\begin{itemize}
\item We start with $H_0:=H$.
\item For $i\geq 0$, try to pick some $\mos\in H_i$ such that $H_i$ itself is not a witness for $\mos$. If there is no such $\mos$, then $H_i$ is the last element of the sequence. Otherwise, $H_{i+1}$ is some $H'\in \Zmc^*$ that is a witness for $\mos$. (Such $H'$ always exists as $H_i$ is not eliminated.)
\end{itemize}
Note that the sequence is finite since each witness $H_{i+1}$ strictly extends $H_i$. Moreover, by definition, the last element $H^*$ of the sequence is a witness for all mosaics $\mos\in H^*$. Clearly, $H^*$ is $\sigma$-consistent and nominal-clean. 

We are now in a position to define $\M_{i}$ and $Z$. Let Nom$_{i}$ be the set of nominals $a\in \Nom \cap (\sig(\phi_1)\cup\sig(\phi_2))$ that do not occur in any $t$ with $t\in \mos_{i}$ 
for some $\mos\in H^*$. Take
\begin{align*}
W^{\M_{i}}=~& \text{Nom}_{i} \cup{}\\ 
& \{(t,j,\mos)\mid t\in \mos_{i}\text{ nom-free},\mos\in H^*, j\in\mathbb{N}, j\geq 1\} \cup{} \\
& \{(t,1,\mos)\mid t\in \mos_{i}\text{ has nom},\mos\in H^*\}
\end{align*}
as the domain of $\M_i$. We next define $R^{\M_{i}}$ for relation symbols $R\in\sig(\phi_1)\cup\sig(\phi_2)$. $H^*$ is a witness for every $\mos\in H^*$. Let us fix, for every $\mos\in H^*$, an $R$-counting function $\Rmc_\mos$. We set:
\begin{gather*}
    (t,j,\mos)\ R^{\Mmc_i}\ (t',j',\mos') \hspace{0.5cm}
    \iff \hspace{0.5cm}
    \Rmc_\mos(t,i,\mos',t')\geq j'.
\end{gather*}
The elements of $\Nom_i$ are not linked by $R^{\Mmc_i}$ with any element.

For a propositional variable $p\in \sig(\phi_1)\cup\sig(\phi_2)$, let 
\[
p^{\M_{i}}=\{ (t,j,\mos)\in W^{\M_i} \mid p\in t\}.
\]

For $a\in \Nom \cap (\sig(\phi_1)\cup\sig(\phi_2))$ we make a case distinction.
If there is $(t,j,\mos)\in W_{i}$ with $a\in t$ then let $a^{\M_{i}}= (t,j,\mos)$.
Otherwise $a\in \text{Nom}_{i}$ and we set $a^{\M_{i}}=a$.

\begin{claim}
    For every $i\in\{1,2\}$, $\chi\in\sub(\phi_1,\phi_2)$ and every $(t,j,\mos)\in W^{\Mmc_i}$, we have: 
    \[
    \chi\in t\text{ iff }\Mmc_i,(t,j,\mos)\models \chi.
    \]
\end{claim}

\begin{proof}
We prove the claim by induction on $\chi$. The only non-trivial case is with graded modalities, i.e.~when $\chi=\mydiam{\bowtie n\ R}\theta$. Let $\Rmc_\mos$ be the $R$-counting function from the construction of $\M_1,\M_2$. We have
\begin{align*}
    \Mmc_i,(t,j,\mos)\models \mydiam{\bowtie n\ R}\theta
    \iff&
    \sum_{\substack{\theta\in\mos'\in H^*\\ t'\in\mos'_i}} \Rmc_\mos(t,i,\mos',t') \bowtie n\\
    \iff&
    \mydiam{\bowtie n\ R}\theta \in t
\end{align*}
where the second equivalence is the third condition~(iii) from the definition of an $R$-counting function.
\end{proof}

Finally define $Z\subseteq W^{\M_1}\times W^{\M_2}$ by setting
\begin{gather*}
(t,j,\mos)\ Z\ (t',j',\mos') \iff \mos=\mos'.
\end{gather*}

\begin{claim}
    $Z$ is a $\Gmc(\sigma)$-bisimulation between $\Mmc_1$ and $\Mmc_2$, and each $\mos\in H^*$ is realized in $\Mmc_1,\Mmc_2$ by $Z$.
\end{claim}
\begin{proof}
The atomic and nominal conditions are immediate. For the back and forth conditions it is enough to show that for every $R\in\sigma$ and points $(t,j,\mos)$, $(t',j',\mos)$ (with the same $\mos$), both points have the same number of $R$-children whose mosaic coordinate is $\mos'$, for every $\mos'\in H^*$.
Fix $(t,j,\mos)\in W^{\M_1}$, $(t',j',\mos)\in W^{\M_2}$ and $R\in\sigma$, and let $\Rmc_\mos$ be the $R$-counting function for $\mos$ from the construction. For every $\mos'\in H^*$, the number of $R$-children of $(t,j,\mos)$ in $\M_1$ which have $\mos'$ on the mosaic coordinate is exactly $\sum_{s\in \mos'_1}\Rmc_\mos(t,1,\mos',s)$. Likewise, the number of $R$-children of $(t',j',\mos)$ in $\M_2$ with $\mos'$ on the mosaic coordinate is also exactly $\sum_{s'\in \mos'_2}\Rmc_\mos(t',2,\mos',s')$. By the fourth condition~\ref{it:R-wit 4} from the definition of an $R$-counting function these two numbers are equal. This proves the claim, and hence also Theorem~\ref{thm:correctness-elimination-graded}
\end{proof}
\end{proof}

\clconstructiongradedclaim*

\begin{proof}
The first item is proven in the same way as in the proof of Claim~\ref{cl:construction claim}. We now prove the other claim that $H'$ is a witness for every mosaic from $H$. The proof that $H'$ extends $H$ is unchanged. It remains to show that for every $R$ and $\mos\in H$, $H'$ is an $R$-witness for $\mos$. Fix $R$ and $\mos\in H$. We need to construct an $R$-counting function \Rmc. By Lemma~\ref{lem:R-counting weak}, it is enough to construct a \emph{weak} $R$-counting function, which we do now. 

Assume $i\in \{1,2\}$ and $t\in\TP$. If $t\notin \mos_i$ then $\Rmc(t,i,\mos',s)=0$ for every $\mos'$ and $s$. In the remaining case fix a model
\[
\M,w\models\case(\cmf^+_i)\wedge\delta_\mos\wedge t
\]
and denote by $V$ the set of $R$-children of the root. We set:
\[
\Rmc(t,i,\mos',s) = 
\left|
\left\{
  v\in V\ \;\middle|\;
  \begin{aligned}
  & d\in\SepCaseType,\ \mos'=(\mos_1^d,\mos_2^d),\\
  & \M,v\models s\wedge d
  \end{aligned}
\right\}
\right|
\]
for all $\mos',s$. This completes the construction of $\Rmc(t,i,\ \cdot\ ,\ \cdot\ )$. Running it for every possible values of $t,i$ we construct the entire $\Rmc$.
% \[
% \Rmc(t,i,\mos',s) =\min(\alpha,\kappa)
% \]
% where
% \[
% \alpha =
% \left|
% \left\{
%   v\in V\ \;\middle|\;
%   \begin{aligned}
%   & d\in\SepCaseType,\ \mos'=(\mos_1^d,\mos_2^d),\\
%   & \M,v\models s\wedge d
%   \end{aligned}
% \right\}
% \right|.
% \]

We show that \Rmc is a weak $R$-counting function.
Conditions~\ref{it:R-wit 1}, and~\ref{it:R-wit 2} from the definition are straighforward to check. Condition~\ref{it:R-wit 3} follows from $\M,w\models t$, because for every $s$, the number of $R$-children of $w$ which satisfy $s$ is exactly $\sum_{\mos'\in H'} \Rmc(t,i,\mos',s)$. For (the weak variant of) the last condition~\ref{it:R-wit 4} observe that if $R\in\sigma$ then $\delta_{\mos}$ specifies, exactly up to $\lambda$, the number $f_R(d)$ of $R$-children satisfying $d$, for every $d\in\SepCaseType$.
Thus, whenever $t\in\mos_i$, we have 
\[
\sum_{s\in\mos'_i}\Rmc(t,i,\mos',s)
=
\sum_{\substack{d\in\SepCaseType\\ \mos'=(\mos_1^{d},\mos_2^{d})}} |V_d|
=^\lambda
\sum_{\substack{d\in\SepCaseType\\ \mos'=(\mos_1^{d},\mos_2^{d})}} f_R(d)
\]
where $V_d=\{v\in V\ |\ \M,v\models d\}$, and $\M,V$ are the model and its subset from the definition of $\Rmc(t,i,\ \cdot\ ,\ \cdot\ )$. The second equality above follows from the fact that $\M,w\models\delta_\mos$, and hence for every $d$ we have $|V_d|=^\lambda f_R(d)$ by $\delta_\mos\models\mydiam{= f_R(d)}d$. Since this does not depend on the choice of $\M,w$ for $t,i$, it follows that for all $i,j=1,2$, $t\in \mos_i,t'\in \mos_j$ and $\mos'\in H'$:
\[
\sum_{s\in\mos'_i}\Rmc(t,i,\mos',s)
=^\lambda
\sum_{\substack{d\in\SepCaseType\\ \mos'=(\mos_1^{d},\mos_2^{d})}} f_R(d)
=^\lambda
\sum_{s'\in\mos'_j}\Rmc(t',j,\mos',s').
\]
This proves (the weak variant of) condition~\ref{it:R-wit 4}, thereby proving that \Rmc is a weak $R$-counting function, and hence completes the proof of Claim~\ref{cl:construction graded claim}.
\end{proof}

\section{Proofs for Section~\ref{sec:lowerbounds}}

\thmlowerbound*

We first give the missing arguments that $\varphi_1,\varphi_2$ described in the main part are not $\Hfull(\sigma)$-separable by a formula of small modal depth. Then we show how to achieve triply exponential lower bound for the formula size. 

\medskip
For the first point, we first modify $\varphi_1$ in a way that it only talks about points reachable via $R,S$ from the root of the model. Observe that Properties~(i) and~(ii) from the main part can be achieved with $\varphi_1$ of shape $\varphi_1=\psi_1\wedge \mybox{U}\psi_2$ for $\psi_1,\psi_2\in \Hmc$. Introduce a fresh proposition $s$ and consider $\varphi_1'$ defined as follows:
\[\varphi_1' = \psi_1\wedge s\wedge \mybox{U}\left(s\to (\psi_2\wedge \mybox{R}s\wedge \mybox{S}s)\right).\]
It should be clear that $\varphi_1'$ still satisfies Properties~(i) and~(ii) and is as required. So from now on we assume that $\varphi_1$ can only see points reachable via $R,S$ from the root. 

We now show that no $\Hfull$ formula over $\sigma$ with modal depth smaller than $2^{2^n}$ separates $\phi_1,\phi_2$. Assume towards contradiction that $\theta$ is such a formula and let $d<2^{2^n}$ be its modal depth. 

As described in the main part, we start with model 
$\N,v\models\phi_2$ consisting of a single infinite $R$-path $\pi$ with an $S$-self-loop around the first point $v$, such that consecutive segments of $\pi$ implement correct increments of $r$-values modulo $2^{2^n}$, starting from $0$. Consider its unraveling $\N^v,v$ into a tree.
It is not a good $SR$-tree, but one can turn it into a good $SR$-tree $\M^0,v$ by modifying only the valuation of $r$ in positions of depth $>d$. By Property~(ii), this $\M^0,v$ can be turned into a model $\M,v$ of $\varphi_1$ by changing non-$\sigma$ symbols. We have $\M,v\models\theta$. Since $\N,v\models \neg \theta$, also $\N^v,v\models\neg\theta$. If $\theta\in\Hmc$ this is already a contradiction since the respective reducts of $\M,v$ and $\N^v,v$ to $\sigma$ are isomorphic on the first $d$-levels, and $\theta$ has modal depth $d$.

We now derive contradiction in the general case $\theta\in\Hfull$. Let $\M^*$ be the disjoint union of $\N$ and $\M$, denote by $w_1$ the copy of $v$ from $\M$, and by $w_2$ the copy of $v$ from $\N$. It is easily verified that
$\M^*,w_1\models\phi_1$ and $\M^*,w_2\models\phi_2$; the former is a consequence of
our assumption that $\varphi_1$ only sees points reachable via $R,S$.  Consequently, $\M^*,w_1\models\theta$ and $\M^*,w_2\models\neg\theta$. Let $\theta'$ be obtained from $\theta$ by replacing every subformula $\chi$ which starts with $\mybox{U}$ or $\mydiam{U}$ by $\top$ if $\chi$ is true in $\M^*$ and by $\bot$ otherwise. By construction, $\M^*,w\models\theta'$ and $\M^*,v\models\neg\theta'$, and moreover $\theta'$ does not use $U$ and has modal depth at most $d$. But this is impossible since the reachable parts of $\M^*,w_1$ and $\M^*,w_2$ are isomorphic to $\M,v$ and $\N^v,v$, respectively, and the $\sigma$-reducts of these are isomorphic up to level $d$.

The above proves that $\phi_1,\phi_2$ require separators of doubly exponential modal depth. To achieve $\phi_1^+,\phi_2^+$ which require triply exponential formula size we use the following trick. Instead of the relation $R$ we use a union of two different $R_1,R_2$. That is, in $\phi_1$ we replace every $\mydiam{R}\zeta$ with $\mydiam{R_1}\zeta\vee\mydiam{R_2}\zeta$ and every $\mybox{R}\zeta$ with $\mybox{R_1}\zeta\wedge\mybox{R_2}\zeta$. The other formula $\phi_2^+=\phi_2$ is unchanged. Since the construction of $\phi_1$ can be carried out in such a way that its modal depth does not depend on $n$, it follows that $\phi_1^+$ is polynomial in $n$.

We claim that if $\theta$ separates $\phi_1^+,\phi_2^+$ then it embeds a full binary tree of depth $2^{2^n}$ in the following sense. For every sequence $T_1T_2...\in\{R_1,R_2\}^{2^{2^n}}$ there is a sequence $\zeta_1\zeta_2...$ of length $2^{2^n}$ of subformulae of $\theta$ where each $\zeta_i$ begins with $\mydiam{T_i}$ or $\mybox{T_i}$, and $\theta_i$ is a strict subformula of $\theta_j$ whenever $i<j$. This claim is proven with an argument very similar to the above argument for the lower bound on modal depth of separators for $\phi_1,\phi_2$. The only adaptation is that, if some sequence $\pi=T_1T_2...$ of modalities is missing in $\theta$, then the model $\N,v$ (and hence the models derived from it) use this missing $\pi$ in place of the infinite $R$-path.

Finally, note that for every $d<2^{2^n}$ our construction provides models $\M^*,w_i\models\phi_i$ for both $i\in\{1,2\}$, such that the $\sigma$-reducts of the reachable parts of these two models are isomorphic up to depth $d$. This proves that actually no $\theta\in\Gmc$ of modal depth $d<2^{2^n}$ separates $\phi_1,\phi_2$. The same trick as with $\Hfull$ shows that this is not even possible with $\theta\in\Gfull$ of modal depth $d<2^{2^n}$.
This completes the proof of Theorem~\ref{thm:size-lower-bound}.

\section{Proofs for Section~\ref{sec:uniform}}

\thmundecunform*

The proof is by reduction of the following undecidable tiling problem~\cite{EmdeBoas97}.
%
% \vspace*{-\medskipamount}
\begin{definition}
	A \emph{tiling system} $\mathcal{S}=(\mathcal{T},H,V,R,L,T,B)$ consists of a finite set
	$\Tmc$ of \emph{tiles}, horizontal and vertical \emph{matching
		relations} $H,V \subseteq \Tmc \times \Tmc$, and sets $R,L,T,B \subseteq
	\Tmc$ of \emph{right} tiles, \emph{left} tiles, \emph{top} tiles, and
	\emph{bottom} tiles. A \emph{solution} to $\mathcal{S}$ is a triple $(n,m,\tau)$
	where $n,m \geq 1$ and $\tau: \{0,\ldots,n\} \times \{0,\ldots,m\} \rightarrow \Tmc$ such
	that the following hold:
	\begin{enumerate}
		
		\item $(\tau(i,j),\tau(i+1,j)) \in H$, for all $i<n$ and $j \leq m$;
		
		\item $(\tau(i,j),\tau(i,j+1)) \in V$, for all $i\leq n$ and $j<m$;
		
		\item $\tau(0,j) \in L$ and $\tau(n,j) \in R$, for all $0 \leq j \leq m$;
		
		\item $\tau(i,0) \in B$ and $\tau(i,m) \in T$, for all $0 \leq i \leq n$.
		
	\end{enumerate}
	\vspace*{-\medskipamount}
\end{definition}
%
%We first show undecidability \emph{with} the universal modality %$\Box_{u}$.
%Thus, we can use description logic notation for now when defining $\phi$.

Let $\mathcal{S}=(\Tmc,H,V,R,L,T,B)$ be a tiling system as above. We construct a formula $\varphi_{\mathcal{S}}$ and signature $\sigma$ such
that the following are equivalent:
\begin{enumerate}
    \item $\mathcal{S}$ has no solution;
    \item there is a uniform FO$(\sigma)$-interpolant for $\varphi_{\Smc}$ in $\Hmc$;
    \item there is a uniform $\Hmc(\sigma)$-interpolant for $\varphi_{\Smc}$ in FO.
\end{enumerate}
(Observe that (2) $\Rightarrow$ (3) always holds.) To construct $\varphi_{\Smc}$, regard the tiles $t\in \Tmc$
as propositional variables and let $\sigma$ contain $\Tmc$, the symbols $S,E,\mathsf{min}$ introduced in the paper, propositional variables $\mn{left}, \mn{right}, \mn{top}, \mn{bottom}$, and binary relations $F,R_{x},R_{y}$.

$E$ will again allow us to access the nodes in a finite linear order. $F$ will play a similar role as it will 
allow us to access the nodes on a grid that provides the space for a solution to the tiling problem.
%We use $E$ and $R$ as approximations of the universal %modality. In the intended 
Let $\phi_{\Smc} =\psi_{1}\wedge \cdots \wedge \psi_{4}$ where $\psi_{1},\ldots,\psi_{4}$ are defined as follows: 
%
	
	%	\item The roles $r_x$, $r_y$, and their inverses are functional:
	%	%
	%	$$
	%	\top \sqsubseteq (\leqslant 1\ r), \text{ for } r \in
	%	\{r_x,r_y,r_x^-,r_y^-\}
	%	$$
	%	
$\psi_{1}$ says that every grid node is labeled with exactly
	one tile and the matching conditions are satisfied:
    \begin{align*}
	\mybox{F} \bigvee_{t\in \Tmc}(t \wedge \bigwedge_{t' \in \Tmc,\; t'\not=t} \neg t') \\[4mm]
	\mybox{F} \bigwedge_{t\in \Tmc}(t \rightarrow (\bigvee_{(t,t') \in H} \mybox{R_{x}}t' \wedge \bigvee_{(t,t') \in V} \mybox{R_{y}}t'))
	\end{align*}
	
$\psi_{2}$ says that propositional variables \mn{left}, \mn{right}, \mn{top}, \mn{bottom} mark the borders of the grid in the expected way:
	\begin{align*}	
		\mybox{F}(\mn{bottom} \rightarrow \neg \mn{top} \wedge\mybox{R_{x}}\mn{bottom}) \\	
		\mybox{F}(\mn{right} \rightarrow \mybox{R_{y}}\mn{right})\\
		\mybox{F}(\mn{left} \rightarrow \neg \mn{right} \wedge \mybox{R_{y}}\mn{left})\\
		\mybox{F}(\mn{top} \rightarrow \mybox{R_{x}}\mn{top})\\
	%	\mn{bottom} & \sqsubseteq & \mn{right} \sqcup \exists r_{x}. \mn{bottom}
		\mybox{F}(\neg \mn{top} \leftrightarrow\mydiam{R_{y}}\top) \\
        \mybox{F}(\neg \mn{right} \leftrightarrow\mydiam{R_{x}}\top	)\\
        \mybox{F}(\mn{bottom} \rightarrow \bigvee_{t\in B}t), \quad
        \mybox{F}(\mn{right} \rightarrow \bigvee_{t\in R}t)\\
        \mybox{F}(\mn{left} \rightarrow \bigvee_{t\in L}t), \quad
        \mybox{F}(\mn{top} \rightarrow \bigvee_{t\in T}t).
\end{align*}

$\psi_{3}=\mydiam{F}(\mn{left} \wedge \mn{bottom})$.

\medskip
$\psi_{4}$ says that either (i) there is an infinite outgoing $R_x$/$R_y$-path starting at a node in $\mn{left} \wedge \mn{bottom}$ satisfying $Q$ on the grid or (ii) there is a finite outgoing $R_x$/$R_y$-path starting at a node in $\mn{left} \wedge \mn{bottom}$ satisfying $Q$ and leading to a grid cell that does not close or (iii) the formula $\varphi_{even}$ defined in the paper holds. Hence $\psi_{4}$ is the conjunction of $\mydiam{F}(Q \wedge \mn{left} \wedge \mn{bottom})$ and
	$$	
	\mybox{F}(Q \rightarrow \varphi_{Q,P} \vee \varphi_{even})
    $$
    where 
    $$
    \varphi_{P,Q}= \mydiam{R_{x}}Q \vee \mydiam{R_{y}}Q \vee
	(\mydiam{R_{x}}\mydiam{R_{y}}P \wedge \mydiam{R_{y}}\mydiam{R_{x}}\neg P))
	$$
%Set $\sigma= \sig(\varphi)\setminus \{Q,P\}$.
The following lemma shows that (1) $\Rightarrow$ (2).
\begin{lemma}
    If $\mathcal{S}$ has no solution, then $\psi_{1}\wedge\psi_{2}\wedge\psi_{3}$ is a uniform FO$(\sigma)$-interpolant for $\phi_{\Smc}$.
\end{lemma}
\begin{proof}
If $\mathcal{S}$ has no solution, then one can show that every model $\Mmc,w$ of $\psi_{1}\wedge\psi_{2}\wedge\psi_{3}$
can be expanded to a model $\Mmc',w$ of $\phi_{\Smc}$ by either making $Q$ true in an infinite $R_{x}/R_{y}$-path within $\{d \mid (w,d)\in F^{\Mmc}\}$ starting at some $w'\in (\mn{left} \wedge \mn{bottom})^{\Mmc}$ with
$(w,w')\in F^{\Mmc}$ or making $Q$ true in a finite such $R_{x}/R_{y}$-path that leads to a node satisfying $(\mydiam{R_{x}}\mydiam{R_{y}}P \wedge \mydiam{R_{y}}\mydiam{R_{x}}\neg P)$. The claim follows immediately.
\end{proof}
The following lemma shows that (3) $\Rightarrow$ (1).
\begin{lemma}
    If $\mathcal{S}$ has a solution, then there is no uniform $\Hmc(\sigma)$-interpolant for $\varphi_{\Smc}$ in FO.
\end{lemma}
\begin{proof} \
	Assume, for a proof by contradiction, that $\mathcal{S}$ has a solution $\tau$ consisting of a properly tiled $n_{0}\times m_{0}$ grid and that there is a uniform $\Hmc(\sigma)$-interpolant $\chi(x)$ for $\varphi_{\Smc}$ in FO. We show that then there is a formula in FO distinguishing between finite linear orders of even and odd length. To this end, we first expand the model $\Mmc_{n}$ defined in the paper to a model 
    $$
    \Nmc_{n}=(V_{n},(r^{\Nmc_{n}})_{r\in\sigma})
    $$
    where $V_{n}$ is obtained by taking the nodes in $W_{n}$, adding nodes
    \begin{itemize}
        \item $n$ and $a_{i,j}$ for $0\leq i \leq n_{0}$ and $0\leq j \leq m_{0}$.
    \end{itemize}
    and identifying $a$ with $a_{0,0}$.
    Take the following interpretation of the remaining symbols: 
\begin{itemize}
	\item $a_{i+1,j}$ is the only $R_{x}$-successor of $a_{i,j}$, for $i+1\leq n_{0}$;
	\item $a_{i,j+1}$ is the only $R_{y}$-successor of $a_{i,j}$, for $j+1\leq m_{0}$;
	\item the borders of the grid satisfy the respective propositional variables $\mn{left}, \mn{right}$, $\mn{top}$, $\mn{bottom}$;
	%\item $a$ is identified with $a_{0,0}$.
    \item $t^{\Nmc_{n}}=\{a_{i,j} \mid \tau(i,j)=t\}$, for $t\in \Tmc$;
    \item $(n,a_{i,j})\in F^{\Nmc_{n}}$ for all $i\leq j_{0}$ and $j\leq m_{0}$.
\end{itemize}
Then $\Nmc_{n}$ can be expanded to a model satisfying $\varphi_{\Smc}$ iff $n$ is even.

Now we describe $\Nmc_{n}$ using a formula $\chi_{n}'$ in $\Hmc$ by regarding $a_{i,j}$ as nominals and taking
the conjunction of the following formulae (with $\chi_{n}$ defined in the paper):
\begin{align*}
     \mydiam{F}\chi_{n} \wedge \bigwedge_{i,j}\mydiam{F}a_{i,j}\\
     \mybox{F}\bigvee_{i,j}a_{i,j} \wedge \mybox{F}\bigwedge_{i,j\not=i',j'}\neg (a_{i,j} \wedge a_{i',j'})\\
     \mybox{F}\bigwedge_{i,j}(a_{i,j}\rightarrow \mydiam{R_{x}}a_{i+1,j}), \quad \mybox{F}\bigwedge_{i,j}(a_{i,j} \rightarrow \mybox{R_{x}}a_{i+1,j})\\
     \mybox{F}\bigwedge_{i,j}(a_{i,j}\rightarrow \mydiam{R_{y}}a_{i,j+1}), \quad \mybox{F}\bigwedge_{i,j}(a_{i,j} \rightarrow \mybox{R_{y}}a_{i,j+1})\\
     \mybox{F}\bigwedge_{j}(a_{n_{0}},j)\rightarrow \mybox{R_{x}}\bot)\\
     \mybox{F}\bigwedge_{i}(a_{i,m_{0}})\rightarrow \mybox{R_{y}}\bot)\\
     \mydiam{F}(a_{i,j}\wedge t) \text{ if }\tau(i,j)= t    \end{align*}
Then $\phi_{\Smc}\models \neg \chi_{n}'$ iff $n$ is an even number. Hence a uniform $\Hmc(\sigma)$-interpolant $\chi$ for $\varphi_{\Smc}$ in FO would have the property that $\chi\models \neg \chi_{n}$ iff $n$ is an even number. But this contradicts non-FO-definability of finite linear orders of even length in FO. 
\end{proof}

\end{document}